\documentclass[11pt,oneandhalf,chaparabic]{metu}

\usepackage{graphicx}
\usepackage{algorithm}
\usepackage{algorithmic}
\usepackage{amsmath}
\usepackage{txfonts}
\usepackage{bbm}
\usepackage{amssymb}
\usepackage{colortbl}

\usepackage{hhline}
\newtheorem{theorem}{Theorem}[chapter]
\newtheorem{definition}{Definition}
\newtheorem{lemma}[theorem]{Lemma}
\newtheorem{corollary}[theorem]{Corollary}
\newtheorem{example}{Example}[chapter]

\newenvironment{defn*}{\begin{definition}}{\end{definition}}
\newenvironment{proof}{\noindent{\bf Proof.}}{\hfill$\blacksquare$}

\newcommand{\ue}{\"{u}}
\newcommand{\Ue}{\"{U}}
\newcommand{\qe}{\"{o}}
\newcommand{\Oe}{\"{O}}
\newcommand{\gh}{\~{g}}
\newcommand{\Ie}{\.{I}}
\newcommand{\sh}{\c{s}}

\newcommand{\ch}{\c{c}}
\newcommand{\Ch}{\c{C}}

\newcommand{\set}[1]{\mathcal{#1}}

\newcommand{\GF}{\mathrm{GF}}
\newcommand{\vect}[1]{\mathbf{#1}}
\newcommand{\vectex}[2]{\mathbf{#1}_{\setminus \left\{#2\right\}}}

\newcommand{\norm}[1]{\left\| #1 \right\|}
\newcommand{\setsize}[1]{|\set{#1}|}

\newcommand{\abs}[1]{\left|{#1}\right|}
\newcommand{\innerproduct}[2]{<#1,#2>}

\newcommand{\ssum}{\sideset{^\boxplus}{}\sum}

\newcommand{\rfield}{\mathbb{R}}

\newcommand{\rfieldqn}{\mathbb{R}^{(q^N)}}

\newcommand{\field}[1]{\mathbb{F}_{#1}}
\newcommand{\fieldq}{\field{q}}
\newcommand{\fieldn}[1]{\mathbb{F}_{#1}^{N}}
\newcommand{\fieldqn}{\fieldn{q}}
\newcommand{\fieldnn}[2]{\field{#1}^{#2}}

\newcommand{\fieldql}{\field{q}^L}

\newcommand{\setfield}[1]{\set{P}_{\field{#1}}}
\newcommand{\setfieldq}{\set{P}_{\fieldq}}
\newcommand{\setfieldn}[1]{\set{P}_{\fieldn{#1}}}
\newcommand{\setfieldqn}{\set{P}_{\fieldqn}}
\newcommand{\setfieldnn}[2]{\set{P}_{\fieldnn{#1}{#2}}}

\newcommand{\n}[2]{\set{C}_{#1}\left\{ #2 \right\}}
\newcommand{\evalat}[1]{\Big|_{#1}}
\newcommand{\Evalat}[1]{\Bigg|_{#1}}

\newcommand{\image}[1]{\mathrm{im}\left\{#1\right\}}
\newcommand{\kernel}[1]{\mathrm{ker}\left\{#1\right\}}
\newcommand{\rank}[1]{\mathrm{rank}\left(#1\right)}
\newcommand{\expectation}[1]{\mathbf{E}\left[#1\right]}

\newcommand{\indicator}[2]{\mathbbm{1}_{\set{#1}}\left(#2\right)}

\newcommand{\operator}[2]{\mathcal{#1}\left\{#2\right\}}
\newcommand{\pseudoinv}[2]{\mathcal{#1}^{+}\left\{#2\right\}(x)}
\newcommand{\operatorn}[2]{\mathcal{#1}_{N}\left\{#2\right\}}
\newcommand{\pseudoinvn}[2]{\mathcal{#1}_{N}^{+}\left\{#2\right\}(\vect{x})}

\newcommand{\pseudoinvarg}[3]{\mathcal{#1}^{+}\left\{#2\right\}(#3)}
\newcommand{\pseudoinvi}[2]{\pseudoinvarg{#1}{#2}{x_i}}

\newcommand{\summary}[1]{\sim\{#1\}}
\newcommand{\sopsub}[2]{\mathcal{S}_{#1}\left\{#2\right\}}
\newcommand{\imagesop}[1]{\image{\mathcal{S}_{#1}}}

\newcommand{\complex}{\mathbb{C}}
\newcommand{\real}{\mathbb{R}}
\renewcommand{\Re}[1]{\mathrm{Re}\left\{ #1 \right\}}
\renewcommand{\Im}[1]{\mathrm{Im}\left\{ #1 \right\}}

\newcommand{\psk}[1]{\mu_q\left(#1\right)}
\newcommand{\bpsk}[1]{\beta\left(#1\right)}
\newcommand{\pam}[2]{\eta_{#1}\left(#2\right)}
\newcommand{\gpam}[2]{\kappa_{#1}\left(#2\right)}
\newcommand{\qpsk}[1]{\nu\left(#1\right)}
\newcommand{\integer}[1]{\textrm{int}(#1)}

\newcommand{\nokta}{\textrm{.}}
\newcommand{\virgul}{\textrm{,}}

\author{Muhammet Fatih Bayramo\~{g}lu}
\authoruppercase{MUHAMMET FAT\.{I}H BAYRAMO\~{G}LU}  
\graduateschool{Natural and Applied Sciences}
\shortdegree{Ph.D.}
\director{Prof. Dr. Canan \"{O}zgen}
\headofdept{Prof. Dr. \.{I}smet Erkmen}

\title{The Hilbert Space of Probability Mass Functions\\ and Applications on Probabilistic Inference}
\degree{Doctor of Philosophy}
\department{Electrical and Electronics Engineering}
\date{September 2011}
\supervisor{Assoc. Prof. Dr. Ali \"{O}zg\"{u}r Y\i lmaz}
\departmentofsupervisor{Electrical and Electronics Engineering Dept., METU}


\turkishtitle{Olasilik K\Ue tles\Ie{}  Fonks\Ie yonlarInIn H\Ie lbert UzayI \\ve OlasIlIsIksal B\Ie lg\Ie{}  \Ch IkarImI \Ue zer\Ie ne UygulamalarI}
\turkishdegree{Doktora}
\turkishdepartment{Elektrik Elektronik M\ue hendisli\gh i B\"ol\"um\"u}
\turkishdate{Eyl\"ul 2011}
\turkishsupervisor{Do\ch . Dr. Ali \Oe zg\ue r Y\i lmaz}
\anahtarklm{Olas\i l\i k k\ue tlesi fonksiyonlar\i n\i n Hilbert uzay\i{}, Olas\i l\i k k\ue tlesi fonksiyonlar\i n\i n \ch arpanlara ayr\i lmas\i{}, 
  olas\i l\i s\i ksal bilgi  \ch \i kar\i m\i{}, \ch ok-girdili
  \ch ok-\ch\i kt\i l\i{} sezim, Markov rastgele alanlar}

\committeememberi{Prof. Dr. Yal\ch \i n Tan\i k}
\committeememberii{Assoc. Prof. Dr. Ali \Oe zg\ue r Y\i lmaz}
\committeememberiii{Prof. Dr. Mustafa Kuzuo\gh lu}
\committeememberiv{Assoc. Prof. Dr. Emre Akta\sh}
\committeememberv{Assist. Prof. Dr. Ay\c{s}e Melda Y\"uksel}
\affiliationi{Electrical and Electronics Eng. Dept., METU}
\affiliationii{Electrical and Electronics Eng. Dept., METU}
\affiliationiii{Electrical and Electronics Eng. Dept., METU}
\affiliationiv{Electrical and Electronics Eng. Dept., Hacettepe University}
\affiliationv{Electrical and Electronics Eng. Dept., \\TOBB University of Economy and Technology}

\begin{document}

\begin{preliminaries}
  \maketitle
  \makeapproval
  \plagiarism
  \setlength{\parindent}{0em}
  \setlength{\parskip}{10pt}

  \begin{abstract}\oneandhalfspacing
  The Hilbert space of probability mass functions (pmf) is introduced  in this thesis.
  A factorization method for multivariate pmfs is proposed by using the tools provided 
  by the Hilbert space of pmfs. The resulting factorization is special for two 
  reasons. First, it reveals the algebraic relations between the involved random variables. 
  Second, it determines the conditional independence relations between the random variables. 
  Due to the first property of the resulting factorization, it can be shown that channel decoders
  can be employed in the solution of  probabilistic inference problems other than decoding. This approach might lead
  to  new probabilistic inference algorithms and new hardware options for the implementation of these algorithms. 
  An example of new inference algorithms inspired by the idea of using channel decoder
  for other inference tasks is a multiple-input multiple-output (MIMO) detection algorithm which has a complexity of the 
  square-root of the optimum MIMO detection algorithm.

  \end{abstract}
  \begin{oz}\oneandhalfspacing

  Bu tezde olas\i l\i k k\ue tlesi fonksiyonlar\i n\i n Hilbert uzay\i{}  sunulmaktad\i r. 
  Bu Hilbert uzay\i n\i n sa\gh lad\i \gh \i{} olanaklar kullan\i larak \ch ok de\gh i\sh kenli 
  olas\i l\i k k\ue tlesi fonksiyonlar\i n\i{} \ch arpanlar\i na ay\i rmak  i\ch in bir
  y\qe ntem \qe nerilmi\sh tir. Bu y\qe ntemden elde edilen \ch arpanlara ay\i rma iki nedenle 
  \qe zeldir. \Ie lk olarak, bu \ch arpanlara ay\i rma rastgele de\gh i\sh kenler aras\i ndaki 
  cebirsel ba\gh \i nt\i lar\i{} ortaya koyar. \Ie kinci olarak, rastgele de\gh i\sh kenler aras\i ndaki
  ko\sh ullu ba\gh \i ms\i zl\i k ili\sh kilerini belirler. Birinci \qe zellik sayesinde 
  kanal kod \ch\qe z\ue c\ue lerinin, kod \ch\qe zmekten ba\sh ka olas\i l\i ksal bilgi \ch\i kar\i m\i{}
  problemlerinin  \ch\qe z\ue m\ue nde  de kullan\i labilece\gh i g\qe sterilebilir. Bu yakla\sh\i m
  yeni olas\i l\i ksal bilgi \ch\i kar\i m\i{} algoritmalar\i na ve bu algoritmalar\i{} ger\ch eklemek
  i\ch in yeni donan\i m olanaklar\i na yol a\ch abilir. Kod \ch\qe z\ue c\ue lerin kod \ch\qe zmekten ba\sh ka
  bilgi \ch\i kar\i m\i{} g\qe revlerinde kullan\i lmas\i{} fikrinden esinlenen algor\i tmalar\i n bir 
  \qe rne\gh i, karma\sh\i kl\i \gh \i{} en  iyi algoritman\i n karek\qe k\ue{} olan bir \ch ok-girdili
  \ch ok-\ch\i kt\i l\i{} sezim algoritmas\i d\i r. 


  \end{oz}
  \dedication{\flushright{\textit{Kar\i ma\\ To my wife}}}
  \setlength{\parindent}{0em}
  \setlength{\parskip}{10pt}
  \begin{acknowledgments}\oneandhalfspacing
    Firstly, I would like to thank sincerely my supervisor Assoc. Prof. Dr. Ali \"{O}zg\"{u}r Y\i lmaz. 
    He is an exception in this department regarding both his scientific vision and his 
    personality. His trust and encouragement was crucial to me while working on this thesis.

    I would like to thank the members of the thesis progress monitoring committee members
    Prof. Dr. Mustafa Kuzuo\~{g}lu and Assoc. Prof. Dr. Emre Akta\c{s} for their
    valuable comments and contributions. Moreover, operator theory course of Prof. Kuzuo\~{g}lu
    helped me a lot in this thesis. Furthermore, I appreciate the financial support
    that Assoc. Prof. Akta\c{s} provided to me in the last periods of my thesis study 
    and his understanding. I would like to  thank also the rest of thesis jury members 
    Prof. Dr. Yal\c{c}\i n Tan\i k and Assist. Prof. Dr. Melda Y\"{u}ksel for their 
    valuable comments. 

    I would like express my gratitude to the other two exceptional faculty members of 
    this department who are Prof. Dr. Arif Erta\c{s}  and Assoc. Prof. Dr. \c{C}a\~{g}atay
    Candan.  Prof. Erta\c{s} is a person really deserving his name ``Arif''.  Assoc. Prof.
    Candan is one of the easiest persons that I can communicate with and I appreciate
    his ``always open'' door. 

    I would like to  thank Prof. Dr. Zafer \"{U}nver. I learned  a lot from him while
    assisting EE213 and EE214 courses. 

    Education is a long haul run. Hence, sincere thanks go to my high school mathematics teachers 
    Ahmet Cengiz, H\"{u}seyin \c{C}ak\i r, Perihan \"{O}zdingi\c{s}, and of course Demir Demirhas. 
    A special thanks goes to my undergraduate advisor Prof. Dr. G\"{o}n\"{u}l
    Turhan Sayan. 

    I would like to acknowledge the free software community. I have never needed and used any commercial software 
    during my Ph.D. research.

    I would like to thank Dr. Jorge Cham for phdcomics which introduced some smiles to our overly stressed lives. 
    
    I would like to thank my friends Alper S\"{o}yler, Murat K\i{}l\i{}\c{c}, Mehmet Akif Antepli, Murat \"{U}ney,
    Y\i lmaz Kalkan, Serdar Gedik, and Onur \"Oze\c{c} for their valuable friendship. 

    My sincere gratitude goes to my parents Nezahat and Mustafa Bayramo\~{g}lu. This thesis could not finish without
    their prayers. But the good manners I learned from them is much more valuable to me than this Ph.D. degree. 
    I would like to thank my brother Etka for being the kindest brother in the world. 
    
    My deepest thanks goes to my wife Neslihan, or more precisely Dr. Neslihan Yal\c{c}\i{}n Bay-ramo\~{g}lu. I appreciate
    everything she sacrificed for me. I studied on this thesis on times that I stoled from her and she really deserves
    at least the half of the credit for this thesis. Her support not only was vital 
    for me during the Ph.D. but also will continue to be vital during the rest of my life.

  \end{acknowledgments}

  \begin{preface}    
    This thesis summarizes the research work carried out in six years starting from September 2005. 
    The research topic arose while I was trying to develop an analysis method for the convergence 
    rate of the iterative sum-product algorithm. Since the messages (beliefs)  passed between the nodes in the iterative
    sum-product algorithm are probability mass functions (pmf), I thought that representing 
    the pmfs in a Hilbert space structure would prove useful in the analysis of the sum-product 
    algorithm. Analyzing the convergence of the sum-product algorithm would be an application of
    the norm in the Hilbert space of pmfs. However, later I noticed that the inner product has 
    much more interesting applications and preferred  focusing on the applications of the inner product
    to dealing with the convergence which led to this thesis. 

    In order to read the thesis a basic understanding of inner product spaces and finite fields 
    is necessary. Anybody with this background can follow the chapters from the second to the fifth. 
    I believe that these chapters are the core of the thesis. Chapter 6  contains some 
    applications from communication theory and might require a communication theory background.

    This preface is an adequate place to note  some  observations about my country and university.
    I am happy to observe that  Turkey improved economically and democratically during 
    my graduate studies. On the other hand, I am sad to observe that Middle East Technical University downgraded 
    scientifically and democratically during the same time.

    This thesis is related probability theory. Probability theory is an area  which is close 
    to the border between science and belief. Although Laplace's book on celestial mechanics misses 
    to mention  God,  my explanation on the relation between probability and willpower makes me to 
    believe in God and  I would like to start to the rest of the thesis by a quote from the translation
    of  Qur'an which explains what is science to me: ``Glory be to You, we have no knowledge 
    except what you have taught us. Verily, it is You (Allah), the All-Knower, the All-Wise''.

 \end{preface}

  \setlength{\parindent}{0em}
  \setlength{\parskip}{3pt}
  \tableofcontents


\listoffigures

\end{preliminaries}

\setlength{\parindent}{0em}
\setlength{\parskip}{10pt}

\chapter{INTRODUCTION}

\section{Motivation}

A linear vector space structure over a set provides algebraic tools
such as addition and scaling to carry out on the elements of the set. 
If a vector space can be endowed with an inner product 
then it becomes an inner product space. An inner product provides
geometric concepts such as norm, distance, angle, and projections. 
If every Cauchy sequence in an inner product space converges
with respect to the inner product induced norm then the inner product
space becomes a Hilbert space. Needless to say a Hilbert space
structure is very useful and find application areas  in diverse
fields of science. Communication theory is not an exception. For instance, 
the signal space representation in communication theory 
relies on the Hilbert space structure constructed over the set of square 
integrable functions.

One of the mathematical objects that is too frequently 
used in communication and information theories is the 
probability mass functions (pmf)   which are discrete
equivalents of probability density functions. Although, pmfs
are so frequently used in  communication and information
theories a Hilbert space structure for them was missing. 
A Hilbert space of pmfs might have many interesting applications.

A possible application  for the Hilbert space
of probability mass functions might be analyzing
the characteristics of a multivariate pmf. An important
characteristic of a multivariate pmf is the conditional independence
relations imposed by it. The conditional independence 
relation imposed by a multivariate pmf is determined 
by the factorization of the pmf to local functions\footnote{Local 
functions are functions (not necessarily pmfs) 
which have less arguments than the original multivariate pmf.} as explained in \cite{hct1,hct2}.

The factorization structure of a multivariate pmf into local functions 
also determines the algorithms which can perform inference on 
the pmf, in other words, maximize or marginalize the pmf.
The sum-product algorithm, which is also called belief propagation, 
and the max-product algorithm effectively marginalize or maximize
a multivariate pmf by exploiting the pmfs' factorization 
structure \cite{fgsp}. Modern decoding algorithms
such as low-density parity-check decoding and turbo decoding,
which have become highly popular in the last decade, relies on this fact. 

Some multivariate pmfs, for instance the pmf resulting from a hidden Markov
model, has an apparent factorization structure. However, 
one cannot be sure whether this factorization structure is the 
``best'' possible factorization or not. On the other hand, some pmfs, 
for instance the pmfs obtained empirically, might 
not have an apparent factorization structure at all.  Therefore, 
developing a method which obtains the factorization of a
multivariate pmf systematically would  prove useful in many areas.



\section{Contributions}

The first contribution in this thesis is the derivation of the 
Hilbert space structure for pmfs. The Hilbert space
of pmfs not only provides a vectorial representation of 
evidence but also it proves to be a useful tool in analyzing
the pmfs. 

The second contribution of this thesis is a systematic
method for obtaining factorization of a multivariate pmf. 
The resulting factorization is unique and is the ultimate
factorization possible. Hence, we call the 
resulting factorization as the canonical factorization. 
The canonical factorization of a multivariate pmf is obtained 
by projecting the pmf onto orthogonal basis pmfs 
of the Hilbert space of pmfs. Hence, this factorization 
method heavily relies on the Hilbert space of pmfs. 

The basis pmfs  mentioned in the paragraph
above are special pmfs such that their value is determined only by a  linear
combination of their arguments. In order to be able to talk 
about linear combinations of arguments addition and multiplication
must be well defined between arguments of the pmf. Hence, the canonical
factorization of a pmf can be obtained only if the pmf is a pmf
of finite-field-valued random variables. This is an important limitation of 
the canonical factorization.

The property of the basis pmfs mentioned in the previous paragraph 
causes an important limitation but also this property leads 
to the third and the probably the  most important 
contribution of the thesis. Since the basis pmfs are functions of their 
arguments, the canonical factorization reveals the algebraic
dependencies between the random variables. Thanks to this fact,
 it can be shown that channel decoders 
can be employed as an apparatus for tasks beyond decoding. 
This idea leads to new hardware options as well as  new inference
algorithms. 

The fourth contribution of the thesis is an application of 
the idea explained in the paragraph above. This contribution is 
a multiple-input multiple-ouput
(MIMO) detection algorithm which employs the decoder of a tail biting 
convolutional code as a processing device.
 This algorithm is an approximate soft-input soft-output 
MIMO detection algorithm whose  complexity is the square-root
of that of the optimum MIMO detection algorithm. 

The final contribution of the thesis is another property 
of the canonical factorization. It can be shown that 
the conditional dependence relationships 
imposed by a multivariate pmf can be determined from the 
canonical factorization of the pmf. In other words, 
the conditional independence relationships 
imposed by a pmf can be determined by using the geometric
tools provided by the Hilbert space of pmfs.   
This property of the canonical factorization might lead
to applications in experimental fields such as bioinformatics dealing with large 
amounts of data.



\section{Comparison to earlier work}

A Hilbert space of probability density functions   is first presented in literature in a very different area of science, stochastic
geology, in \cite{firstHilbert}. 
Their derivation is for a class of continuous probability density functions. On the other hand our
derivation is for  pmfs. Although, the resulting 
Hilbert space structures in both their and our derivations are quite similar, 
 our derivation is independent of theirs. Furthermore, we provide many applications of the 
Hilbert space of pmfs on probabilistic inference. 

The canonical factorization proposed in this thesis can be compared to the factorization of pmfs
provided by the Hammersley-Clifford theorem \cite{hct1,hct2}. Both the Hammersley-Clifford 
theorem and the canonical factorization can completely determine the conditional independence
relationships imposed by a pmf. But Hammersley-Clifford theorem does not highlight 
the algebraic dependence relationships between random variables while the canonical factorization 
does. Moreover, the canonical factorization is unique whereas the factorization  of the Hammersley-Clifford
theorem is not.

The results obtained in this thesis can be located in the factor graph literature as follows.
Factor graphs are bipartite graphical models which represent the factorization of a pmf \cite{fgsp}. 
The bipartite graphs were first employed by Tanner to describe low complexity codes in \cite{tanner}. 
A very crucial step in  achieving the factor graph representation is the Ph.D. thesis of Wiberg  \cite{Wiberg,wiberg}.
In his thesis Wiberg  showed the connection between various codes and decoding algorithms by  
introducing hidden state nodes to the graphs described by Tanner and 
characterized the message passing algorithms running on these graphs. 
Local constraints in \cite{Wiberg} are  behavioral constraints, such as parity check constraints. 
The factor graphs are the generalization of the graphical models introduced in \cite{Wiberg}
by allowing local constraints to be arbitrary  functions rather than behavioral constraints \cite{fgsp}.

The canonical factorization proposed in this thesis can also be represented by a factor graph. 
Moreover, the factor functions appearing in the canonical factorization can be transformed
into usual parity check constraints by introducing some auxiliary variables. Therefore, the
factor graph representing the canonical factorization can be transformed 
into a Tanner graph by introducing some auxiliary variable nodes which are very different from 
the hidden state nodes introduced in \cite{Wiberg}. This is essentially
an explanation of the claim that the channel decoders can be employed for inference tasks 
beyond decoding. 


\section{Outline}

After this chapter, the  thesis continues with the introduction of the Hilbert space 
of pmfs in Chapter \ref{hilbertspacechapter}. The Hilbert space 
of pmfs is the main tool to be used throughout the thesis.  The canonical factorization 
is introduced in Chapter \ref{thechapter}. Chapter \ref{specialcasechapter} investigates
the properties and special cases of the canonical factorization. Chapter \ref{decodingchapter}
explains how a channel decoder can be used for other probabilistic inference 
tasks other than its own purpose. This explanation is based on the canonical factorization. 
Some possible consequences of this result are also explained in  Chapter \ref{decodingchapter}.
Chapter \ref{applicationchapter} provides some basic examples from communication theory 
on the use of channel decoders for other inference tasks beyond decoding.  The  MIMO detector
which uses the decoder of a tail biting convolutional code is also introduced in this chapter.
Chapter \ref{markovchapter} shows that the conditional independence relations can be completely 
determined from the canonical factorization. The thesis is concluded with some possible future
directions in Chapter \ref{conclusionchapter}.  For the sake of neatness of the thesis 
some proofs and derivations are collected in the Appendix.  

\section{Some remarks on notation}

Throughout the thesis we denote the deterministic variables with lowercase  letters
and random variables with uppercase letters. 
We represent functions of multiple variables as functions of  
vectors and denote vectors with boldface letters. Lowercase boldface letters
denote deterministic vectors and capital boldface letters denote random variables. 
All vectors encountered in the thesis are row vectors except a few cases in Chapter \ref{applicationchapter}.
 
Matrices are also denoted with capital boldface letters which might lead to a confusion with 
random vectors. Throughout the thesis, we used $\vect{V}$, $\vect{W}$, $\vect{X}$, $\vect{Y}$, and $\vect{Z}$ 
to denote random vectors. All the other capital boldface letters are matrices.

Unfortunately, many different types of additions are included in the thesis such as finite field addition,
real number addition, vector addition, and even direct sum of subspaces. 
We reserve $\oplus$ symbol for the direct sum of subspaces for the sake
of consistency with the linear algebra literature.  We use $\boxplus$ symbol for the vectorial addition
operation of pmfs which is defined in Chapter \ref{hilbertspacechapter}. We have to use the remaining $+$
symbol for all the rest of addition operations such as real number addition, finite field addition,
and vectorial addition in $\mathbb{R}^{N}$. Fortunately, the type of the addition employed can be 
determined from the types of the operands. 

A possible confusion might arise while using the summation symbol $\sum$. For instance, 
$\sum_{i=1}^{N}p_{i}(x)$ might refer to both $p_{1}(x) +p_{2}(x)+\ldots +p_{N}(x)$ 
and  $p_{1}(x) \boxplus p_{2}(x) \boxplus \ldots \boxplus p_{N}(x)$ which are really two different
summations. In order to avoid this confusion we  denote the latter summation with $\ssum_{i=1}^{N}p_i(x)$,
although summations like the former is never encountered in the thesis.

\chapter{THE HILBERT SPACE OF PROBABILITY MASS FUNCTIONS\label{hilbertspacechapter}}

\section{Introduction}

The Hilbert space of probability mass functions (pmf), which is the main tool to be employed in the thesis, 
is introduced in this chapter. Throughout the thesis we are only interested in the 
pmfs of the finite-field-valued random variables. Therefore, we  define what a finite-field-valued
random variable is first in Section \ref{ffvrv}. We  introduce the set of pmfs on which we construct the 
Hilbert space in Section \ref{theset}. Then we construct the algebraic and geometric structures 
over this set in Section \ref{algebraicstructure} and Section \ref{geometricstructure} respectively. 
Section \ref{relhsrv} emphasizes the differences between the Hilbert space of random variables
and the Hilbert space of pmfs in order to avoid possible confusion. Finally, in Section \ref{mvpmf}
the idea of the construction of the Hilbert space is repeated on the set of multivariate pmfs. 
 
\section{ Finite-Field-Valued Random Variables \label{ffvrv}}

Traditionally a random variable is a mapping from the event space to the real or complex fields. However, 
in some experiments, e.g., the experiments with discrete event spaces, 
it might be useful to map the outcomes of the experiment to a finite (Galois) field.
Such a mapping would allow to carry out \emph{meaningful} algebraic operations between the outcomes
of different experiments, for instance as in \cite{massey}.  
A finite-field-valued random variable is defined below.  

\begin{definition}\emph{Finite-field-valued random variable:} Let $\Omega$ be the event space of an
experiment and $\fieldq=\GF(q)$ be the finite field of $q$ elements. Moreover, let 
a function $X:\Omega \rightarrow \fieldq$ be defined as
\begin{equation}
   X( \omega\in \set{E}_{i})\triangleq i\quad \forall i \in \fieldq \textrm{,} \nonumber
\end{equation}
where   $\{\set{E}_{i}:i\in \fieldq\}$ are  events (subsets of $\Omega$) of  this experiment. 
The function $X$ is called an $\fieldq$-valued random variable if the events 
$\{\set{E}_{i}:i\in \fieldq\}$ are mutually exclusive and collectively exhaustive, i.e.,
\begin{eqnarray}
  \set{E}_{i}\neq \set{E}_{j}&\implies& \set{E}_{i} \cap \set{E}_{j} = \emptyset  
  \quad \forall i,j \in \fieldq \nonumber  \textrm{,}\\
  \bigcup_{i \in \fieldq}\set{E}_i&=& \Omega \nonumber \textrm{.}
\end{eqnarray} 
\end{definition}

Actually, we do not need to restrict ourselves to the finite-field-valued random variables
in this chapter since the ideas presented in this chapter can be applied to any discrete random variable. 
We need the concept of finite-field-valued random variables starting from the next chapter. 
However, we introduce the finite-field-valued random variables starting from this 
chapter in order to make the representation simpler.  

\section{ The Set of Strictly Positive Probability Mass Functions \label{theset}}

Many different experiments can be represented with an $\fieldq$-valued random variable. All these
experiments may lead to different pmfs. Furthermore, we may have different pmfs even for the
same experiment if the outcome is conditioned on some other event. Let $\setfieldq$
be the set of all \emph{strictly positive} pmfs that an $\fieldq$-valued random variable
might possess, i.e.,  
\begin{equation}
\setfieldq \triangleq \left\{p(x):\fieldq \rightarrow (0,1)\subset \mathbb{R} \textrm{ s.t. }
\sum_{x\in \fieldq}p(x)=1 \right\}  \label{setdefinition} \textrm{.}
\end{equation}
The Hilbert space of pmfs is going to be constructed on $\setfieldq$. This set excludes the pmfs
which take value zero for some values. The reason under this restriction will be clear
after  scalar multiplication is defined on this set. 

We are going to represent the pmfs with lowercase letters such as $p(x)$, $r(x)$, or $s(x)$. These pmfs
may represent the pmfs of random variables representing different experiments as well as they may represent
the pmfs of the same random variable conditioned on different events. 

\subsection{The normalization operator}

We employ a normalization operator to obtain pmfs from strictly positive real-valued functions by scaling
them. We denote this normalization operator with $\n{\fieldq}{.}$ and define it as
\begin{equation}
  \n{\fieldq}{\alpha(x)}:\set{F}_{\fieldq} \rightarrow \setfieldq \triangleq \frac{\alpha(x)}{\sum_{i\in\fieldq} \alpha(i)} \textrm{,}
\end{equation}
where the set $\set{F}_{\fieldq}$ denotes the set of all functions from $\fieldq$ to $\mathbb{R}^{+}$ and 
$\alpha(x)$ is a function in $\set{F}_{\fieldq}$. An obvious property of the operator 
   $\n{\fieldq}{.}$ that we exploit frequently is given below 
\begin{equation}
  \n{\fieldq}{\beta \alpha(x)} =\n{\fieldq}{ \alpha(x)}\textrm{,} 
\end{equation}
where $\beta$ is  any positive number. 

\section{The Algebraic Structure over $\setfieldq$ \label{algebraicstructure}}
The foundation of the Hilbert space of PMFs is the 
addition operation. Hence, the definition of the 
addition should be meaningful in the  sense of 
probabilistic inference in order to take advantage of the Hilbert space structure
for inference problems. 

The addition operation is inspired by the following 
scenario. Assume that we receive information about a uniformly distributed 
source $X$ via two \emph{independent} channels 
with outputs $y_1$ and $y_2$ as depicted in Figure \ref{additionmeaning}.
Let $p(x)=\Pr\{X=x|y_1\}$, $q(x)=\Pr\{X=x|y_2\}$, and $r(x)=\Pr\{X=x|y_1,y_2\}$.
Since $X$ is uniformly distributed,  $r(x)$ can be derived as
\begin{eqnarray}
  r(x)&=&\frac{p(x)q(x)}{\sum_{x \in \fieldq}p(x)q(x)}   \label{additioneqmean}\\
  &=&\n{\fieldq}{p(x)q(x)}
\end{eqnarray}  
by employing the Bayes' theorem. 
 
The PMFs $p(x)$ and $q(x)$ represent the evidence
 about the source $X$ when only $y_1$ or $y_2$ is known 
respectively. On the other hand, $r(x)$ represents the total
evidence when both outputs are known. In a way, $r(x)$
is obtained by summing $p(x)$ and $q(x)$. Hence, (\ref{additioneqmean})
can be adopted as the definition of addition. For any $p(x)$ and $q(x)$
in $\setfieldq$ their addition is denoted by $\boxplus$ and defined
as
\begin{equation}
  p(x)\boxplus q(x) \triangleq \n{\fieldq}{ p(x)q(x)} \textrm{.}
 \end{equation}

The definition of the addition operation is such a critical point of this
thesis that the rest of the thesis will be built upon this definition. 
 
This definition of addition operation is the same as parallel information combining operation as 
defined in \cite{informationcomb} and message computation at variable nodes 
in the sum-product algorithm \cite{fgsp}.

\begin{figure}
\begin{center}
\includegraphics[scale=.7]{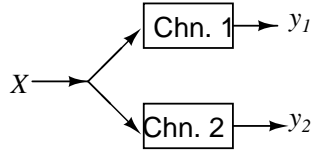}
\caption{The scenario for explaining the meaning of addition operation. \label{additionmeaning}}
\end{center}
\end{figure}

Defining the addition operation also enforces the scalar multiplication 
to have such a form that scalar multiplication is consistent with the addition.
The scalar multiplication, which is denoted
by $\boxtimes$, should satisfy the relation below for positive integers $n$
\begin{eqnarray}
  n\boxtimes p(x)&=&\underbrace{p(x)\boxplus p(x)\boxplus \ldots \boxplus p(x) }_{\textrm{n times}} \nonumber \\
  &=&\n{\fieldq}{(p(x))^n} \label{consistency} \textrm{.}
\end{eqnarray}
Generalizing (\ref{consistency}) to any $\alpha$ in $\mathbb{R}$ leads to the 
definition of scalar multiplication below

\begin{equation}
  \alpha \boxtimes p(x)\triangleq \n{\fieldq}{(p(x))^{\alpha}} \textrm{.} \label{multiplicationdef}
\end{equation}

In order to be able to scale $p(x)$ with 
negative coefficients it is necessary that $p(x) \neq 0$ for any $x$ in $\fieldq$. 
Hence, in the definition of $\setfieldq$ an open interval is used
rather than a closed interval in (\ref{setdefinition}).

\begin{theorem}
  The set $\setfieldq$ together with operations $\boxplus$ and $\boxtimes$
forms a linear vector space over $\mathbb{R}$. 
\end{theorem}
\begin{proof}
  The closure of $\setfieldq$ under  both operations 
is ensured by the normalization operators in their definitions. The commutativity 
and associativity are obvious from the definition of $\boxplus$ operation. 
The neutral element with respect to (w.r.t.) the addition operation is the
uniform distribution given by 
\begin{equation}
\theta(x)=\frac{1}{q}  \nonumber \textrm{.}
\end{equation}
 Consequently,  the additive inverse of
$p(x)$, which is denoted by $\boxminus p(x)$, is 
\begin{equation}
  \boxminus p(x) =\n{\fieldq}{\frac{1}{p(x)}} =   -1 \boxtimes p(x) \nonumber \textrm{.}
\end{equation}

The compatibility of scalar multiplication with the multiplication in $\mathbb{R}$ 
is obvious from (\ref{multiplicationdef}). The distributivity of multiplication 
over scalar and vector additions are direct consequences of the definitions of scalar multiplication
and addition. Clearly, $1$ is the identity 
element of scalar multiplication. Hence, $\setfieldq$ becomes a \emph{linear vector space} over $\mathbb{R}$.
\end{proof}

\begin{example} \label{basicalgebraicexample}
The algebraic relations between 
some conditional pmfs is examined in this example in which a combined  experiment 
is taking place in a two dimensional universe. 

First a fair die with three faces\footnote
{We can have a die with three faces in a two dimensional universe. This is the reason why the 
experiment takes place in a two dimensional universe. } is rolled. Then one of the 
three urns is selected corresponding to the outcome of the die rolling 
experiment. These three urns contain balls of six different colors.
 The number of balls of different colors in 
each urn is given in the table below. 
A ball is drawn from the selected urn and replaced back a few times. 

\begin{table}[h]
  \caption{Number of balls in different colors in each urn mentioned in Example \ref{basicalgebraicexample}.}
\begin{center}
\begin{tabular}{|c|c|c|c|c|c|c|}
 \hline
 & Red (R) & Yellow (Y) & Orange (O) & Blue (B) & Green (G) & Purple (P)  \\
\hline
 Urn 1 & 1   &  9   &   9   &  3   &   1    &  1   \\
 Urn 2 & 9  &   1   &   9   &  1    &  3    &  1    \\
 Urn 3 & 9  &  9   &    1   &  1    &   1    &  3    \\
\hline
\end{tabular}
 \end{center}
\end{table}

Let the event space of the die rolling experiment  be mapped to a $\field{3}$-valued random variable $X$
such that the faces $1,2$, and $3$ are mapped to $0,1$, and $2$ in $\field{3}$. Let six 
pmfs of X conditioned on the color of the ball drawn  be defined as follows when a single ball 
is drawn. 
\begin{equation}
  \begin{array}{ccc}
    r(x)\triangleq \Pr\{X=x|\textrm{ R }\} &  y(x)\triangleq \Pr\{X=x|\textrm{ Y }\} 
 &o(x)\triangleq \Pr\{X=x|\textrm{ O }\}  \\
    b(x)\triangleq \Pr\{X=x|\textrm{ B }\} &  g(x)\triangleq \Pr\{X=x|\textrm{ G }\} &
 p(x)\triangleq \Pr\{X=x|\textrm{ P }\}  
  \end{array} \label{condpmfsex}
\end{equation} 

For instance, assume that a ball is drawn from the selected urn and replaced back six times
and the colors of the balls drawn are B, B, G, G, G, and Y. Then the a posteriori pmf
of $X$ can be expressed by using  the definitions of 
addition and scalar multiplication in $\setfield{3}$ as
\begin{equation}
  \Pr\{X=x|B,B,G,G,G,Y\}= 2\boxtimes b(x)\boxplus 3\boxtimes g(x) \boxplus y(x) \nonumber \textrm{.}
\end{equation}
 
Now assume that the process of drawing a ball and replacing is repeated three times and
the colors of the drawn balls are R, Y, and O. Then due to the symmetry in the 
problem the a posteriori pmf of $X$ is
\begin{equation}
  \Pr\{X=x|R,Y,O\}=\frac{1}{3} \textrm{.} \nonumber
\end{equation}
Vectorial representation of this equation in $\setfield{3}$
is 
\begin{equation}
  r(x)\boxplus y(x)\boxplus o(x)= \theta(x) \label{rel1}\textrm{.}
\end{equation}
Similarly, $b(x)$, $g(x)$, and $p(x)$ are also related as 
\begin{equation}
  b(x)\boxplus g(x) \boxplus p(x) = \theta(x)\label{rel2} \textrm{.}
\end{equation}

Now assume that the process of drawing a ball and replacing is repeated twice. The a posteriori 
pmf of $X$ given the colors of the balls  are R and Y is 
\begin{equation}
  \Pr\{X=x|R,Y\}=r(x)\boxplus y(x)=\left\{\begin{array}{ccc}1/11 &,&x=0 \\ 1/11 &,&x=1 \\ 9/11 &,&x=2 \end{array}\right.
\end{equation} 
and the a posteriori pmf of $X$ given both balls are P is
\begin{equation}
  \Pr\{X=x|P,P\}=2\boxtimes p(x)=\left\{\begin{array}{ccc}1/11 &,&x=0 \\ 1/11 &,&x=1 \\ 9/11 &,&x=2 \end{array}\right. 
\textrm{.}
\end{equation} 
Combining these last two results yields
\begin{equation}
  r(x)\boxplus y(x)=2 \boxtimes p(x) \label{rel3}\textrm{.}
\end{equation}
The following two relations can be obtained similarly. 
\begin{eqnarray}
  r(x) \boxplus o(x)&=&2 \boxtimes g(x) \label{rel4} \\
  o(x) \boxplus y(x)&=&2 \boxtimes b(x) \label{rel5}
\end{eqnarray}

Actually, the algebraic relations (\ref{rel1}), (\ref{rel2}), (\ref{rel3}), (\ref{rel4}), and (\ref{rel5}) are all obtained
 by using only the basic tools of probability and the definitions of addition and scalar multiplication
in $\setfield{3}$. We did not make use of the algebraic structure defined on $\setfield{3}$ to derive these 
relations. Further algebraic
relations between the conditional pmfs defined in (\ref{condpmfsex}) can be obtained by using (\ref{rel1}), (\ref{rel2}), (\ref{rel3}), 
(\ref{rel4}), and (\ref{rel5}) and 
exploiting the algebraic structure of $\setfield{3}$. Some of these relations are given below.  
\begin{equation}
  \begin{array}{ccc}
   o(x)=-2\boxtimes p(x) & y(x)=-2\boxtimes g(x)  & r(x)=-2\boxtimes b(x) \\
   p(x)=-\frac{1}{2}\boxtimes o(x) & g(x)=-\frac{1}{2}\boxtimes y(x) & b(x)=-\frac{1}{2}\boxtimes r(x)  
     \end{array}
\end{equation}

\end{example}

\begin{example} Since it is proven that $\setfieldq$ is  a linear vector space
we can talk about linear mappings (transformations) from $\setfieldq$
 to other linear vector spaces. In this example 
we are going to provide a familiar example for such a mapping.

The log-likelihood ratio (LLR), which is defined for binary valued pmfs 
as 
\begin{equation}
  \Lambda\{p(x)\} \triangleq \log \frac{p(0)}{p(1)} \textrm{,}
\end{equation}
is a frequently employed tool in detection theory and channel decoding. 
 For any $\alpha,\beta \in \mathbb{R}$ and $p(x),r(x)\in \setfield{2}$, 
\begin{eqnarray}
  \Lambda\left\{\alpha\boxtimes p(x) \boxplus \beta \boxtimes r(x)\right\}&=& \log \frac{\n{\field{2}}{(p(x))^{\alpha}(r(x))^{\beta}}\evalat{x=0}}
  {\n{\field{2}}{(p(x))^{\alpha}(r(x))^{\beta}}\evalat{x=1}} \nonumber \\
  &=& \log \frac{ (p(0))^{\alpha}(r(0)^{\beta})}{(p(1))^{\alpha}(r(1))^{\beta}} \nonumber\\
  &=& \alpha \Lambda\{p(x)\}+\beta \Lambda\{r(x)\} \nonumber \textrm{.}
\end{eqnarray} 
Hence, the LLR  is a linear mapping from $\setfield{2}$ to $\mathbb{R}$.
\end{example}

\section{The Geometric Structure over $\setfieldq$ \label{geometricstructure}}

The geometric structure over a vector space is defined by means of an inner product. We are going
to define an inner product on $\setfieldq$ by first mapping the vectors of $\setfieldq$ to 
$\mathbb{R}^{q}$ and then borrowing the usual inner product (dot product) on $\mathbb{R}^{q}$. Such 
a mapping should posses the properties stated in the following lemma. 
 
\begin{lemma}\label{mappinglemma} Let $\operator{M}{.}$ be a mapping from $\setfieldq$ to $\mathbb{R}^{q}$
and a function $\sigma(.,.):\setfieldq\times\setfieldq\rightarrow \mathbb{R}$ be defined as
\begin{equation}
  \sigma(p(x),r(x)) \triangleq \innerproduct{\operator{M}{p(x)}}{\operator{M}{r(x)}}_{\mathbb{R}^{q}} \label{innerproductdef}\textrm{,}
\end{equation}
where $\innerproduct{.}{.}_{\mathbb{R}^{q}}$ denotes the usual inner product on $\mathbb{R}^{q}$.
$\sigma(p(x),r(x))$ is an inner product on $\setfieldq$ if $\operator{M}{.}$ is linear and injective (one-to-one).
\end{lemma}
The proof of this lemma is given in Appendix \ref{proofmappinglemma}.

We propose the following mapping from $\setfieldq$ to $\mathbb{R}^{q}$ and show later that it is linear and injective 
\begin{eqnarray}
  \operator{L}{p(x)}\triangleq \sum_{i\in\fieldq}\left( \log p(i) -\frac{1}{q} \sum_{j\in \fieldq} \log p(j) \right)\vect{e}_{i} \textrm{,}
\label{operatordef}
\end{eqnarray}
where $\vect{e}_{i}$ is the $i^{th}$ canonical basis vector of $\mathbb{R}^{q}$\footnote{The canonical basis vectors of $\mathbb{R}^{q}$ are 
usually enumerated with integers from $1$ up to $q$. In this thesis we enumerate the canonical basis vectors of $\mathbb{R}^{q}$
with the elements of $\fieldq$. Since there are $q$ canonical basis vectors of $\mathbb{R}^q$ and $q$ elements in $\fieldq$ there is not 
any problem in this enumeration.}. The proposal for $\operator{L}{.}$ is inspired by the meaning of angle between two pmfs. The details of 
 arriving at  the definition of $\operator{L}{.}$ is given in Appendix \ref{anglerationale}.  

\begin{lemma}\label{operatorlemma}
 The mapping $\operator{L}{.}:\setfieldq \rightarrow \mathbb{R}^q$ as defined in (\ref{operatordef})
 is linear and injective. 
\end{lemma}

The proof is given Appendix \ref{operatorlemmaproof}. 

It is a common practice to map pmfs to log-probability vectors 
in the turbo decoding and sum-product algorithm literature.
The main difference between those mappings and the mapping 
$\operator{L}{.}$ that we propose is the normalization
($-\frac{1}{q}\sum_{j\in \fieldq}\log p(j)$ ) 
in the definition of  $\operator{L}{.}$. This normalization 
is necessary to make the operator $\operator{L}{.}$ linear
and consequently allows us to borrow the inner product
on $\mathbb{R}^q$. In other words, it is this normalization 
which allows us to construct a geometric structure on 
$\setfieldq$. We believe that omitting this normalization 
in the literature hindered discovering the geometric
relations between pmfs.  

Obviously, the mapping $\operator{L}{.}$ is not the only mapping
which satisfies the conditions imposed by Lemma \ref{mappinglemma}.
However, $\operator{L}{.}$ exhibits a symmetric form.
This symmetry leads us to a useful geometric structure on $\setfieldq$.  

\begin{theorem} The function  $\innerproduct{.}{.}:\setfieldq\times\setfieldq\rightarrow \mathbb{R}$
defined for any $p(x),r(x)\in \setfieldq$ as
\begin{equation}
  \innerproduct{p(x)}{r(x)} \triangleq \innerproduct{\operator{L}{p(x)}}{\operator{L}{r(x)}}_{\mathbb{R}^{q}} \textrm{,}  
\label{ipdef}
\end{equation}
where $\operator{L}{.}$ is defined in (\ref{operatordef}), is an inner product on 
$\setfieldq$. 
\end{theorem}

The proof directly follows from  Lemma \ref{mappinglemma} and Lemma \ref{operatorlemma}. 

The definition of the inner product on $\setfieldq$ can be simplified as follows. 
\begin{eqnarray}
  \innerproduct{p(x)}{r(x)}&=& \innerproduct{\operator{L}{p(x)}}{\operator{L}{r(x)}}_{\mathbb{R}^{q}}\nonumber\\
  &=&\sum_{i \in \fieldq} \left(\log p(i)-\frac{1}{q}\sum_{j \in \fieldq} \log p(j)\right)
\left(\log r(i)-\frac{1}{q}\sum_{j \in \fieldq} \log r(j)\right) \label{innerproduct1} \\
&=& \sum_{i \in \fieldq} \log p(i)\log r(i)-\frac{1}{q}\left(\sum_{i \in \fieldq} \log p(i)\right)
\left(\sum_{i \in \fieldq} \log r(i)\right)
\end{eqnarray}
The equation above resembles the covariance of two random variables. Indeed, it is 
 possible to express the definition of inner product in the form of 
a covariance of two real-valued random variables, which is shown in Appendix \ref{covariancederivation}.

The vector space $\setfieldq$ evolves into an inner product space by the definition 
of the inner product in (\ref{ipdef}).  Although we haven't shown what $\dim \setfieldq$ 
is yet, we can conclude that $\setfieldq$ is finite dimensional 
since there exist an injective mapping from $\setfieldq$ to $\mathbb{R}^{q}$
\footnote{We are going to show that $\dim \setfieldq = q-1$ in Theorem \ref{dimensionalitytheorem}}.  
It is well known from functional analysis theory
that any finite dimensional inner product space is complete.
Therefore, $\setfieldq$ is a \emph{Hilbert space}.

\subsection{The norm, distance, and angle on $\setfieldq$}

The inner product on $\setfieldq$ induces the following norm on $\setfieldq$
\begin{eqnarray}
  \norm{p(x)} &\triangleq & \sqrt{\innerproduct{p(x)}{p(x)}} \\
  &=& \sqrt{ \sum_{i \in \fieldq} \left(\log p(i)\right)^2 -\frac{1}{q} \left(\sum_{i\in\fieldq }\log p(i) \right)^2} \textrm{.}
\end{eqnarray}
A distance function between two pmfs can be obtained by combining this norm with the definition 
of subtraction in $\setfieldq$ as in 
\begin{eqnarray}
  D(p(x),r(x)) &\triangleq& \norm{p(x)\boxminus r(x)} \\
  &=& \sqrt{ \sum_{i \in \fieldq} \left(\log \frac{p(i)}{r(i)}\right)^2 -\frac{1}{q} \left(\sum_{i\in\fieldq }\log \frac{p(i)}{r(i)} \right)^2}
\textrm{.}
\end{eqnarray}
Since $\norm{.}$ is a proper norm, this distance is a metric distance. In other words, it is 
 nonnegative, symmetric, and it satisfies the triangle equality.   

Similar to any Hilbert space, the angle between any two pmfs $p(x),r(x)$ in $\setfieldq$ is given by
\begin{equation}
  \angle(p(x),r(x)) \triangleq \arccos \frac{\innerproduct{p(x)}{r(x)}}{\norm{p(x)}\norm{r(x)}} \textrm{.} \label{angledefinition}
\end{equation}

\subsection{The pseudo inverse of $\operator{L}{.}$}

\begin{lemma}\label{operatorortholemma} For any $p(x)$ in $\setfieldq$
  \begin{equation}
    \operator{L}{p(x)} \perp \vect{1} \textrm{,}
  \end{equation}
  where $\vect{1}$ denotes the all one vector in $\mathbb{R}^{q}$. 
\end{lemma}

The proof is given Appendix \ref{operatorortholemmaproof}. 

Since $\operator{L}{p(x)}$ is always orthogonal to $\vect{1}$ it is not a surjection (onto). Consequently,
it is not a bijection (injection and surjection). A mapping which is not a bijection does not 
have an inverse. Nonetheless, a pseudo inverse for $\operator{L}{.}$
exists which satisfies 
\begin{equation}
\pseudoinv{L}{\operator{L}{p(x)}}=p(x) \textrm{,}\nonumber
\end{equation}
where $\pseudoinv{L}{.}$ denotes the 
pseudo inverse of $\operator{L}{.}$.

$\pseudoinv{L}{.}$ is a mapping from $\mathbb{R}^q$ to $\setfieldq$. 
We propose the following definition for $\pseudoinv{L}{.}$
\begin{equation}
  \pseudoinv{L}{\vect{p}}\triangleq \n{\fieldq}{\exp\left(-\frac{1}{2}\norm{\vect{p}-\vect{s}(x)}^{2}\right)}\textrm{,} \label{pseudoinvdef}
\end{equation}
where $\vect{p}$ is any vector in $\mathbb{R}^{q}$ and $\vect{s}(x)$ is the vector-valued function from 
$\fieldq$ to $\mathbb{R}^{q}$ given by
\begin{equation}
  \vect{s}(x) \triangleq \vect{e}_{x}-\frac{1}{q}\vect{1} \label{simplexmod}\textrm{.}
\end{equation}
The definition of $\pseudoinv{L}{.}$ can be interpreted as in
\begin{equation}
  \pseudoinv{L}{\vect{p}}=\Pr\{X=x|\vect{s}(X)+\vect{N}=\vect{p}\} \textrm{,}
\end{equation} 
where $\vect{N}$ is random vector whose components are all independent, real, zero-mean Gaussian random 
variables with unit variance. 
Furthermore, notice that 
the function  $\vect{s}(x)$ maps the elements of $\fieldq$ to 
$\mathbb{R}^q$  as in the simplex modulation.

\begin{lemma}\label{pseudoinvlemma} $\pseudoinv{L}{.}:\mathbb{R}^{q}\rightarrow \setfieldq$ defined in (\ref{pseudoinvdef})  
satisfies
\begin{equation}
\pseudoinv{L}{\operator{L}{p(x)}}=p(x) 
\end{equation}
for all $p(x)$ in $\setfieldq$. Moreover,
\begin{equation}
  \operator{L}{\pseudoinv{L}{\vect{p}}}=\vect{p}
\end{equation}
if $\vect{p} \perp \vect{1}$. 
\end{lemma}

The proof is given in Appendix \ref{pseudoinvlemmaproof}. 


\begin{theorem}\label{dimensionalitytheorem}
$\setfieldq$ is a $q-1$ dimensional Hilbert space, i.e..  
  \begin{equation}
    \dim \setfieldq = q-1
  \end{equation}
\end{theorem}
\begin{proof} 
Due to the rank-nullity theorem in linear algebra
\begin{equation}
  \dim \setfieldq = \dim \image{\set{L}}+\dim \kernel{\set{L}} \nonumber \textrm{,}
\end{equation}
where $\image{\set{L}}$ and $\kernel{\set{L}}$ denote the image and kernel (null space) of $\operator{L}{.}$ respectively. 
Since $\operator{L}{.}$ is shown to be an injection in Lemma \ref{operatorlemma}, $\kernel{\set{L}}$ only contains $\vect{0}$. 
It can be deduced from Lemma \ref{operatorortholemma} that the image (range space) of $\operator{L}{.}$ is a subset of $\vect{1}^{\perp}$,
where $\vect{1}^{\perp}$ is the subspace of $\vect{R}^{q}$ given by
\begin{equation}
  \vect{1}^{\perp}\triangleq \left\{\vect{p} \in \mathbb{R}^q : \innerproduct{\vect{p}}{\vect{1}}_{\mathbb{R}^q}=0 \right\}
\end{equation} 
The second part of Lemma  \ref{pseudoinvlemma} improves this result as it clearly shows that the image of 
$\operator{L}{.}$ is exactly equal to $\vect{1}^{\perp}$. Therefore,
\begin{eqnarray}
  \dim \setfieldq &=& \dim \vect{1}^{\perp}+\dim \{\vect{0}\} \nonumber \\
  &=& q-1 \textrm{,}
\end{eqnarray}
 which completes the proof. 
\end{proof}

\subsection{A set of  orthonormal basis pmfs for $\setfieldq$ \label{orthobasispmfs}} 

A set of $q-1$ linearly independent vectors are necessary to form a basis
for $\setfieldq$. An orthonormal basis for $\setfieldq$ can be 
obtained by finding a set of orthonormal vectors in $\vect{1}^{\perp}$
and then by mapping these vectors to  $\setfieldq$ via  $\pseudoinv{L}{.}$.  
Let $q-1$ vectors in $\mathbb{R}^{q}$ be defined as
\begin{equation}
  \begin{array}{cccccccc}
    \vect{s}_1 & \triangleq  [&\frac{1}{\sqrt{2}} & -\frac{1}{\sqrt{2}}  & 0  & \ldots & 0&] \\
    \vect{s}_2 & \triangleq  [&\frac{1}{\sqrt{6}} &  \frac{1}{\sqrt{6}} & -\frac{2}{\sqrt{6}}  & \ldots & 0&] \\

    \vdots & & \vdots& \vdots & \vdots& \ddots & \vdots \\
    \vect{s}_{q-1} &\triangleq [&\frac{1}{\sqrt{q(q-1)}} & \frac{1}{\sqrt{q(q-1)}}& \frac{1}{\sqrt{q(q-1)}} &\ldots &-\frac{q-1}{\sqrt{q(q-1)}}&]   
\end{array} \textrm{.}\label{orthoinR}
\end{equation}
Clearly, all of these vectors are all in $\vect{1}^{\perp}$ and they are all mutually orthonormal. $q-1$ pmfs in $\setfieldq$
can be obtained by mapping these vectors to $\setfieldq$ via $\pseudoinv{L}{.}$ as follows.
\begin{equation}
  s_{i}(x) \triangleq \pseudoinv{L}{\vect{s}_i}\quad  \textrm{ for }i=1,2,\ldots,q-1 \label{orthoinpmf} \textrm{.}
\end{equation}
Due to the definition of the inner product and the second part of Lemma \ref{pseudoinvlemma},
\begin{eqnarray}
  \innerproduct{s_{i}(x)}{s_{j}(x)} &=& \innerproduct{\operator{L}{s_{i}(x)}}{\operator{L}{s_{j}(x)}}_{\mathbb{R}^{q}}\nonumber \\
&=&\innerproduct{\vect{s}_i}{\vect{s}_j}_{\mathbb{R}^q} \nonumber \\
  &=&\left\{ \begin{array}{cc} 1 & \textrm{, for }i=j\\ 0 & \textrm{, for }i\neq j \end{array}  \right. \textrm{.}
\end{eqnarray}
Therefore, $\{s_{1}(x),s_2(x),\ldots,s_{q-1}(x)\}$ is an orthonormal basis for $\setfieldq$.

\begin{example} \label{basicgeometricexample}
In Example \ref{basicalgebraicexample} basic algebraic relations between six pmfs, which 
are in $\setfield{3}$, is investigated. An orthonormal basis for $\setfield{3}$ is composed
of two pmfs. $s_1(x)$ and $s_2(x)$ given below forms such a basis for $\setfield{3}$.
\begin{eqnarray}
  s_{1}(x)&=& \pseudoinv{L}{\left[
      \begin{array}{ccc}\frac{1}{\sqrt{2}}&-\frac{1}{\sqrt{2}}&0 \end{array}\right]^T  } \nonumber \\
  &\simeq&\left\{\begin{array}{cc}0.57598, &x=0 \\ 0.14002, & x=1 \\ 0.28400,& x=2 \end{array} \right.
\end{eqnarray} 

\begin{eqnarray}
  s_{2}(x)&=& \pseudoinv{L}{\left[
      \begin{array}{ccc}\frac{1}{\sqrt{6}}&\frac{1}{\sqrt{6}}&-\frac{2}{\sqrt{6}} \end{array}\right]^T   } \nonumber \\
  &\simeq&\left\{\begin{array}{cc}0.43595, & x=0 \\ 0.43595, & x=1 \\ 0.12810,& x=2 \end{array} \right.
\end{eqnarray} 

The coordinates of a pmf in $\setfield{3}$, 
with respect to (w.r.t.) the basis $\{s_1(x), s_2(x)\}$ is simply the inner product of the pmf with $s_1(x)$ and $s_2(x)$. 
For instance, $r(x)$ mentioned in Example \ref{basicalgebraicexample}  can be expressed as
\begin{eqnarray}
  r(x)&=&\innerproduct{r(x)}{s_1(x)}\boxtimes s_1(x) \boxplus \innerproduct{r(x)}{s_2(x)}\boxtimes s_2(x)  \nonumber \\
  &\simeq & -1.5537\boxtimes s_1(x) \boxplus -0.89701\boxtimes s_2(x) \textrm{.}
\end{eqnarray}
The coordinates of all the pmfs mentioned in Example \ref{basicalgebraicexample} are given in the 
table below and depicted in Figure \ref{pmfsplot}.

\begin{table}[h]
\caption{Coordinates of the pmfs mentioned in Examples \ref{basicalgebraicexample} and \ref{basicgeometricexample}}
\begin{center}
\begin{tabular}{|c|c|c|c|c|c|c|}
\hline
  &$r(x)$ & $y(x)$ & $o(x)$ & $b(x)$ & $g(x)$ & $p(x)$ \\
\hline
 $s_1(x)$ & $-1.55367$  & $ 1.55367$& $0$       & $0.77684$& $-0.77684$ & $0$ \\
 $s_2(x)$ & $-0.89701$  & $-0.89701$& $1.79403$ & $0.44851$& $0.44851$ & $-0.89701$\\
\hline
\end{tabular}
\end{center}
\end{table}
\end{example}

\begin{figure}
  \begin{center}
    \includegraphics[scale=.5]{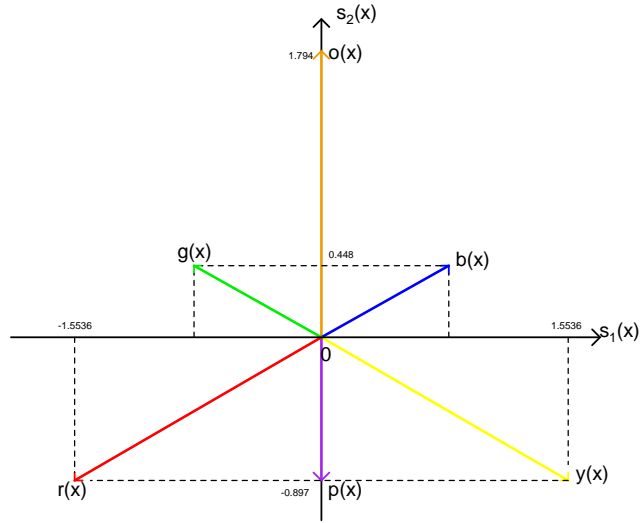}
    \caption{Plot of the pmfs mentioned in Examples \ref{basicalgebraicexample} and \ref{basicgeometricexample}\label{pmfsplot}}
  \end{center}
\end{figure} 

\section{Relation to the Hilbert  space of random variables \label{relhsrv}}

The Hilbert space of probability mass functions of finite field-valued random variables might 
be confused with the Hilbert space of random variables with a finite second order moment which is already well known \cite{papoulis}. 
However, these two Hilbert spaces are quite different from each other. First of all, the vectors of 
the former Hilbert space are pmfs of the random variables whereas the vectors of the latter Hilbert 
space are the random variables themselves. Second, the former Hilbert space is related to the
finite-field valued random variables whereas the latter is related to the \emph{complex-valued}
random variables. Finally, the former is meaningful in the Bayesian detection sense whereas the
latter is not. 
 
Although this thesis is about the Hilbert space of the pmfs of finite field-valued random variables,
it is adequate to summarize the Hilbert space of random variables.  
The set of \emph{complex-valued} random variables forms vector space with the usual random variable
addition and scaling over $\mathbb{C}$. This vector space can be endowed with the following inner product
which is nothing but the autocorrelation between two random variables.
\begin{equation}
\innerproduct{X}{Y}=\expectation{XY^{*}}   \textrm{,} \label{rvinnerproduct}
\end{equation}
where $X$ and $Y$ are two \emph{complex-valued} random variables and $\expectation{.}$ denotes the expectation. 
The set of complex-valued random variables with finite second order moment is complete w.r.t. the norm 
induced by the inner product above. Therefore, this set forms a Hilbert space over $\mathbb{C}$ with 
the usual random variable addition, scaling, and the inner product given in (\ref{rvinnerproduct}).
Many important algorithms, such as the Wiener filter, relies upon the orthogonality in this Hilbert space. 
 
Notice that the Hilbert space structure over random variables is constructed over complex-valued random variables. 
Although it is also possible to construct a similar \emph{vector space} over the set of finite field-valued random variables, 
the vector space of finite field-valued random variables does not have an inner product. In other words, 
the set of $\fieldq$-valued random variables forms a vector space with the usual random variable addition
and scaling over $\fieldq$. On the contrary to complex-valued random variable case, the expected value is not
a well defined concept for finite field-valued random variables. Consequently, autocorrelation 
between two $\fieldq$-valued random variables is not well defined either. Therefore, we cannot construct a Hilbert space
structure  over the set of $\fieldq$-valued random variables as we could for the complex valued random variables.
If we had a Hilbert space structure over the set of $\fieldq$-valued random variables then 
we would have  decoding algorithms for linear channel codes with polynomial complexity.

\subsection{Comparison between the convergence of random variables and pmfs}

Another possible confusion might arise between the convergence of finite field-valued random variables
and the convergence of pmfs of finite field-valued random variables. As explained above 
expectation is not well defined for finite field-valued random variables. Therefore, 
convergence in the mean square sense is not well defined for finite field-valued random variables either.  
On the other hand, convergence modes such as convergence almost everywhere and convergence in probability can 
still be well defined. However, due to the topological nature of the finite fields these two convergence modes
are essentially equivalent.  Convergence  of a sequence of finite-field-valued random variables in probability
is formally defined below. 

\begin{definition}\emph{Convergence of a sequence of finite-field-valued random variables in probability:}
A sequence of $\fieldq$-valued random variables, $\left\{X_n\right\}_{n=1}^{\infty}$,
converges in probability to an $\fieldq$-valued random variable $X$ if and only if for each $\epsilon>0$ 
there exist an integer $N$ such that 
\begin{equation}
    n>N \implies \Pr\{X_{n}=X\}>1-\epsilon
\end{equation}
and this convergence is denoted by
\begin{equation}
\lim_{n\rightarrow \infty} \Pr\{X_{n}=X\} =1 \textrm{.}
\end{equation}
\end{definition} 

Convergence of $\fieldq$-valued random variables
in probability, might be confused with the convergence of pmfs in $\setfieldq$. The following 
example aims to clarify the distinction between these two convergences. 

\begin{example}
Let the event space of an experiment $\Omega$ be $[0,1]\subset \mathbb{R}$ and each outcome of the 
experiment is equally likely, i.e. 
\begin{equation}
  \Pr\{\omega\leq c\} = c \textrm{,}
\end{equation}
where $\omega$ denotes the outcome of the experiment. A sequence of $\field{2}$-valued
random variables, $\left\{X_n\right\}_{n=1}^{\infty}$, are assigned to this experiment
as follows.
\begin{equation}
  X_{n}(\omega)\triangleq\left\{\begin{array}{cc} 0, &\omega\in [0, 1-2^{-n}]\\ 1, & \omega \in (1-2^{-n}, 1] \end{array}\right.
\end{equation}
Clearly, the sequence $\left\{X_n\right\}_{n=1}^{\infty}$ converges in probability to a random variable $X$ which 
is defined as 
\begin{equation}
  X\triangleq \left\{ \begin{array}{cl} 0, &  \omega \in [0, 1]\\ 1, & \omega \in \emptyset \end{array}\right.
\end{equation}
In other words,
\begin{equation}
  \lim_{n\rightarrow \infty}\Pr\{X_{n}= X\}=1 \textrm{.}
\end{equation}

Let a sequence $\left\{p_{n}(x)\right\}_{n=1}^{\infty}$ of pmfs in $\setfield{2}$ be defined as
\begin{eqnarray}
  p_{n}(x)&\triangleq&\Pr\{X_{n}=x \}\nonumber \\
  &=&\left\{\begin{array}{cc}1-2^{-n}, & x=0 \\
2^{-n},& x=1  \end{array} \right. 
\end{eqnarray} 
and  $p(x)$ denote $\Pr\{X=x\}$. Due to the basic axioms of probability 
\begin{equation}
  p(x)=\left\{\begin{array}{cc} 1 &,x=0\\ 0 &,x=1 \end{array}\right. \textrm{.}
\end{equation}
It might appear at a first glance that the sequence $\left\{p_{n}(x)\right\}_{n=1}^{\infty}$ converges
to $p(x)$. However, this would contradict with the completeness of $\setfield{2}$ since
$p(x)\notin \setfield{2}$. The  truth is $\left\{p_{n}(x)\right\}_{n=1}^{\infty}$ is not a Cauchy sequence
in $\setfield{2}$. This fact can be shown as  follows. For any $m>n>0$ 
\begin{eqnarray}
  D(p_{m}(x),p_n(x)) &=&\sqrt{\sum_{i\in \field{2}}\left(\log \frac{p_{m}(i)}{p_{n}(i)} \right)^2-\frac{1}{q}
  \left(\sum_{i\in \field{2}}\log \frac{p_{m}(i)}{p_{n}(i)}\right)^2} \nonumber \\
  &=&\frac{1}{\sqrt{2}}\left(\log \frac{p_{m}(0)}{p_{n}(0)}+\log\frac{p_{n}(1)}{p_{m}(1)}\right) \nonumber \textrm{.}
\end{eqnarray}
Since $p_{m}(0)>p_{n}(0)$
\begin{eqnarray}
  D(p_{m}(x),p_n(x)) &>&\frac{1}{\sqrt{2}}\left(\log\frac{p_{n}(1)}{p_{m}(1)}\right) \nonumber \\
&=&\frac{\log 2}{\sqrt{2}}(m-n) \textrm{.}
\end{eqnarray}
Therefore, $\left\{p_{n}(x)\right\}_{n=1}^{\infty}$ is not a Cauchy sequence and 
the limit  $\lim_{n\rightarrow \infty}p_{n}(x)$ does not exist. This example demonstrates that 
convergence of a sequence of random variables in probability does not imply the 
convergence of their pmfs. 

\end{example}

\section{The Hilbert space of multivariate pmfs \label{mvpmf}}

The construction of the Hilbert space on $\setfieldq$ can be applied to the 
set of multivariate (joint) pmfs as well. Basically, we should replace the indeterminate variable $x$
in the Hilbert space of pmfs with a  vector $\vect{x}$ while constructing the
Hilbert space structure on multivariate pmfs.   

Let $\vect{X}=[X_1,X_2,\ldots,X_N]$
be a random vector where $X_i$ is a $\fieldq$-valued random variable. Furthermore,
let $\setfieldqn$ denote the set of all strictly positive pmfs that $\vect{X}$ might posses, i.e.
\begin{equation}
  \setfieldqn\triangleq \left\{ p(\vect{x}):\fieldqn\rightarrow (0,1)\subset
 \mathbb{R}, \sum_{\vect{x}\in \fieldqn} p(\vect{x})=1\right\} \textrm{.}
\end{equation}
The addition and scalar multiplication on $\setfieldqn$ can be defined for any 
$p_1(\vect{x}),p_2(\vect{x}),p(\vect{x})\in \setfieldq$ and $\alpha \in \mathbb{R}$ as 
\begin{eqnarray}
  p_1(\vect{x})\boxplus p_2(\vect{x}) & \triangleq & \n{\fieldqn}{p_1(\vect{x})p_2(\vect{x})} \\
\alpha \boxtimes p(\vect{x})&\triangleq & \n{\fieldqn}{(p(\vect{x}))^{\alpha}} 
\end{eqnarray}
The normalization operator in the multivariate case, which is denoted by  $\n{\fieldqn}{.}$ above, 
maps any strictly positive function of $\fieldqn$, $\alpha(\vect{x})$, to a pmf in $\setfieldqn$ as  follows.
\begin{equation}
  \n{\fieldqn}{\alpha(\vect{x})}\triangleq \frac{\alpha(\vect{x})}{\sum_{\vect{i}\in \fieldqn}\alpha (\vect{i})} 
\end{equation}
Similar to the univariate case, $\setfieldqn$ together with the $\boxplus$ and $\boxtimes$ operations forms
a vector space over $\rfield$.

The analogue of the mapping $\operator{L}{.}$ in the multivariate case  is denoted by $\operatorn{L}{.}$
and maps the pmfs in $\setfieldqn$ to $\rfieldqn$. Before giving the definition of 
$\operatorn{L}{.}$ we need to establish a one-to-one matching between the \emph{vectors} in $\fieldqn$ and
the \emph{canonical basis vectors} of $\rfieldqn$.  We can do this matching since $\fieldqn$ contains $q^{N}$
vectors which is equal to the dimension of  $\rfieldqn$. Since the mapping $\operatorn{L}{.}$ is going to be employed
in borrowing the inner product in $\rfieldqn$ the order of matching is not important.

Using this matching  $\operatorn{L}{.}$
is defined as
\begin{equation}
  \operatorn{L}{p(\vect{x})}:\setfieldqn\rightarrow \rfieldqn \triangleq 
  \sum_{\vect{i}\in\fieldqn} \left(\log p(\vect{i})-\frac{1}{q^N}\sum_{\vect{j}\in \fieldqn} \log  p(\vect{j}) \right)
  \vect{e}_{\vect{i}} \textrm{,}
\end{equation}
where $\vect{e}_{\vect{i}}$ denotes the canonical basis vector of $\rfieldqn$ matched to $\vect{i}\in\fieldqn$. 
$\operatorn{L}{.}$ is a linear and injective mapping as $\operator{L}{.}$. 
Then the inner product of any two $p(\vect{x}),r(\vect{x}) \in \setfieldqn$ becomes
\begin{eqnarray}
\innerproduct{p(\vect{x})}{r(\vect{x})} &\triangleq& \innerproduct{\operatorn{L}{p(\vect{x})}}
{\operatorn{L}{r(\vect{x})}}_{\rfieldqn}  \\
&=&\sum_{\vect{i}\in \fieldqn}\log p(\vect{i})\log r(\vect{i}) 
-
\frac{1}{q^{N}}
\left( \sum_{\vect{i}\in\fieldqn}\log p(\vect{i})\right)
\left( \sum_{\vect{i}\in\fieldqn}\log r(\vect{i})\right) \textrm{.}
\end{eqnarray} 
The definition of inner product makes $\setfieldqn$ an inner product space. Since $\setfieldqn$
is definitely finite dimensional it is also a Hilbert space.

The pseudo inverse of $\operatorn{L}{.}$ is
\begin{equation}
  \pseudoinvn{L}{\vect{p}}:\rfieldqn\rightarrow \setfieldqn 
  \triangleq \n{\fieldqn}{\exp\left(-\frac{1}{2}\norm{\vect{p}-\vect{s}_{N}(\vect{x})}^2 \right)} \textrm{,}
\end{equation} 
where $\vect{s}_{N}(\vect{x})$ is 
\begin{equation}
  \vect{s}_{N}(\vect{x}) \triangleq \vect{e}_{\vect{x}}-\frac{1}{q^{N}}\vect{1} \textrm{.} 
\end{equation}
The vector $\vect{1}$ above denotes the all one vector in $\rfieldqn$. 
Similar to the univariate case it can be shown that $\pseudoinvn{L}{.}$
satisfies 
\begin{eqnarray}
  &\pseudoinvn{L}{\operatorn{L}{p(\vect{x})}} = p(\vect{x}) &\quad \forall p(\vect{x})\in \setfieldq \\
  &\operatorn{L}{\pseudoinvn{L}{\vect{p}}}= \vect{p} &\quad \forall \vect{p} \in \vect{1}^{\perp} \subset \rfieldqn  \textrm{.}
\end{eqnarray}
Consequently, 
\begin{equation}
  \image{\set{L}_{N}} = \vect{1}^{\perp}  \subset \rfieldqn
\end{equation}

\begin{theorem} $\setfieldqn$ is a $q^{N}-1$ dimensional Hilbert space, i.e.
\begin{equation}
  \dim \setfieldqn = q^{N}-1 \textrm{.}
\end{equation}
\end{theorem}
\begin{proof}
Due to the rank-nullity theorem in linear algebra
\begin{eqnarray}
  \dim \setfieldqn &=& \dim \kernel{\set{L}_N} + \dim \image {\set{L}_N} \\
  &=&q^{N}-1 \textrm{.} 
\end{eqnarray} 
\end{proof}

As a minor consequence of this theorem we can conclude that 
$\setfieldqn$ is isomorphic to $\setfield{q^{N}}$. This is a quite 
expected result since $\fieldqn$ is isomorphic to $\field{q^{N}}$.


\chapter[THE CANONICAL FACTORIZATION OF MULTIVARIATE \\ PROBABILITY MASS FUNCTIONS]{THE CANONICAL FACTORIZATION OF \\MULTIVARIATE  PROBABILITY MASS FUNCTIONS \label{thechapter}}

\section{Introduction}

The factorization of a multivariate pmf is important in many aspects. For instance, the conditional 
dependence of the random variables distributed by a pmf can be determined by how the
pmf factors. Existence of low complexity maximization and marginalization 
algorithms for a multivariate pmf, such as Viterbi and BCJR,  also depends
on the factorization of the pmf. A very special factorization of  multivariate pmfs 
which we call as the canonical factorization is introduced in this chapter. 

This chapter begins with representing the factorization of a pmf in $\setfieldqn$. 
Then we introduce the soft parity check constraints using which we decompose
  $\setfieldqn$ into orthogonal subspaces. Finally, we obtain the canonical
factorization of pmfs as the projection of pmfs onto these subspaces.

\section{Representing the  factorization of pmfs}

The Hilbert space $\setfieldqn$ provides a suitable environment for analyzing the 
factorization of multivariate pmfs. Suppose that a pmf in $\setfieldqn$
can be factored as
\begin{equation}
  p(\vect{x})=\prod_{i=1}^{K}\phi_{i}(\vect{x}) \textrm{.} \label{firstfactorization}
\end{equation}
Each $\phi_{i}(\vect{x})$ function appearing above may be called 
a factor function, a local function, a constraint, or an interaction. The factor functions
are not necessarily pmfs but they can be assumed to be 
positive. Hence, we can obtain a pmf in $\setfieldq$ 
by scaling the factor functions as in
\begin{eqnarray}
  r_{i}(\vect{x})&=&\n{\fieldqn}{\phi_{i}(\vect{x})} \\
  &=&\frac{1}{\gamma_i}\phi_{i}(\vect{x})\textrm{,}
\end{eqnarray}
where $\gamma_i=\sum_{\vect{i}\in \fieldqn}\phi_{i}(\vect{x})$.
After this normalization the factorization in (\ref{firstfactorization})
becomes
\begin{eqnarray}
  p(\vect{x})&=&\prod_{i=1}^{K}\gamma_{i}r_{i}(\vect{x}) \\
  &=&\n{\fieldqn}{\prod_{i=1}^{K}r_{i}(\vect{x})} \textrm{,}
\end{eqnarray}
which can be represented using the addition in $\setfieldq$ as 
\begin{equation}
  p(\vect{x})=\ssum_{i=1}^{K}r_{i}(\vect{x}) \textrm{.}
\end{equation}
This representation suggests that  a multivariate pmf in $\setfieldqn$ can be 
factored by expressing it as a linear combination of some basis vectors (pmfs)  
in $\setfieldqn$. If these basis pmfs are chosen to be orthogonal 
then we can employ the inner  product on $\setfieldqn$ to determine
the expansion coefficients. However, the basis pmfs should be selected
in such a way that the resulting factorization becomes \emph{useful}. 

We know from the literature on the sum-product  algorithm  \cite{fgsp,aloefg,Wiberg,wiberg} and 
Markov random fields \cite{bishop,hct1,hct2,pabbeel} that 
the factorization of $p(\vect{x})$ given in (\ref{firstfactorization}) is useful
if the factor functions on the right hand side of (\ref{firstfactorization})
are \emph{local}. A factor function of $p(\vect{x})$ is said to be local if it depends
on some but not all of the components of the argument vector $\vect{x}$. Therefore, 
the  basis functions mentioned in the paragraph above should also be selected to be as local 
as possible.

\section{The multivariate pmfs that can be expressed as a function of a linear \\
combination of their arguments}

In this section we propose a special type of multivariate pmfs which will serve as basis vectors
to obtain a factorization of pmfs in $\setfieldq$. We show in the next chapter that 
the factorization obtained using these basis pmfs is quite useful. These basis pmfs
are inspired by the parity check relations in $\fieldq$. Suppose that the components of
an $\fieldqn$-valued random vector $\vect{X}=[X_1,X_2,\ldots,X_N]$ satisfy the following
parity check relation
\begin{equation}
  a_1X_1+a_2X_2+\ldots+a_NX_N=0 \textrm{,} \label{paritycheckrelation}
\end{equation}
where $a_i$ is a constant in $\fieldq$. If all configurations 
satisfying this relation are assumed to be equiprobable 
then the joint pmf of $\vect{X}$, which is denoted by $p(\vect{x})$, is
\begin{equation}
  p(\vect{x})= \left\{\begin{array}{cc}\frac{1}{q^{N-1}}, & \sum_{i=1}^{N}a_ix_i=0 \\ 0, &\textrm{otherwise} \end{array} \right. \textrm{.}
\end{equation}
This pmf can be expressed in a more compact form as 
\begin{eqnarray}
  p(\vect{x})&=&\frac{1}{q^{N-1}}\delta(\vect{a}\vect{x}^T) \\
&=&\n{\fieldqn}{\delta(\vect{a}\vect{x}^T)} \label{pcc} \textrm{,}
\end{eqnarray}
where $\vect{a}$ is $[a_1,a_2,\ldots,a_N]$ and $\delta(.)$ denotes the Kronecker delta. 

The multivariate pmfs which can be expressed in the form as in (\ref{pcc}) are called parity 
check or zero-sum constraints. A parity check constraint depends only on the variables which have nonzero 
coefficients associated with them. Hence, they posses local function properties as we desire
from a basis pmf. Therefore, parity check constraints could be good candidates for being basis 
pmfs if they were elements of $\setfieldqn$. However, parity check constraints are not elements 
of $\setfieldqn$, since their value is zero for the configurations which do not satisfy 
the parity check relation. 

We can obtained a ``softened'' version of the parity check constraints as follows. Suppose
that the components of the random vector $\vect{X}$ satisfy the following relation instead
of (\ref{paritycheckrelation})
\begin{equation}
  a_1X_1+a_2X_2+\ldots+a_NX_N=U \textrm{,} \label{softparitycheckrelation}
\end{equation}
where $U$ is an $\fieldq$-valued random variable distributed with an $r(u)\in\setfieldq$. 
If all configurations resulting with the same value of $U$ are assumed to be equiprobable
then joint pmf of $\vect{X}$ in this case becomes
\begin{eqnarray}
  p(\vect{x}) &=& \left\{\begin{array}{cl} \frac{1}{q^{N-1}}r(0), & \sum_{i=1}^{N}a_ix_i=0 \\
      \frac{1}{q^{N-1}}r(1), & \sum_{i=1}^{N}a_ix_i=1 \\ 
    \vdots &\vdots \\
     \frac{1}{q^{N-1}}r(q-1),& \sum_{i=1}^{N}a_ix_i=q-1 
  \end{array}\right.
\end{eqnarray} 
which can be expressed in a more compact form  as
\begin{eqnarray}
  p(\vect{x})&=&\frac{1}{q^{N-1}}r(\vect{a}\vect{x}^{T}) \\
  &=&
\n{\fieldqn}{r(\vect{a}\vect{x}^{T})} \textrm{.}
\end{eqnarray}

\begin{definition}A multivariate pmf 
$p(\vect{x})$  in 
$\setfieldqn$ is called \emph{a soft parity
check (SPC) constraint} if there exist a $r(x) \in \setfieldq$ and 
a vector $\vect{a}=[a_0,a_1,\ldots,a_{N-1}] \in \fieldqn$ such that
\begin{eqnarray}
  p(\vect{x})&=& \n{\fieldqn}{ r(\vect{a}\vect{x}^T)} \label{spcc}\textrm{.}
\end{eqnarray}
 The  vector $\vect{a}$ is called the \emph{parity check coefficient vector} of the SPC constraint
$p(\vect{x})$.
\end{definition}

The difference between parity check and SPC constraint is the distribution of the  weighted sum
of the random variables $X_0$, $X_1$, $\ldots$, $X_{N-1}$, which is denoted by $U$
in (\ref{softparitycheckrelation}). $U$ is distributed with $\delta(u)$ in the parity check case 
whereas it is distributed with a $r(u)$ in
$\setfieldq$ in the SPC constraint case. The term ``soft'' arises from the fact that the weighted sum 
can take all values with some probability rather than guaranteed to be zero. 
Therefore, unlike parity check constraints SPC constraints are in $\setfieldqn$, since all configurations have nonzero
probabilities. 

\begin{example}\label{spcexample1}
Let two pmfs in $\setfieldnn{3}{2}$  are given with a slight abuse of notation as
\begin{eqnarray}
p_1(x_0,x_1) &=&\frac{1}{30}\left[ \begin{array}{ccc}3 & 6 & 1 \\ 
6 & 1 & 3 \\ 1 & 3 & 6\end{array} \right]
\nonumber \\
p_2(x_0,x_1) &=&\frac{1}{157}\left[ \begin{array}{ccc}12 & 30 & 1 \\ 
10 & 6 & 48 \\ 12 &8 &30\end{array} \right] \nonumber \textrm{,}
\end{eqnarray} 
where $p_k(x_0=i,x_1=j)$ is given by the entry in the $(i+1)^{th}$ row and the $(j+1)^{th}$ column of the
corresponding matrix.

Notice that  $p_1(x_0,x_1)$ can be expressed as 
\begin{eqnarray}
 p_1(x_0,x_1)&=&\frac{1}{3}r(x_0+x_1) \nonumber\\
 &=&\n{\fieldnn{3}{2}}{r(x_0+x_1)} \nonumber 
\end{eqnarray}
 where $r(x)\in \setfieldn{3}$  is
\begin{equation}
  r(x)=\left\{\begin{array}{cc} 0.3, &x=0 \\ 0.6, & x=1\\ 0.1 & x=1    \end{array} \right. \textrm{.} \nonumber
\end{equation}
 Therefore, $p_{1}(x_0,x_1)$ is an SPC constraint with parity check coefficient vector $[1, 1]$. 
On the other hand, we cannot find a similar expression for $p_{2}(x_0,x_1)$. Hence,
$p_{2}(x_0,x_1)$ is not an  SPC constraint.  
\end{example}

Notice that we exploited the field structure of $\fieldq$ in the discussion above. Parity check relations
could also be described in finite rings but the number of configurations satisfying a parity check
relation depends on the parity check coefficients in a finite ring. Therefore, the SPC constraints
in a finite ring would not be in a nice form as above.

In the rest of this chapter we are going to show that SPC constraints form a complete set of orthogonal basis functions 
for $\setfieldqn$. The first step of this process is the following lemma which analyzes the inner product of 
two SPC constraints.

\begin{lemma}\label{spcinnerproduct}\emph{Inner product of two SPC constraints:} Let  $p_{1}(\vect{x}),p_{2}(\vect{x})\in \setfieldqn$ 
are two SPC constraints such that 
\begin{eqnarray}
  p_{1}(\vect{x}) &=& \n{\fieldqn}{r_1(\vect{a}\vect{x}^T)} \\
  p_{2}(\vect{x}) &=& \n{\fieldqn}{r_2(\vect{b}\vect{x}^T)} \textrm{,} \nonumber  
\end{eqnarray}
where $r_{1}(x),r_2(x)\in \setfieldq$. If $\vect{a}$ and $\vect{b}$ are both nonzero vectors
in $\fieldqn$ then
\begin{equation}
  \innerproduct{p_{1}(\vect{x})}{p_2(\vect{x})} = \left\{\begin{array}{cc} 
q^{N-1}\innerproduct{r_{1}(x)}{r_{2}(\alpha x)},  & \exists \alpha \in \fieldq : \vect{b}=\alpha \vect{a} \\
0, & \textrm{otherwise} \end{array} \right. 
\end{equation}
\end{lemma}

The proof of this lemma is given in Appendix \ref{spcinnerproductproof}. 

\section{Orthogonal Subspace Decomposition of $\setfieldqn$}
Generating an SPC constraint in $\setfieldqn$ based on a pmf in $\setfieldq$
and a parity check coefficient vector $\vect{a}$ can be viewed as a mapping from
$\setfieldq$ to $\setfieldqn$ parameterized on $\vect{a}$ as given 
below. 
\begin{equation}
\sopsub{\vect{a}}{p(x)}:\setfieldq \rightarrow \setfieldqn \triangleq \n{\fieldqn}{p(\vect{a}\vect{x}^T)}
\end{equation}
Any SPC constraint with parity check coefficient vector $\vect{a}$ is in $\imagesop{\vect{a}}$. The first
of the following pair of lemmas states that $\imagesop{\vect{a}}$ is a subspace of 
$\setfieldqn$ and the second one investigates the relation between two such subspaces. 

\begin{lemma} \label{imagesubspacelemma} For any nonzero parity check coefficient vector $\vect{a}$ in $\fieldqn$,
 $\imagesop{\vect{a}}$ is a $q-1$ dimensional subspace of $\setfieldqn$.
\end{lemma}

\begin{lemma}\label{subspacerelationlemma} For any two nonzero parity check coefficient vectors $\vect{a},\vect{b}\in\fieldqn$ 
\begin{eqnarray}
\exists \alpha \in \fieldq :  \vect{a}=\alpha \vect{b} &\implies& \imagesop{\vect{a}}=\imagesop{\vect{b}}\\  
\nexists \alpha \in \fieldq : \vect{a}=\alpha \vect{b}&\implies& \imagesop{\vect{a}} \perp\imagesop{\vect{b}}
\end{eqnarray}
\end{lemma}

The proofs of this lemmas are given in Appendix \ref{imagesubspacelemmaproof} and Appendix \ref{subspacerelationlemmaproof} respectively. 

Lemma \ref{subspacerelationlemma} suggests that $\setfieldqn$ can be decomposed into orthogonal subspaces
by using a sufficient number of parity check coefficient vectors which are all pairwise linearly
independent. Fortunately, we can borrow such a set of parity check coefficient vectors
from  coding theory as explained by the following theorem.   

\begin{theorem}\label{decompositiontheorem}
  There exists a set $\set{H}$ of pairwise linearly independent parity check vectors in $\fieldq$ of length $N$ such that
\begin{equation}
\bigoplus_{\vect{a}\in\set{H}} \imagesop{\vect{a}} = \setfieldqn  \textrm{.} \label{orthogonaldecompositionequation}
\end{equation}
where $\bigoplus$ denotes orthogonal direct summation. 
\end{theorem}
\begin{proof}
For all nonzero $\vect{a}$, $\imagesop{\vect{a}}$ is a subspace 
of $\setfieldqn$. Orthogonal direct sum of subspaces is again 
a subspace of $\setfieldqn$. Therefore, we can complete
the proof by finding an $\set{H}$ which makes
\begin{equation}
   \dim \bigoplus_{\vect{a}\in\set{H}} \imagesop{\vect{a}} = \dim \setfieldqn \textrm{.}
\end{equation}

  Let the elements of $\set{H}$ be selected by \emph{transposing the columns} of 
 the parity check matrix of the Hamming code in $\fieldq$ with $N$ rows. 
 It is known from coding theory that the parity check matrix of 
 such a Hamming code consists of $\frac{q^{N}-1}{q-1}$ columns all
of which are pairwise linearly independent \cite{blahut}. Therefore,
$\set{H}$ contains $\frac{q^{N}-1}{q-1}$ \emph{pairwise linearly 
independent} vectors.  Since these vectors are pairwise linearly 
independent, for any $\vect{a},\vect{b} \in \set{H}$
\begin{equation} 
  \imagesop{\vect{a}} \perp \imagesop{\vect{b}}  
\end{equation}
due to Lemma \ref{subspacerelationlemma}. Hence, 
\begin{equation}
  \dim \bigoplus_{\vect{a}\in \set{H}} \imagesop{\vect{a}} =\sum_{\vect{a}\in \set{H}} \dim \imagesop{\vect{a}} \textrm{,}
\end{equation}
 since these subspace are all orthogonal. $\imagesop{\vect{a}}$ is a $q-1$ dimensional subspace due to Lemma \ref{imagesubspacelemma}.
Therefore,  
\begin{eqnarray}
  \dim \bigoplus_{\vect{a}\in \set{H}} \imagesop{\vect{a}} &=&\sum_{\vect{a}\in \set{H}} (q-1) \nonumber \\
  &=&\setsize{H}(q-1) \nonumber \\
  &=&q^{N}-1 \nonumber \\
  &=&\dim \setfieldqn \textrm{,}
\end{eqnarray}
which completes the proof. 
\end{proof}
\section{The Canonical Factorization}
\begin{corollary}\label{canoniccorollary}\emph{(The fundamental result of the thesis:)} Any multivariate pmf in $\setfieldqn$ 
can be expressed as a product of functions that depend on a linear combination of their arguments. 
\end{corollary}
\begin{proof}
Let $\set{H}$ be set of parity check vectors satisfying (\ref{orthogonaldecompositionequation}), existence
of which is guaranteed by Theorem \ref{decompositiontheorem}.
Let the vectors in $\set{H}$ be enumerated as $\vect{a}_1,\vect{a}_2,\ldots,\vect{a}_{\setsize{H}}$.
Then any $p(\vect{x}) \in \setfieldqn$ can be expressed as 
\begin{equation}
  p(\vect{x})=\ssum_{i=1}^{\setsize{H}} p_i(\vect{x}) 
\end{equation}
where $p_i(\vect{x})$ is the projection of $p(\vect{x})$ onto
$\imagesop{\vect{a}_i}$. Since $p_i(\vect{x})$ is in $\imagesop{\vect{a}_i}$,
there exist an $r_i(x)\in\setfieldq$ such that 
\begin{equation}
  p_i(\vect{x}) =\n{\fieldqn}{ r_i(\vect{a}_i\vect{x}^T)} \textrm{.}
\end{equation}
Then $p(\vect{x})$ can be expressed as
\begin{equation}
  p(\vect{x})=\ssum_{i=1}^{\setsize{H}} \n{\fieldqn}{  r_{i}(\vect{a}_i\vect{x}^{T})} \textrm{.}
\end{equation}
Employing the definition of addition in $\setfieldqn$ yields the desired factorization. 
\begin{eqnarray}
  p(\vect{x})&=&\n{\fieldqn}{\prod_{i=1}^{\setsize{H}}r_{i}(\vect{a}_i\vect{x}^{T})} \\
  &=&\frac{1}{\gamma} \prod_{i=1}^{\setsize{H}}r_{i}(\vect{a}_i\vect{x}^{T}) \textrm{,}
\end{eqnarray}
where $\gamma$ is equal to $\sum_{\forall \vect{i} \in \fieldqn}\prod_{i=1}^{\setsize{H}}r_{i}(\vect{a}_i\vect{x}^{T})$.
\end{proof}

\begin{definition}\label{canonicdefinition}
\emph{The canonical factorization:} A factorization of a multivariate pmf is called the canonical factorization of the 
pmf if all factor functions are SPC factors and parity check coefficient vectors of all SPC factors
are pairwise linearly independent.  
\end{definition}

The canonical factorization of a multivariate pmf in $\setfieldqn$ can be obtained by projecting 
the pmf onto the subspaces $\imagesop{\vect{a}_i}$ for $\vect{a}_i \in \set{H}$. In order to compute
this projection a set of orthonormal basis pmfs for $\imagesop{\vect{a}_i}$ is required. 
We can derive such a set of orthonormal basis pmfs from the orthonormal basis
pmfs for $\setfieldq$  given in Section \ref{orthobasispmfs} by using the first part of Lemma \ref{spcinnerproduct}. 
The inner product of two SPC constraints $\n{\fieldqn}{s_j(\vect{a}_i\vect{x}^{T})}$ and $\n{\fieldqn}{s_k(\vect{a}_i\vect{x}^{T})}$
which are derived from $s_j(x)$ and $s_k(x)$ defined in (\ref{orthoinpmf}) is
\begin{equation}
  \innerproduct{\n{\fieldqn}{s_j(\vect{a}_i\vect{x}^{T})}}{\n{\fieldqn}{s_k(\vect{a}_i\vect{x}^{T})}}=q^{N-1}\innerproduct{s_j(x)}{s_k(x)} 
\end{equation}
due to Lemma  \ref{spcinnerproduct}. Consequently,
\begin{equation}
  \innerproduct{\n{\fieldqn}{s_j(\vect{a}_i\vect{x}^{T})}}{\n{\fieldqn}{s_k(\vect{a}_i\vect{x}^{T})}}=
  \left\{\begin{array}{cc}q^{N-1}, & k=j \\ 0 \end{array} \right. 
\end{equation}
Therefore, the set 
given below is a set of orthonormal basis pmfs 
for $\imagesop{\vect{a}_i}$.
\begin{equation}
\left\{q^{-\frac{N-1}{2}}\boxtimes \n{\fieldqn}{s_1(\vect{a}_i\vect{x}^{T})},q^{-\frac{N-1}{2}}
\boxtimes \n{\fieldqn}{s_2(\vect{a}_i\vect{x}^{T})},\ldots, 
  q^{-\frac{N-1}{2}}\boxtimes \n{\fieldqn}{s_{q-1}(\vect{a}_i\vect{x}^{T})}\right\} 
\end{equation}
Then the projection of $p(\vect{x})$ onto $\imagesop{\vect{a}_i}$, which is denoted by $\n{\fieldqn}{r_{i}(\vect{a}_i\vect{x}^T)}$,
can be obtained as 
\begin{equation}
  \n{\fieldqn}{r_{i}(\vect{a}_i\vect{x}^T)}=\ssum_{j=1}^{q-1}q^{-(N-1)}\boxtimes
  \innerproduct{\n{\fieldqn}{s_j(\vect{a}_i\vect{x}^{T})}}{p(\vect{x})}\boxtimes
  \n{\fieldqn}{s_j(\vect{a}_i\vect{x}^{T})} \textrm{.} \label{theanalysisequation}
\end{equation}
Moreover, due to the linearity of the mapping $\sopsub{\vect{a}_i}{.}$
\begin{equation}
  r_i(x)=\ssum_{j=1}^{q-1}q^{-(N-1)}\boxtimes
  \innerproduct{\n{\fieldqn}{s_j(\vect{a}_i\vect{x}^{T})}}{p(\vect{x})}\boxtimes s_j(x) \textrm{.} \label{theanalysisequation2}
\end{equation}

\begin{example}
Suppose that we are required to find the canonical factorization of $p_{2}(x_0,x_1)$ given
in Example \ref{spcexample1}. We can decompose $\setfieldnn{3}{2}$  into orthogonal subspaces
with a set $\set{H}$ containing $\frac{3^{2}-1}{3-1}=4$ pairwise linearly independent 
parity check  vectors of length two. Such an $\set{H}$
can be selected as 
\begin{equation}
  \set{H}=\left\{[1,0],[0,1],[1,1],[1,2]\right\}
\end{equation} 
The subspaces of $\setfieldnn{3}{2}$  based on these parity check vectors are
\begin{eqnarray}
  \imagesop{[1,0]}&=&\left\{p(x_0,x_1)=\frac{1}{3}r(x_0)=
    \frac{1}{3}\left[\begin{array}{ccc}r(0) &r(0)&r(0)\\r(1)&r(1)&r(1)\\r(2)&r(2)&r(2) \end{array}\right]: 
    r(x)\in\setfield{3} \right\} \nonumber \\
  \imagesop{[0,1]}&=&\left\{p(x_0,x_1)=\frac{1}{3}r(x_1)=
    \frac{1}{3}\left[\begin{array}{ccc}r(0) &r(1)&r(2)\\r(0)&r(1)&r(2)\\r(0)&r(1)&r(2) \end{array}\right]: 
    r(x)\in\setfield{3} \right\} \nonumber 
\end{eqnarray}
\begin{eqnarray}
  \imagesop{[1,1]}&=&\left\{p(x_0,x_1)=\frac{1}{3}r(x_0+x_1)=
    \frac{1}{3}\left[\begin{array}{ccc}r(0) &r(1)&r(2)\\r(1)&r(2)&r(0)\\r(2)&r(1)&r(0) \end{array}\right]: 
    r(x)\in\setfield{3} \right\} \nonumber \\
  \imagesop{[1,2]}&=&\left\{p(x_0,x_1)=\frac{1}{3}r(x_0+2x_1)=
    \frac{1}{3}\left[\begin{array}{ccc}r(0) &r(1)&r(2)\\r(2)&r(0)&r(1)\\r(1)&r(2)&r(0) \end{array}\right]: 
    r(x)\in\setfield{3} \right\} \nonumber 
\end{eqnarray}

Let the projections of $p_{2}(x_0,x_1)$ onto these subspaces be denoted with 
$\frac{1}{3}r_1(x_0)$, $\frac{1}{3}r_2(x_1)$, \mbox{$\frac{1}{3}r_3(x_0+x_1)$}, and $\frac{1}{3}r_4(x_0+2x_1)$ respectively.
These pmfs can be computed using  (\ref{theanalysisequation2}) as
\begin{eqnarray}
  r_1(x)=\left\{\begin{array}{cc}0.2, &x=0 \\0.4, &x=1 \\ 0.4, &x=2 \end{array} \right. &,&
  r_2(x)=\left\{\begin{array}{cc}\frac{1}{3}, &x=0 \\\frac{1}{3}, &x=1 \\ \frac{1}{3}, &x=2 \end{array} \right. \nonumber \\
  r_3(x)=\left\{\begin{array}{cc}0.4, &x=0 \\0.5, &x=1 \\ 0.1, &x=2 \end{array} \right. &,&
  r_4(x)=\left\{\begin{array}{cc}0.3, &x=0 \\0.6, &x=1 \\ 0.1, &x=2 \end{array} \right. \nonumber  \textrm{.}
\end{eqnarray}
 Finally, it can be verified that 
\begin{equation}
  p_{2}(x_0,x_1)=\frac{1500}{157}r_{1}(x_0)r_{2}(x_1)r_{3}(x_0+x_1)r_4(x_0+2x_1) \nonumber \textrm{.}
\end{equation}
\end{example}
\chapter[PROPERTIES AND SPECIAL CASES OF THE  CANONICAL\\ FACTORIZATION]{PROPERTIES AND SPECIAL CASES OF \\THE  CANONICAL FACTORIZATION\label{specialcasechapter}}

\section{Introduction}

The canonical factorization deserves its name by possessing some 
important properties. This chapter explains these properties first and then
some special cases of the canonical factorization is derived. These
special cases will be important while applying the canonical 
factorization to communication theory problems in Chapter \ref{applicationchapter}. 
This chapter begins with introducing a matrix notation to represent
local functions in Section \ref{localfunctionrepsection}. Then it is shown
in Section \ref{ultimatenesssection} that the canonical factorization 
is the ultimate factorization possible. 
The uniqueness of the canonical factorization is 
explained Section \ref{uniquenesssection}. The canonical 
factorization of pmfs with known alternative factorizations
is derived in Section \ref{alternativefactorizationsection}.
This chapter ends with deriving the canonical factorization 
of the joint pmf a random vector obtained by linear transformation
of another random vector.

\section{Representation of local functions\label{localfunctionrepsection}}

In the rest of the thesis we deal frequently with local functions.
We adopt a matrix notation to indicate the variables that a factor function 
depends. We use $\fieldq$-valued diagonal matrices such that some of their entries on the main
diagonal are $1$ and the rest are all $0$. For instance, 
\begin{equation}
  p(\vect{x})=p(\vect{x}\vect{D})
\end{equation}  
indicates that the pmf $p(\vect{x})$ depends on only to the components of $\vect{x}$
associated with a $1$ on the diagonal of the matrix $\vect{D}$. We call such matrices
 dependency matrices. Some special dependency matrices we use in the thesis are 
$\vect{E}_i$, $\vect{I}$, and $\vect{O}$. $\vect{E}_i$ denotes the dependency matrix
with a $1$ only on the $i^{th}$ entry of its diagonal. The other two matrices
are the identity matrix and the all-zeros matrix respectively. 

A local pmf is orthogonal to  some SPC constraints as shown by the following lemma.
This lemma is  quite useful not only in this chapter but also 
in Chapter \ref{markovchapter}.

\begin{lemma} \label{propertyorthogonallemma} For any $p(\vect{x})\in\setfieldqn$, any nonzero $\vect{a}\in \fieldqn$, 
and any dependency matrix $\vect{D}$
\begin{equation}
  p(\vect{x})=p(\vect{x}\vect{D}) \land \vect{a}\vect{D}\neq \vect{a} \implies p(\vect{x}) \perp \imagesop{\vect{a}} \textrm{.}
\end{equation}
\end{lemma}

The proof is given in Appendix  \ref{propertyorthogonallemmaproof}.

\section{Ultimateness of the canonical factorization\label{ultimatenesssection}}

The ultimate goal of any mathematical factorization operation is to factor the 
mathematical object to its most basic building blocks. For instance, the goal 
of integer factorization is to express a natural number as a product 
of prime numbers. Similarly, the ultimate goal of polynomial factorization
is to express a polynomial as a product of irreducible polynomials. 

In the case of factoring a strictly positive multivariate pmf into  strictly positive factor functions, it is difficult to 
set an ultimate goal or to describe 
the most basic building blocks of multivariate pmfs. 
Since, any factor function in any factorization can still be expressed
as a product of other positive factor functions, a multivariate pmf 
can be factored arbitrarily in many different ways and  the factorization operation 
can continue indefinitely.  In this aspect, factoring a strictly positive 
pmf is similar to trying to factor a real number.

However, not every factorization is useful in practice. A factorization of a multivariate pmf
is useful if it expresses the pmf as a product of \emph{local} functions. Therefore, 
it is reasonable to continue to factor  a multivariate  pmf if any factor function
can still be expressed as a product of \emph{more local} factor functions. For instance,
let a factor function $\phi(\vect{x}\vect{D})$ of $p(\vect{x})$  be
expressed as 
\begin{equation}
  \phi(\vect{x}\vect{D})=\phi_1(\vect{x}\vect{D}_1)\phi_2(\vect{x}\vect{D}_2) \nonumber
\end{equation}
where $\vect{D}_i\neq \vect{D}$ but $\vect{D}\vect{D}_i=\vect{D}_i$ for $i=1,2$.
Since $\phi_1(.)$ and $\phi_2(.)$ have less number of arguments than $\phi(.)$ has,
the ultimate factorization of $p(\vect{x})$ should contain 
the product $\phi_1(\vect{x}\vect{D}_1)\phi_2(\vect{x}\vect{D}_2)$
rather than  $\phi(\vect{x}\vect{D})$. 
In this point of view, the canonical factorization is the ultimate 
factorization that one can achieve as stated by the following theorem.

\begin{theorem} An  SPC factor function with a nonzero norm cannot be factored further to functions
having less number of arguments. 
\end{theorem}
\begin{proof}
 Assume that an SPC constraint, $\n{\fieldqn}{r(\vect{a}\vect{x}^T)}$, with a nonzero norm can be factored to functions having less number of
arguments. In other words, assume that $\n{\fieldqn}{r(\vect{a}\vect{x}^T)}$ can be expressed as  
\begin{eqnarray}
  \n{\fieldqn}{r(\vect{a}\vect{x}^T)}& =& \phi_1(\vect{x}\vect{D}_1)\phi_2(\vect{x}\vect{D}_2)\nonumber \\
&=& \n{\fieldqn}{\phi_1(\vect{x}\vect{D}_1)}\boxplus  \n{\fieldqn}{\phi_2(\vect{x}\vect{D}_2)}\label{theult} \textrm{,}
\end{eqnarray}
where $\vect{D}_1$ and $\vect{D}_2$ are such dependency matrices that $\vect{a}\vect{D}_1\neq \vect{a}$ and $\vect{a}\vect{D}_2\neq \vect{a}$.  
$\n{\fieldqn}{\phi_1(\vect{x}\vect{D}_1)}$ and $\n{\fieldqn}{\phi_2(\vect{x}\vect{D}_2)}$ are orthogonal to 
$\n{\fieldqn}{r(\vect{a}\vect{x}^T)}$ due to Lemma  \ref{propertyorthogonallemma}. Then (\ref{theult})
is only possible if 
\begin{equation}
  \n{\fieldqn}{r(\vect{a}\vect{x}^T)}=\n{\fieldqn}{\phi_1(\vect{x}\vect{D}_1)}=
  \n{\fieldqn}{\phi_2(\vect{x}\vect{D}_2)}=\theta(\vect{x}) \textrm{,} \nonumber
\end{equation}
which is a contradiction completing the proof. 
\end{proof}

\section{Uniqueness of the  canonical factorization\label{uniquenesssection}}

Recall that we need a set $\set{H}$ composed of $\frac{q^{N}-1}{q-1}$ pairwise linearly independent
vectors in $\fieldqn$ to derive the canonical factorization of  a pmf in $\setfieldqn$. 
There are $2^{N}-1$ nonzero vectors in $\fieldn{2}$ all of which are pairwise linearly independent. 
Hence, the set $\set{H}$ should contain all the nonzero vectors in $\fieldn{2}$. Consequently, 
the set $\set{H}$ required in the derivation of the canonical factorization 
of a pmf in $\setfieldn{2}$ is unique.  Moreover, the canonical 
factorization obtained from such a set $\set{H}$ is also unique. 

If the $\fieldq$ is not the binary field then there are $q^{N}-1$ nonzero vectors in $\fieldqn$.
Hence, we can have more than one distinct sets which contain $\frac{q^{N}-1}{q-1}$ 
pairwise linearly independent vectors in $\fieldqn$ if $q$ is not equal to two. Let $\set{H}_1=\{\vect{a}_1,\vect{a}_2,\ldots,\vect{a}_M\}$
and $\set{H}_2=\{\vect{b}_1,\vect{b}_2,\ldots,\vect{b}_M\}$ be two distinct sets containing $M=\frac{q^{N}-1}{q-1}$ pairwise linearly independent
vectors in $\fieldqn$. Using these two sets we can obtain two different canonical factorizations
of a multivariate pmf $p(\vect{x})\in \setfieldqn$ as in
\begin{eqnarray}
  p(\vect{x})&=&\n{\fieldqn}{\prod_{i=1}^{M}r_i(\vect{a}_i\vect{x}^{T})}\label{canonic1} \textrm{,}\\
  p(\vect{x})&=&\n{\fieldqn}{\prod_{i=1}^{M}t_i(\vect{b}_i\vect{x}^{T})} \label{canonic2} \textrm{,}
\end{eqnarray}
where $r_{i}(\vect{a}_i\vect{x}^{T})$ and $t_{i}(\vect{b}_i\vect{x}^{T})$ denote the projections 
of $p(\vect{x})$ onto $\imagesop{\vect{a}_i}$ and $\imagesop{\vect{b}_i}$ respectively. Notice
that any $\vect{a}_i$ in $\set{H}_1$ is definitely linearly dependent with one of the 
$\vect{b}_i$ vectors in $\set{H}_2$. In other words for any $\vect{a}_i\in\set{H}_1$ there exist a $\vect{b}_j \in \set{H}_2$
such that 
\begin{equation}
  \vect{b}_j=\alpha \vect{a}_i \textrm{.}
\end{equation}
Consequently, $\imagesop{\vect{a}_i}$ is equal to $\imagesop{\vect{b}_j}$ due to Lemma \ref{subspacerelationlemma}. 
Therefore, the projection of $p(\vect{x})$ onto these same subspaces should also be equal, i.e.,
\begin{equation}
  r_i(\vect{a}_i\vect{x}^{T})=t_j(\vect{b}_j\vect{x}^{T}) \textrm{,} 
\end{equation}
which means that the factorizations in (\ref{canonic1}) and (\ref{canonic2}) are essentially the
same factorization although they appear different. Since different sets of  parity check coefficient vectors
leads to the same canonical factorization, we can conclude that the canonical factorization of a given pmf is unique. 
Since the  selection of the vectors in $\set{H}$ does not affect the resulting 
canonical factorization, in 
the rest of the thesis we use $\set{H}$ to denote any set containing
$\frac{q^N-1}{q-1}$ pairwise linearly independent vectors in $\fieldqn$.

\section{The canonical factorization of pmfs with  alternative factorizations \label{alternativefactorizationsection}}

In the most general case, the canonical factorization of a multivariate pmf in $\setfieldqn$ is composed
of $\setsize{H}$ SPC factors. However, for some special pmfs some of these $\setsize{H}$ SPC factors are essentially 
constants. For these pmfs  less than $\setsize{H}$ SPC factors  may suffice to  express
the canonical factorization. 

The first group of these special types of pmfs consists of pmfs which depend on only a subset of their
arguments. The canonical factorization of these types of pmfs is investigated in the following lemma.

\begin{lemma}\label{localfactorizationlemma}
  Let $\set{D}$ be a subset of $\set{H}$ defined for a dependency matrix $\vect{D}$ as 
  \begin{equation}
    \set{D}\triangleq \{\vect{a}_i\in \set{H}: \vect{a}_i\vect{D}=\vect{a}_i\} \textrm{.}
  \end{equation}
  The canonical factorization of a multivariate pmf $p(\vect{x})\in \setfieldq$ is in the form of 
  \begin{equation}
    p(\vect{x})= \n{\fieldqn}{\prod_{\vect{a}_i\in \set{D}}r_i(\vect{a}_i\vect{x}^T)} \label{localfactorizationthm1}
  \end{equation}
  if and only if 
  \begin{equation}
    p(\vect{x})=p(\vect{x}\vect{D}) \textrm{.}
  \end{equation}
\end{lemma}
\begin{proof}
  Due to Theorem \ref{decompositiontheorem} 
any pmf in $\setfieldqn$ can be expressed as 
\begin{eqnarray}
  p(\vect{x}) &=&\ssum_{\vect{a}_i\in \set{H}}\n{\fieldqn}{r_{i}(\vect{a}_i\vect{x}^{T})} \\
  &=&\ssum_{\vect{a}_i\in \set{D}}\n{\fieldqn}{r_{i}(\vect{a}_i\vect{x}^{T})} \boxplus 
  \ssum_{\vect{a}_i\in \set{H}\setminus \set{D}}\n{\fieldqn}{r_{i}(\vect{a}_i\vect{x}^{T})}  \textrm{,}
\end{eqnarray}
where $\n{\fieldqn}{r_{i}(\vect{a}_i\vect{x}^{T})}$ is the projection of 
$p(\vect{x})$ onto $\imagesop{\vect{a}}$. But $p(\vect{x})$ is orthogonal
to $\imagesop{\vect{a}_i}$ for $\vect{a}_i \in \set{H}\setminus\set{D}$  due to 
Lemma \ref{propertyorthogonallemma}. Hence,
\begin{equation}
\n{\fieldqn}{r_{i}(\vect{a}_i\vect{x}^{T})}=\theta(\vect{x}) \textrm{,} 
\end{equation}
for $\vect{a}_i \in \set{H}\setminus\set{D}$. Consequently,
\begin{eqnarray}
  p(\vect{x}) &=&\ssum_{\vect{a}_i\in \set{D}}\n{\fieldqn}{r_{i}(\vect{a}_i\vect{x}^{T})}\\
  &=&\n{\fieldqn}{\prod_{\vect{a}_i \in \set{D}}r_{i}(\vect{a}_i\vect{x}^{T})} \textrm{,}
\end{eqnarray}
which is the desired factorization to prove the theorem in the forward direction. 

The proof in the backward direction is straight forward. If $p(\vect{x})$ can 
be factored as in (\ref{localfactorizationthm1}) then
\begin{eqnarray}
  p(\vect{x}\vect{D})  &=& \n{\fieldqn}{\prod_{\vect{a}_i\in \set{D}}r_i(\vect{a}_i\vect{x}^T)} \Evalat{\vect{x}=\vect{x}\vect{D}} \\
  &=& \n{\fieldqn}{\prod_{\vect{a}_i\in \set{D}}r_i(\vect{a}_i\vect{D}^T\vect{x}^T)} \textrm{.}
\end{eqnarray}
Since $\vect{D}$ is symmetric and $\vect{a}_i\vect{D}=\vect{a}_i$ for $\vect{a}_i\in \set{D}$,
\begin{eqnarray}
  p(\vect{x}\vect{D})&=&\n{\fieldqn}{\prod_{\vect{a}_i\in \set{D}}r_i(\vect{a}_i\vect{x}^T)}  \\
  &=&p(\vect{x})  \textrm{,}
\end{eqnarray}
which completes the proof. 
\end{proof}

This lemma tells in practice that any pmf satisfying the relation $p(\vect{x})=p(\vect{x}\vect{D})$ can be expressed
as a product of $\setsize{D}$ SPC factors  rather than 
 $\setsize{H}$ SPC factors. Moreover, 
parity check coefficient vectors of these SPC factors satisfy the relation $\vect{a}=\vect{a}\vect{D}$.
We do not need to compute the projection of  $p(\vect{x})$ onto $\imagesop{\vect{a}}$ if $\vect{a}$ is not
in $\set{D}$, since the result of that projection would be $\theta(\vect{x})$ 
definitely.

The next theorem investigates the canonical factorization of pmfs with known alternative factorizations. 

\begin{theorem}\label{localfactorizationtheorem} If a multivariate pmf $p(\vect{x})\in \setfieldqn$ can be factored
as   
  \begin{equation}
    p(\vect{x})= \n{\fieldqn}{\prod_{j=1}^{K}\phi_{j}(\vect{x}\vect{D}_j) } \textrm{,}
  \end{equation}
where $\vect{D}_1$, 
$\vect{D}_2$, $\ldots$, $\vect{D}_K$ are dependency matrices then 
the canonical factorization of $p(\vect{x})$ is in the
form of 
\begin{equation}
  p(\vect{x})= \n{\fieldqn}{
\prod_{j=1}^{K} \prod_{\vect{a}_i\in \set{D}_{j}}r_i(\vect{a}_i\vect{x}^T)} \label{localfactorizationcor1} \textrm{,}
\end{equation}
where $\set{D}_j$ is the subset of $\set{H}$ given by
\begin{equation}
  \set{D}_j \triangleq \{\vect{a}_i\in \set{H}:  \vect{a}_i\vect{D}_j=\vect{a}_i\} 
  \textrm{.}
\end{equation}
\end{theorem}
\begin{proof}
This theorem is actually a  direct consequence of Lemma \ref{localfactorizationlemma}. 
Let $t_j(\vect{x})\in\setfieldqn$ be 
\begin{equation}
  t_j(\vect{x}) \triangleq \n{\fieldqn}{\phi_{j}(\vect{x}\vect{D}_j)} \textrm{.}
\end{equation}
Since $t_j(\vect{x})$ is equal to $t_j(\vect{x}\vect{D}_j)$, 
\begin{equation}
  t_j(\vect{x})=\ssum_{\vect{a}_i\in \set{D}_j} r_{i}(\vect{a}_i\vect{x}^{T}) \textrm{,}
\end{equation}
 due to Lemma \ref{localfactorizationlemma}. Then $p(\vect{x})$ is
\begin{eqnarray}
  p(\vect{x})&=&\ssum_{j=1}^{K}t_j(\vect{x})\\
  &=&\ssum_{j=1}^{K}\ssum_{\vect{a}_i \in \set{D}_j}r_{i}(\vect{a}_i\vect{x}^{T})\\
&=&\n{\fieldqn}{
\prod_{j=1}^{K} \prod_{\vect{a}_i\in \set{D}_{j}}r_i(\vect{a}_i\vect{x}^T)}
  \textrm{,}
\end{eqnarray}
which completes the proof. 
\end{proof}

The practical consequence of this theorem is that the canonical factorization of a pmf with an alternative
factorization can be derived by obtaining the canonical factorization of the 
factor functions in the alternative factorization. This approach significantly simplifies the derivation 
of the canonical factorization for such pmfs and extensively used in Chapter \ref{applicationchapter}.
\section{The effect of reversible linear transformations on the canonical factorization\label{reversiblesection}}
If two random vectors are related with a \emph{reversible} linear transformation then the canonical 
factorization of the pmf of the one of random vectors can be derived from 
the canonical factorization of the other random vector's pmf. 
Let $\vect{X}$ be an $\fieldqn$-valued random vector distributed with  $p(\vect{x})\in\setfieldqn$.
Moreover, let $\vect{Y}$ be another $\fieldqn$-valued random vector which 
is related to $\vect{X}$ as in 
\begin{equation}
  \vect{Y}=\vect{X}\vect{B}
\end{equation}
where $\vect{B}$ is an  \emph{reversible}  matrix in $\fieldq^{N\times N}$. 
Since $\vect{B}$ is reversible, for each $\vect{y}\in\fieldqn$ there is 
one and only one $\vect{x}\in\fieldqn$ vector satisfying
  $\vect{y}=\vect{x}\vect{B}$, 
which  is given by $\vect{x}=\vect{y}\vect{B}^{-1}$. Hence,
\begin{eqnarray}
  \Pr\{\vect{Y}=\vect{y}\}&=&\Pr\{\vect{X}=\vect{y}\vect{B}^{-1}\} \\
  &=&p(\vect{y}\vect{B}^{-1}) \textrm{.}
\end{eqnarray} 
If the canonical factorization of $p(\vect{x})$ is as given in 
\begin{equation}
  p(\vect{x}) = \prod_{\vect{a}_i \in \set{H}} r_{i}(\vect{a}_i\vect{x}^{T})
\end{equation}
then the canonical factorization of $\Pr\{\vect{Y}=\vect{y}\}$
is simply
\begin{eqnarray}
  \Pr\{\vect{Y}=\vect{y}\} &=& \prod_{\vect{a}_i \in \set{H}} r_{i}(\vect{a}_i\vect{x}^{T}) \evalat{\vect{x}=\vect{y}\vect{B}^{-1}} \\
  &=&\prod_{\vect{a}_i \in \set{H}} r_{i}\left(\vect{a}_i(\vect{B}^{-1})^{T}\vect{y}^{T}\right)  \label{factory}
\end{eqnarray}

If $\vect{B}$
was not reversible then $\Pr\{\vect{Y}=\vect{y}\}$ would be zero for some $\vect{y}$ vectors
in $\fieldqn$. Hence, $\Pr\{\vect{Y}=\vect{y}\}$ would not be a multivariate pmf in $\setfieldqn$ and consequently
we could not talk about the canonical factorization of  $\Pr\{\vect{Y}=\vect{y}\}$. 

An interesting question about the linear transformations of $\fieldqn$-valued random vectors 
might be whether there exists a linear transformation $\vect{K}$ for a random vector $\vect{X}$
such that the components of the vector $\vect{Y}=\vect{X}\vect{K}$ are statistically independent. 
If such a transformation exists it would prove useful in computing the marginal pmfs
of the components of the random vector $\vect{X}$.
The canonical factorization of $\Pr\{\vect{Y}=\vect{y}\}$ given in (\ref{factory}) provides a clue
to this question.  

\begin{theorem}\label{statisticalindependencetheorem} There exist a matrix $\vect{K}$ in $\fieldq^{N\times N}$ for an $\fieldqn$-valued
random vector $\vect{X}$ such that the components of the random vector $\vect{Y}$ given by
\begin{equation}
\vect{Y}=\vect{X}\vect{K}
\end{equation}
 are \textbf{{statistically independent}} if  the canonical factorization 
of the pmf of $\vect{X}$ is composed of at most $N$ SPC factors 
whose parity check coefficient vectors are all \textbf{{linearly independent}}.  
\end{theorem}
\begin{proof} The proof is constructive. Let $p(\vect{x})$ be the pmf of
$\vect{X}$ and the canonical factorization of $p(\vect{x})$ be denoted as
\begin{equation}
  p(\vect{x})= \n{\fieldqn}{\prod_{\vect{a}_i\in \set{K}} r_i(\vect{a}_i\vect{x}^{T})} 
\end{equation}
where $\set{K}$ is a subset of $\set{H}$ containing at most $N$ linearly independent
vectors. Let $\set{K}_c$ be a subset of $\set{H}$ such that it is a superset of $\set{K}$
and it contains exactly $N$ linearly independent vectors. Since $r_i(\vect{a}_i\vect{x}^{T})$
is equal to $\theta(\vect{x})$ for $\vect{a}_i \in \set{K}_c\setminus \set{K}$, 
the canonical factorization of 
$p(\vect{x})$ can also be expressed as
\begin{equation}
  p(\vect{x})= \n{\fieldqn}{\prod_{\vect{a}_i\in \set{K}_c} r_i(\vect{a}_i\vect{x}^{T})}
\end{equation}
We may define a matrix $\vect{K}_c$ whose rows are the elements of $\set{K}_c$. 
Using this matrix $\vect{K}_c$ the canonical factorization of $p(\vect{x})$
becomes
\begin{equation}
  p(\vect{x})= 
\n{\fieldqn}{\prod_{j=1}^{N}r_{i(j)}(\vect{f}_j\vect{K}_c\vect{x}^{T})} \textrm{,}
\end{equation}
where $\vect{f}_j$ is the $j^{th}$ canonical basis vector of $\fieldqn$ and $i(j)$ is the index of the vector $\vect{a}_i$ when
$\vect{a}_i=\vect{f}_j\vect{K}_c$. Then we may
define the matrix $\vect{K}$ as 
\begin{equation}
  \vect{K} \triangleq(\vect{K}_c^{-1})^{T}  \textrm{.}
\end{equation}

With this definition of $\vect{K}$, the canonical factorization of  $\Pr\{\vect{Y}=\vect{y}\}$ becomes 
\begin{eqnarray}
 \Pr\{\vect{Y}=\vect{y}\}&=& \n{\fieldqn}{\prod_{j=1}^{N}r_{i(j)}(\vect{f}_j\vect{K}_c\vect{x}^{T})} \Evalat{\vect{x}=\vect{y}\vect{K}^{-1}}\\
 &=&\n{\fieldqn}{\prod_{j=1}^{N}r_{i(j)}\left(\vect{f}_j\vect{K}_c(\vect{K}^{-1})^{T}\vect{y}^{T}\right)}\\
 &=&\n{\fieldqn}{\prod_{j=1}^{N}r_{i(j)}(\vect{f}_j\vect{y}^{T})}\\
&=&\n{\fieldqn}{\prod_{j=1}^{N}r_{i(j)}(y_j)}\textrm{,}
\end{eqnarray}
where $y_j$ is the $j^{th}$ component of $\vect{y}$. Since $\Pr\{\vect{Y}=\vect{y}\}$ is separable, the components
of $\vect{Y}$ are statistically independent. Moreover, the distribution of the $j^{th}$ component of $\vect{Y}$ is 
simply
\begin{equation} 
  \Pr\{Y_j=y\}= r_{i(j)}(y) \textrm{.}
\end{equation} 
\end{proof}

In  the general case, the marginal pmfs of the components of an $\fieldqn$-valued random vector $\vect{X}$ can be computed 
via the marginalization sum whose complexity is $q^{N}$. If the multivariate pmf of $\vect{X}$ obeys the condition 
imposed in Theorem \ref{statisticalindependencetheorem} 
then $\vect{X}$ can be related to $\vect{Y}$, whose components are statistically independent, as
\begin{equation} 
\vect{X}=\vect{Y}\vect{K}^{-1} \textrm{.}
\end{equation}
This means that any component of $\vect{X}$ is equal to a linear combination of $N$ statistically independent random variables. 
Hence, the marginal pmfs of the components of  $\vect{X}$ can be computed via $N-1$ circular convolutions over $\fieldq$
instead of the marginalization sum. Consequently, the complexity of computing a single marginal pmf is $Nq^2$ and
the complexity of computing all marginal pmfs is $N^2q^2$ instead of $q^N$ 
for such random vectors\footnote{These complexities can be reduced even more
to $Nq\log_2 q$ and $N^2q\log_2 q$ by computing the convolutions via FFT if $\fieldq$ is an extension field of 
the binary field \cite{reducedcomplexity}.  }.

\chapter{EMPLOYING CHANNEL DECODERS FOR INFERENCE TASKS  BEYOND DECODING \label{decodingchapter}}

\section{Introduction}
This chapter explains subjectively the most important consequence of the 
canonical factorization which allows the decoders of the 
linear error correction codes to be utilized in other inference tasks. 

This chapter starts with an overview of channel decoders. Then how a maximum likelihood (ML)
decoder can be used to maximize a multivariate pmf is explained.  It is shown in Section \ref{marginalizationsection}
that symbolwise decoders can be employed to marginalize multivariate pmfs. Section \ref{universalitydualhamming}
highlights that the decoders of the dual Hamming code can be used as universal inference machines. 
Special cases are analyzed in Section \ref{specialdecoders}. The material presented in this 
chapter is summarized with graphical models in \ref{graphicalrepresentation}. This chapter ends with 
explaining the possible applications of employing channel decoders for inference tasks beyond decoding.

\section{An overview of channel decoders\label{decoderoverview}}

A channel decoder is specified by a code and a channel through which
the coded symbols are transmitted.  
A code $\set{C}$ over a finite field $\fieldq$ of length $L$ is defined 
as a subset of $\fieldql$. The code is called a \emph{linear code}
if  $\set{C}$ is a subspace of $\fieldql$. For linear codes
there exists a matrix $\vect{H}$ which satisfies
\begin{equation} 
  \vect{H}\vect{x}^{T}=\vect{0} \quad \forall \vect{x} \in \set{C} \textrm{.}
\end{equation}
The matrix $\vect{H}$ is called the parity check matrix of the code. 

A channel is a  system which maps a $\fieldq$-valued symbol to
an element of the output alphabet in a probabilistic manner \footnote{This definition of
channel includes the modulator when necessary. }. 
We assume that the channel decoders used in the rest of this
chapter are designed for a specific channel. This channel 
relates the inputs to the outputs via the following relation
\begin{equation}
  \vect{Y}_i=\vect{s}(X_i)+\vect{Z}_i  \label{channelmodel} \virgul
 \end{equation}
where $\vect{Z}_{i}$ is a noise vector consisting of independent, zero-mean, 
real Gaussian random variables with unit variance and $\vect{s}(.)$ denotes 
the simplex mapping as defined in (\ref{simplexmod}). 
The likelihood function, which is a conditional probability density function of a continuous random vector, of this channel
is 
\begin{eqnarray}
  f_{\vect{Y}_i|X_i}\{\vect{Y}_i=\vect{y}_i|X_i=x_i\}&\propto&\exp\left(-\frac{1}{2}\norm{\vect{y}_i-\vect{s}(x_i)}^{2} \right) \\
  &\propto&\pseudoinvi{L}{\vect{y}_i} \textrm{.}
\end{eqnarray}
The reasoning behind the selection of
this channel model is explained in Section \ref{simplerchannel}.

Let $\vect{X}=[X_1,X_2,\ldots,X_L]$ denote a codeword belonging to the  code $\set{C}$
 and $\vect{Y}=[\vect{Y}_1,\vect{Y}_2,\ldots,\vect{Y}_L]$
denote the output of the channel when $\vect{X}$ is transmitted  through this channel. 
If all codewords are equally likely then the a posteriori probability  (APP) of 
$\vect{X}$ is 
\begin{eqnarray}
  \Pr\{\vect{X}=\vect{x}|\vect{Y}=\vect{y}\}&=&\n{\fieldql}{\indicator{C}{\vect{x}}\prod_{i=1}^{L}f_{\vect{Y}_i|X_i}\{\vect{Y}_i=\vect{y}_i|X_i=x_i\} }\\
&=&\n{\fieldql}{\indicator{C}{\vect{x}}\prod_{i=1}^{L}\pseudoinvi{L}{\vect{y}_i}} \textrm{,}
\end{eqnarray}
where $\vect{x}=[x_1,x_2,\ldots,x_L]$, $\vect{y}=[\vect{y}_1,\vect{y}_2,\ldots,\vect{y}_L]$, and $\indicator{C}{.}$ denotes
the indicator function i.e.,
\begin{equation}
  \indicator{C}{\vect{x}}\triangleq \left\{\begin{array}{ll} 1, & \vect{x}\in \set{C} \\ 0, & \vect{x}\notin \set{C}  \end{array} \right. \textrm{.}
\end{equation}
If the code $\set{C}$ is a linear code with the parity check matrix $\vect{H}$ consisting of  $M$ rows
then 
\begin{equation}
  \indicator{C}{\vect{x}}= \prod_{i=1}^{M}\delta(\vect{h}_i\vect{x}^{T})\textrm{,}
\end{equation}
where $\vect{h}_i$ denotes the $i^{th}$ row of $\vect{H}$. Consequently, the APP of $\vect{X}$ is
\begin{equation}
\Pr\{\vect{X}=\vect{x}|\vect{Y}=\vect{y}\}=\n{\fieldql}{\prod_{i=1}^{M}\delta(\vect{h}_i\vect{x}^{T})\prod_{i=1}^{L}\pseudoinvi{L}{\vect{y}_i}} \textrm{.}
\label{linearapp}
\end{equation}

There are two decoding problems that can be associated with a code and the channel model defined above \cite{mackaybook}. 
The first one of these decoding problems is the \emph{codeword} decoding problem which is the task of inferring the transmitted codeword. 
This task is accomplished by finding  the codeword which maximizes the APP $\Pr\{\vect{X}=\vect{x}|\vect{Y}=\vect{y}\}$.
Hence, this decoding is called the maximum a posteriori (MAP) codeword decoding. The MAP codeword decoding can 
be formally defined as 
\begin{equation}
  \hat{\vect{x}}_{MAP}\triangleq \arg \max_{\vect{x}\in \set{C}} \Pr\{\vect{X}=\vect{x}|\vect{Y}=\vect{y}\} \textrm{.}
\end{equation}
If all codewords are equally likely then the MAP codeword decoding problem is equal to the 
maximum likelihood (ML) codeword decoding problem which maximizes the 
likelihood function $f_{\vect{Y}|\vect{X}}\{\vect{Y}=\vect{y}|\vect{X}=\vect{x}\}$ instead of the APP, i.e., 
\begin{eqnarray}
  \hat{\vect{x}}_{ML}&\triangleq& \arg \max_{\vect{x}\in \set{C}}f_{\vect{Y}|\vect{X}}\{\vect{Y}=\vect{y}|\vect{X}=\vect{x}\} \\
  &=&\arg \max_{\vect{x}\in \set{C}} \Pr\{\vect{X}=\vect{x}|\vect{Y}=\vect{y}\} \\
  &=&\hat{\vect{x}}_{MAP} \textrm{.}
\end{eqnarray}
Both MAP and ML codeword decoding problems can be solved by the min-sum (max-product) algorithm,
the most famous example of which is the Viterbi algorithm \cite{aloefg,Wiberg,wiberg,mackaybook}. 

The second decoding problem is the \emph{symbolwise} decoding problem which aims to
produce a soft prediction about the individual coded symbols. This task
is accomplished by marginalizing the APP as in
\begin{equation}
  \Pr\{X_i=x_i|\vect{Y}=\vect{y}\}= \sum_{\summary{x_i}} \Pr\{ \vect{X}=\vect{x}|\vect{Y}=\vect{y}\}
\end{equation}  
where \mbox{$\sim$$\{x_{i}\}$} is the summary notation introduced in \cite{fgsp} and indicates
that the summation runs over  the variables $x_1$, $x_2$, $\ldots$, $x_{i-1}$, $x_{i+1}$, $x_{i+2}$, $\ldots$, $x_N$.
The symbolwise decoding 
problem is solved by the sum-product algorithm whose most famous example  is the 
BCJR algorithm \cite{BCJR,mackaybook}. 

\section{Maximizing a multivariate pmf by using an ML codeword decoder \label{maximizationsection}}

We begin this section with the following example. This example might be impractical 
but it is the simplest possible example to demonstrate the idea. We will generalize the idea
after this example. 

\begin{example} \label{decoderexample} Suppose that we are required to implement a device which 
finds the configuration maximizing a pmf $p(x_1,x_2)\in \setfieldnn{2}{2}$.
This device is supposed to return the pair $(x_1,x_2)$ which maximizes  $p(x_1,x_2)$
after receiving the values  $p(0,0)$, $p(0,1)$, $p(1,0)$, and 
$p(1,1)$ as input.
Assume that while implementing this device we can use a \emph{handicapped} processor which 
can only add two numbers, negate a number, and
compute the logarithm of a number but cannot compare two numbers. Further assume that to compensate
the handicap of the processor we are given 
the  ML codeword decoder hardware of the 
linear code with the parity check matrix 
\begin{equation}
\vect{H}=[\begin{array}{ccc}1&1&1\end{array}]\textrm{,}
\end{equation}
which is designed for the channel model described in Section \ref{decoderoverview}.

If the processor at our hand was a regular processor which could compare two numbers then
the solution of this problem would be obvious. Since this processor cannot compare
two numbers, we need to figure out another solution by employing the ML codeword
decoder. In this solution we should use the processor to compute the three input vectors\footnote{Recall that
the channel model given in (\ref{channelmodel}) maps each bit to a vector in $\mathbb{R}^{2}$ } 
to be applied to the decoder from inputs applied to the whole system.   

We sketch a solution as follows. Let the input vectors applied to the decoder be $\vect{y}_1$, $\vect{y}_2$,
and $\vect{y}_3$. By (\ref{linearapp}) this decoder will return the following $\hat{\vect{x}}_{ML}=[\hat{x}_1,\hat{x}_2,\hat{x}_3]$ 
vector
\begin{equation}
 \hat{\vect{x}}_{ML}=\arg\max_{[x_1,x_2,x_3]\in \set{C}}\delta(x_1+x_2+x_3)\prod_{i=1}^{3}\pseudoinvi{L}{\vect{y}_i} \textrm{.}
\end{equation}  
Since $x_3=x_1+x_2$ for every codeword in $\set{C}$, 
\begin{equation}
 \hat{\vect{x}}_{ML}=\arg\max_{[x_1,x_2,x_3]\in \set{C}}
 \pseudoinvarg{L}{\vect{y}_1}{x_1}
 \pseudoinvarg{L}{\vect{y}_2}{x_2}
\pseudoinvarg{L}{\vect{y}_3}{x_1+x_2} \textrm{.}
\end{equation}
Due to Corollary \ref{canoniccorollary} we know that any $p(x_1,x_2)\in \setfieldnn{2}{2}$
can be expressed as 
\begin{equation}
p(x_1,x_2)=\n{\fieldnn{2}{2}}{r_1(x_1)r_2(x_2)r_3(x_1+x_2)} \textrm{.}
\end{equation}
Hence, if we apply $\vect{y}_i=\operator{L}{r_i(x)}$ to the decoder then the decoder computes
\begin{eqnarray}
 \hat{\vect{x}}_{ML}&=&\arg\max_{[x_1,x_2,x_3]\in \set{C}}r_1(x_1)r_2(x_2)r_3(x_1+x_2) \\
 &=&\arg\max_{[x_1,x_2,x_3]\in \set{C}} p(x_1,x_2)\textrm{.}
\end{eqnarray}
The first two components of the $\hat{\vect{x}}_{ML}$ is the result we are looking for. 

The only missing component of the solution is computing $\vect{y}_i=\operator{L}{r_i(x)}$.  These vectors can be 
derived using the discussion in Chapter \ref{thechapter} as
\begin{eqnarray}
\vect{y}_1 &=& [\begin{array}{c}\log p(0,0)+\log p(0,1)-\log p(1,0)-\log p(1,1)\end{array}][\begin{array}{cc} 1 & -1\end{array}] \nonumber \\
\vect{y}_2 &=& [\begin{array}{c}\log p(0,0)-\log p(0,1)+\log p(1,0)-\log p(1,1)\end{array}][\begin{array}{cc} 1 & -1\end{array}]  \\
\vect{y}_3 &=& [\begin{array}{c}\log p(0,0)-\log p(0,1)-\log p(1,0)+\log p(1,1)\end{array}][\begin{array}{cc} 1 & -1\end{array}] \nonumber \textrm{.}
\end{eqnarray}   
Fortunately,
 our handicapped processor can be programmed to accomplish this subtask. The block diagram of the solution is depicted in 
Figure \ref{decoderhardwareex}.
\end{example}
\begin{figure}
 \begin{center}
   \includegraphics[scale=.55]{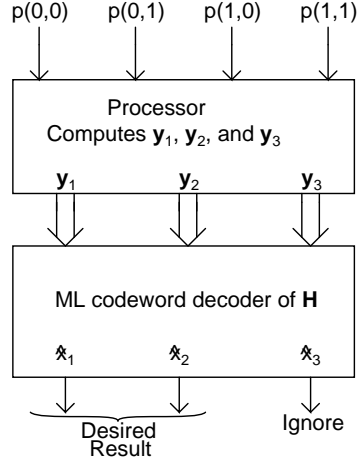}
   \caption{The block diagram of the solution to problem in Example \ref{decoderexample}. \label{decoderhardwareex}}
   \end{center}
\end{figure}

An ML codeword decoder can be utilized to maximize a multivariate pmf $t(\vect{x})$  if it 
can be expressed as a product of parity-check (zero-sum) constraints and degree one factors as in 
\begin{equation}
  t(\vect{x})=\prod_{i=1}^{M}\delta(\vect{h}_i\vect{x}^{T})\prod_{i=1}^{L}\phi_{i}(x_i) \textrm{.} \label{decodersuitablefac}
\end{equation}
The decoder which can maximize this pmf is the decoder of the linear code with the parity
check matrix $\vect{H}$ given by
\begin{equation}
  \vect{H}=\left[\begin{array}{c}\vect{h}_1\\\vect{h}_2\\\vdots \\\vect{h}_M   \end{array} \right] \textrm{.}
\end{equation}
If $\vect{y}_i=\operator{L}{\n{\fieldqn}{\phi_{i}(x)}}$ is applied
as the $i^{th}$ input to the decoder then the ML codeword decoder performs the following maximization 
\begin{eqnarray}
  \hat{\vect{x}}_{ML} &=&\arg \max_{\vect{x}\in \set{C}}\prod_{i=1}^{M}\delta(\vect{h}_i\vect{x}^{T})\prod_{i=1}^{L}\pseudoinvi{L}{\vect{y}_i} \\
  &=&\arg \max_{\vect{x}}\prod_{i=1}^{M}\delta(\vect{h}_i\vect{x}^{T})\prod_{i=1}^{L}\phi_{i}(x_i)\\
&=& \arg \max_{\vect{x}} t(\vect{x})\textrm{,}
\end{eqnarray} 
which is the desired maximization. 

Unfortunately,  most of the  pmfs cannot be factored as in (\ref{decodersuitablefac}). 
Therefore, 
it might seem that utilization of an ML codeword decoder for maximizing a pmf has limited
applicability. However, for any pmf in $\setfieldqn$ we can find a substitute pmf which factors
as in  (\ref{decodersuitablefac}) and can be used to maximize the original pmf.  Consequently,
 ML codeword decoders can be utilized
in the maximization of a broad range of pmfs. 
 
We derive such a substitute pmf based on
the canonical factorization. 
Let $p(\vect{x})\in \setfieldqn$ be the multivariate which we want to maximize by using 
an ML codeword decoder. Due to Corollary \ref{canoniccorollary}, $p(\vect{x})$
can be expressed as a product of SPC constraints as in 
\begin{equation}
  p(\vect{x}) = \n{\fieldqn}{\prod_{i=1}^{\setsize{H}}r_{i}(\vect{a}_i\vect{x}^{T})} \textrm{,}
\end{equation}
where $\set{H}$ is $\{\vect{a}_1,\vect{a}_2,\ldots,\vect{a}_{\setsize{H}}\}$ and $\n{\fieldqn}{r_{i}(\vect{a}_i\vect{x}^{T})}$
is the projection of $p(\vect{x})$ onto $\imagesop{\vect{a}_i}$. Recall that the set $\set{H}$ consists
 of $\frac{q^{N}-1}{q-1}$ pairwise linearly independent parity check vectors. Since all of these parity check
coefficient vectors are pairwise linearly independent, $N$ of them have to be of weight one. Without
loss of generality we may assume that these weight one vectors are the first $N$ parity check vectors
in $\set{H}$, i.e. $\vect{a}_1$, $\vect{a}_2$, $\ldots$, $\vect{a}_N$. Then we may define 
a matrix $\vect{A}$ by using the remaining parity check coefficient vectors in $\set{H}$ as 
\begin{equation}
  \vect{A}\triangleq \left[\begin{array}{c}\vect{a}_{N+1} \\\vect{a}_{N+2} \\\vdots\\ \vect{a}_{\setsize{H}}  \end{array}\right] \label{adef} \textrm{.}
\end{equation}
We will use $\vect{A}$ while defining the substitute pmf for $p(\vect{x})$. This substitute pmf has an extended argument 
vector $\vect{x}_E$ consisting of $L$ components where $L$ is equal to $\setsize{H}$. This extended argument vector is defined as 
\begin{equation}
  \vect{x}_{E}\triangleq[  \vect{x} \quad  \vect{x}_{A} ] \label{xedef} \textrm{,}
\end{equation} 
 where  $\vect{x}$ and $\vect{x}_{A}$ are given by
\begin{eqnarray}
  \vect{x} &\triangleq&[x_1 \ x_{2}\ \ldots \ x_{N}] \textrm{,} \\
  \vect{x}_{A}& \triangleq& [x_{N+1} \ x_{N+2}\ \ldots \ x_{L} ] \label{xadef} \textrm{.}
\end{eqnarray}
Finally, we propose the substitute pmf for $p(\vect{x})$ as 
\begin{equation}
  t_p(\vect{x}_E) \triangleq \left\{ \begin{array}{ll}p(\vect{x}), & \textrm{if }\vect{x}_A^{T}=\vect{A}\vect{x}^{T}\\
    0, &\textrm{otherwise}\end{array} \right.   \textrm{.} \label{tpdef}
\end{equation}
Clearly, $t_p(\vect{x}_E)$ achieves its maximum value at a configuration $\vect{x}_{E,MAX}$ which 
is equal to
\begin{eqnarray}
  \vect{x}_{E,MAX}&\triangleq& \arg \max_{\vect{x}_{E}} t_p(\vect{x}_E) \\
  &=& [\vect{x}_{MAX} \quad \vect{x}_{MAX}\vect{A}^{T}]   \label{maximizingrelation}
\end{eqnarray}
where $\vect{x}_{MAX}$ is the configuration maximizing $p(\vect{x})$. Due to this property
any device which determines the configuration maximizing $t_{p}(\vect{x}_{E})$ also 
determines the configuration maximizing $p(\vect{x})$ at the same time.

Now we need to show that $t_{p}(\vect{x}_E)$ can be maximized by an ML codeword decoder. 
As a first step, we can obtain an equivalent alternative definition of $t_p(\vect{x}_E)$ with using parity check constraints as
\begin{equation}
  t_p(\vect{x}_E)= p(\vect{x})\prod_{i=N+1}^{L}\delta(\vect{a}_{i}\vect{x}^{T}-x_i) \textrm{.}
\end{equation}
Inserting the canonical factorization of $p(\vect{x})$ into the equation above yields
\begin{eqnarray}
  t_p(\vect{x}_E)&=&\n{\fieldql}{\prod_{i=1}^{L}r_{i}(\vect{a}_i\vect{x}^{T})\prod_{i=N+1}^{L}\delta(\vect{a}_{i}\vect{x}^{T}-x_i)}\\
  &=&\n{\fieldql}{\prod_{i=1}^{N}r_i(\vect{a}_i\vect{x}^{T})
    \prod_{i=N+1}^{L}r_{i}(\vect{a}_i\vect{x}^{T})\delta(\vect{a}_{i}\vect{x}^{T}-x_i)}    \textrm{.}
\end{eqnarray}
Recall that we assumed the first $N$ $\vect{a}_i$ vectors to be of weight one while defining the matrix $\vect{A}$. 
Hence, we may safely assume further that these $N$ $\vect{a}_i$ vectors are the canonical basis vectors of $\fieldqn$. With this assumption 
the factorization of $t_p(\vect{x}_E)$ becomes
\begin{equation}
  t_p(\vect{x}_E)=\n{\fieldql}{\prod_{i=1}^{N}r_i(x_i)
    \prod_{i=N+1}^{L}r_{i}(\vect{a}_i\vect{x}^{T})\delta(\vect{a}_{i}\vect{x}^{T}-x_i)}\textrm{.}
\end{equation}
 The only remaining step to obtain a factorization as in (\ref{decodersuitablefac})
is to replace $r_{i}(\vect{a}_i\vect{x}^{T})\delta(\vect{a}_{i}\vect{x}^{T}-x_i)$
with $r_{i}(x_i)\delta(\vect{a}_{i}\vect{x}^{T}-x_i)$ which yields
\begin{eqnarray}
t_p(\vect{x}_E)&=&\n{\fieldql}{\prod_{i=1}^{N}r_i(x_i)
    \prod_{i=N+1}^{L}r_{i}(x_i)\delta(\vect{a}_{i}\vect{x}^{T}-x_i)}\\
&=&\n{\fieldql}{\prod_{i=1}^{L}r_i(x_i)
    \prod_{i=N+1}^{L}\delta(\vect{a}_{i}\vect{x}^{T}-x_i)} \label{substitutefactorization}\textrm{.}
\end{eqnarray}
Since all the factor functions above are either degree one factor functions or parity-check constraints,
an ML decoder of a linear code can be utilized to maximize $t_p(\vect{x}_E)$.

The parity check matrix of the linear code which can be used to maximize $t_p(\vect{x}_E)$ and consequently  $p(\vect{x})$
at the same time can be found as follows. Let a parity check coefficient vector $\vect{h}_i$ of length
$L$ be defined as in
\begin{equation}
  \vect{h}_i\triangleq \left[ \vect{a}_{i+N} \quad \vect{0}_{1\times (i-1)}\quad -1 \quad \vect{0}_{1\times (L-i-N)}  \right]  
\quad \textrm{ for }1\leq i \leq L-N  \textrm{.} \label{hidef}
 \end{equation}
Equation (\ref{substitutefactorization}) can be expressed  using $\vect{h}_i$ as 
\begin{equation}
t_p(\vect{x}_E)=\n{\fieldql}{\prod_{i=1}^{L}r_i(x_i)
    \prod_{i=1}^{L-N}\delta(\vect{h}_{i}\vect{x}^{T}_E)} \textrm{.} \label{tpfactorization1}
\end{equation}
Hence, the parity check matrix $\vect{H}$ of the code whose ML codeword decoder can be used to maximize 
$t_p(\vect{x}_E)$ and $p(\vect{x})$ is 
\begin{eqnarray}
  \vect{H}&\triangleq&\left[\begin{array}{c}\vect{h}_1\\\vect{h}_2 \\ \ldots \\ \vect{h}_{L-N} \end{array} \right] \label{hdef} \\
  &=&\left[ \vect{A} \quad -\vect{I}_{(L-N)\times (L-N)}\right] \textrm{.}
\end{eqnarray}
The $L$ input vectors  that should be applied to maximize $p(\vect{x})$ and $t_p(\vect{x}_E)$ are
\begin{equation}
  \vect{y}_i= \operator{L}{r_{i}(x)} \quad \textrm{ for } 1\leq i \leq L \textrm{.} 
\end{equation}
To sum up, with these input vectors ML codeword decoder of the linear code with parity check
matrix $\vect{H}$ returns 
\begin{eqnarray}
  \hat{\vect{x}}_{E,ML}&=& \arg \max_{\vect{x}_E} \prod_{i=1}^{L} r_{i}(x_i)\prod_{i=1}^{L-N}\delta(\vect{h}_i\vect{x}^{T}_{E}) \\
  &=&\arg \max_{\vect{x}_E} t_{p}(\vect{x}_{E}) \textrm{.}
\end{eqnarray}
Due to (\ref{maximizingrelation}) the leading $N$ components of $\hat{\vect{x}}_{E,ML}$ is the 
 $\vect{x}_{MAX}$ vector  maximizing $p(\vect{x})$ which we are seeking for. We can ignore
the rest of the  $\hat{\vect{x}}_{E,ML}$ vector. The whole process of finding the configuration
maximizing $p(\vect{x})$ is summarized in Figure \ref{maximumsummary}.
\begin{figure}
\def\putbox#1#2#3#4{\makebox[0in][l]{\makebox[#1][l]{}\raisebox{\baselineskip}[0in][0in]{\raisebox{#2}[0in][0in]{\scalebox{#3}{#4}}}}}
\def\rightbox#1{\makebox[0in][r]{#1}}
\def\centbox#1{\makebox[0in]{#1}}
\def\topbox#1{\raisebox{-0.60\baselineskip}[0in][0in]{#1}}
\def\midbox#1{\raisebox{-0.20\baselineskip}[0in][0in]{#1}}
\begin{center}
   \scalebox{1}{%
   \normalsize
   \parbox{5.39062in}{%
   \includegraphics[scale=1]{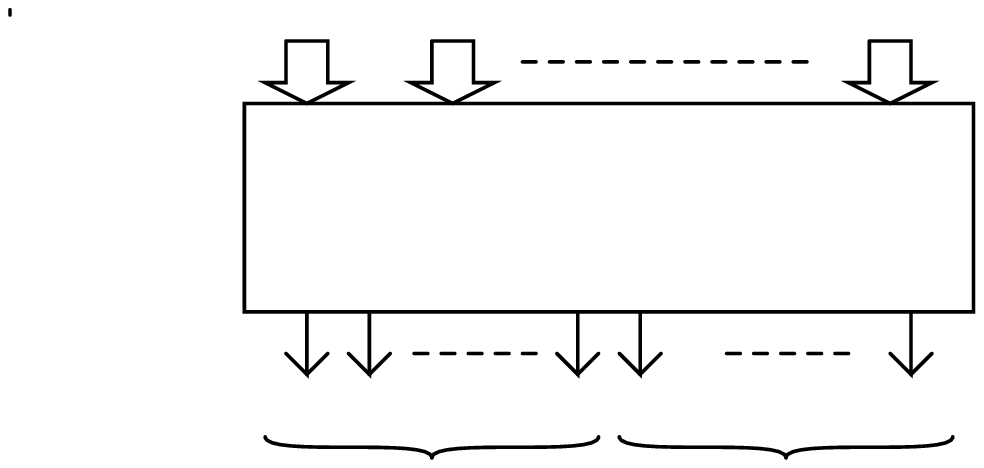}\\
   \putbox{0.83in}{2.14in}{1.20}{$\operator{L}{r_1(x)}$}%
   \putbox{1.58in}{1.56in}{1.20}{ML codeword decoder of}%
   \putbox{1.99in}{1.22in}{1.20}{$\vect{H}=[\vect{A} \quad -\vect{I}]$}%
   \putbox{1.49in}{2.14in}{1.20}{$\operator{L}{r_2(x)}$}%
   \putbox{3.33in}{2.14in}{1.20}{$\operator{L}{r_L(x)}$}%
   \putbox{1.16in}{0.56in}{1.20}{$x_1$}%
   \putbox{1.41in}{0.56in}{1.20}{$x_2$}%
   \putbox{2.24in}{0.56in}{1.20}{$x_N$}%
   \putbox{2.45in}{0.56in}{1.20}{$x_{N+1}$}%
   \putbox{3.49in}{0.56in}{1.20}{$x_L$}%
   \putbox{1.33in}{0.22in}{1.20}{Desired}%
   \putbox{1.41in}{0.06in}{1.20}{result}%
   \putbox{2.91in}{0.22in}{1.20}{Ignore}%
   } 
   } 
   \vspace{-\baselineskip} 
\end{center}
\begin{center}
    \caption{Summary of the utilization of an ML codeword decoder for maximizing a pmf\label{maximumsummary}}
    \end{center}
\end{figure} 

It is well known that ML codeword decoding problem is a special instance of 
maximization of multivariate pmf problems. In this section, we showed
that there exists a special ML codeword decoding problem which can handle
the maximization task of an arbitrary multivariate pmf. Hence, the reverse
of the well known statement above is also true. Therefore, we can
conclude that \emph{ML codeword decoding
and maximization of multivariate pmfs are equivalent problems}.

\section{Marginalizing a multivariate pmf by using a symbolwise decoder\label{marginalizationsection}}

Let an $\fieldq$-valued random vector $\vect{X}=[X_1,X_2,\ldots,X_N]$
be distributed with a $p(\vect{x})\in \setfieldqn$. The marginal pmf of $X_i$ is 
\begin{equation}
  \Pr\{X_i=x_i\}=\sum_{\summary{x_i}}p(\vect{x}) \textrm{.}
\end{equation}
A symbolwise decoder can perform this marginalization if $p(\vect{x})$ can be 
factored into degree one factor functions and parity-check constraints, which 
is not possible for a strictly positive pmf. However, as we did in the previous section,
for each $p(\vect{x})\in \fieldq$
we can obtain a substitute multivariate pmf which has the desired factorization
 and can be used in the marginalization of $p(\vect{x})$.

We can follow a more straightforward path to obtain this substitute pmf when compared to 
the previous section.
Inserting the canonical factorization of $p(\vect{x})$ into the marginalization sum above yields
\begin{eqnarray}
  \Pr\{X_i=x_i\}&=&\sum_{\summary{x_i}}\n{\fieldqn}{\prod_{i=1}^{L}r_{i}(\vect{a}_i\vect{x}^{T})}\\
  &=&\n{\fieldq}{\sum_{\summary{x_i}}\prod_{i=1}^{L}r_{i}(\vect{a}_i\vect{x}^{T})}\\
  &=&\n{\fieldq}{\sum_{\summary{x_i}}\prod_{i=1}^{N}r_i(x_i)\prod_{i=N+1}^{L}r_{i}(\vect{a}_i\vect{x}^{T})} \textrm{,}
\end{eqnarray}
where we make the same assumptions as in the previous section about the  canonical factorization of $p(\vect{x})$. 
This equation shows that $N$ of the factor functions are already of degree one. 
The remaining factor functions, which are SPC constraints, can be expressed by using the sifting property of 
the Kronecker delta function as 
\begin{equation}
  r_{i}(\vect{a}_i\vect{x}^{T})=\sum_{\forall x_i \in \fieldq} \delta(\vect{a}_i\vect{x}^{T}-x_i)r_i(x_i)\quad \textrm{ for } N+1 \leq i \leq L \textrm{.}
\end{equation}
 Since $i$ is greater than $N$, $x_i$ above is not a component of vector $\vect{x}$ and is just a dummy variable.  
Using this identity in the marginalization sum gives
\begin{eqnarray}
  \Pr\{X_i=x_i\}&=&\n{\fieldq}{\sum_{\summary{x_i}}\prod_{i=1}^{N}r_i(x_i)\prod_{i=N+1}^{L}
    \sum_{\forall x_i \in \fieldq} \delta(\vect{a}_i\vect{x}^{T}-x_i)r_i(x_i)}\\
&=&\n{\fieldq}{\sum_{\summary{x_i}}\sum_{\forall \vect{x}_A \in \field{q}^{L-N}}\prod_{i=1}^{L}r_i(x_i)\prod_{i=N+1}^{L}\delta(\vect{a}_i\vect{x}^{T}-x_i) } 
\textrm{,}
\end{eqnarray}
where $\vect{x}_A$ is as defined in (\ref{xadef}). Thanks to the summary notation the summation running over $\vect{x}_A$ can be 
merged to the first summation  which yields
\begin{equation}
  \Pr\{X_i=x_i\}=\n{\fieldq}{\sum_{\summary{x_i}}\prod_{i=1}^{L}r_i(x_i)\prod_{i=1}^{L-N}\delta(\vect{h}_i\vect{x}_E^T)} \textrm{,}
\end{equation}
where $\vect{x}_E$ and $\vect{h}_i$ are defined in (\ref{xedef}) and (\ref{hidef}) respectively. Notice that 
the two products above is the factorization of $t_p(\vect{x}_E)$, which is defined in (\ref{tpdef}), given in (\ref{tpfactorization1}). Therefore,
\begin{eqnarray}
  \Pr\{X_i=x_i\}&=& \n{\fieldq}{\sum_{\summary{x_i}} t_p(\vect{x}_E) }\\
  &=&\sum_{\summary{x_i}} t_p(\vect{x}_E) \label{marginalizationrelation}
  \textrm{.}
\end{eqnarray}
This result shows that the marginal probability of $X_i$, $\Pr\{X_i=x_i\}$, can be computed
either by marginalizing $p(\vect{x})$ or by marginalizing $t_p(\vect{x}_E)$. 

Similar to the maximization of $t_p(\vect{x}_E)$, marginalization of $t_p(\vect{x}_E)$ can be accomplished 
by the \emph{symbolwise decoder} of the linear
code with parity check $\vect{H}$ defined in (\ref{hdef}). When input vectors $\vect{y}_i=\operator{L}{r_i(x)}$
is applied to this symbolwise decoder it returns the marginal probabilities associated
with the APP
\begin{equation}
  \Pr\{\vect{X}_E=\vect{x}_E|\vect{Y}=\vect{y}\}= \n{\fieldql}{\prod_{i=1}^{L}r_i(x_i)\prod_{i=1}^{L-N}\delta(\vect{h}_i\vect{x}_E^T)}\textrm{,}
\end{equation}
which is equal to $t_p(\vect{x}_E)$. Hence, this decoder is capable of both marginalizing $t_p(\vect{x}_E)$ and consequently $p(\vect{x})$
at the same time.

We could achieve the result given in (\ref{marginalizationrelation}) through a much shorter path 
if we started from the definition of $t_p(\vect{x}_E)$ given in (\ref{tpdef}). We preferred the path followed
above to this shorter path, since the path above explains how we reached to the proposed definition of $t_p(\vect{x}_E)$
which is the most critical part of the previous section. 

It is very well known that symbolwise decoding is an instance of  marginalization problems in general.  
In this section we showed that marginalization of multivariate pmfs can be expressed as a particular symbolwise decoding
problem. Hence, it can be concluded that \emph{symbolwise decoding 
and marginalization are equivalent problems}.

\section{The decoder of the dual Hamming code as the universal inference machine\label{universalitydualhamming}}

In the previous two sections we have shown that the ML codeword and symbolwise
decoders of the linear code with parity check matrix $\vect{H}$ can be
used to maximize and marginalize any pmf in $\setfieldqn$. This parity check
matrix belongs to the dual code of a very well known code from 
coding theory. Recall that the matrix $\vect{H}$ is as defined 
as
\begin{eqnarray}
  \vect{H}&=& \left[\vect{A} \quad -\vect{I}_{(L-N)\times(L-N)}\right] \label{hdef1} \\
  &=&\left[\begin{array}{ccccc} \vect{a}_{N+1} & -1 & 0 & \cdots & 0\\ 
      \vect{a}_{N+2} & 0 & -1 & \cdots & 0 \\
      \vdots& \vdots & \vdots & \ddots & 0 \\
      \vect{a}_{L} &0 & 0 & \cdots & -1
    \end{array} \right] \textrm{.}
\end{eqnarray}
The generator matrix of this code is
\begin{equation}
  \vect{G}= \left[\vect{I}_{N\times N}  \quad \vect{A}^{T} \right]\textrm{.}
\end{equation}
In Section \ref{maximizationsection} we assumed that the first $N$ $\vect{a}_i$
vectors are the canonical basis vectors of $\fieldqn$. Therefore, 
the generator matrix can be written as
\begin{equation}
\vect{G}=\left[ \vect{a}_1^{T} \quad \vect{a}_2^{T} \quad \cdots \quad \vect{a}_{L}^{T} \right]\textrm{.}
\end{equation}
Recall that all $\vect{a}_i$ vectors were pairwise linearly independent and $L$ was equal to 
$\frac{q^{N}-1}{q-1}$. Therefore, the generator matrix $\vect{G}$ given above is
actually the parity check matrix of the Hamming code in $\fieldq$ of length $L$. 
 Hence, the parity check matrix $\vect{H}$
given in (\ref{hdef1}) is the parity check matrix of the dual Hamming code in $\fieldq$ of 
length $L$. Consequently, the ML codeword decoder of the $(L,N)$ dual Hamming code
can be configured by adjusting its inputs to maximize any pmf in $\setfieldqn$. 
Similarly, the symbolwise decoder of the $(L,N)$ dual Hamming code can be used 
as an apparatus to marginalize any pmf in $\setfieldqn$. Therefore, the decoders
of the dual Hamming codes are universal inference machines.

\section{Performing inference on special pmfs by decoders\label{specialdecoders}}

In the previous sections we have shown that the decoders of the $(L,N)$ dual
Hamming code designed for the channel model given in (\ref{channelmodel}) can be used to perform inference on any pmf in $\setfieldqn$. 
The analysis presented in the previous sections is for the most general case.
Decoders of shorter codes designed for simpler channel models can be employed
to perform inference on some pmfs  enjoying
special properties in their canonical factorization. 

\subsection{Performing inference with the decoders of shorter codes}

In Chapter \ref{specialcasechapter} we investigated the canonical 
factorizations of some special pmfs.  The canonical 
factorization of these special pmfs consisted of  less than $\frac{q^{N}-1}{q-1}$
SPC factors. We can perform inference on these special pmfs by using
the decoders of the codes whose parity check matrices are the sub-matrices
of the $(L,N)$ dual Hamming code. 

Suppose that we would like to perform inference on  a special pmf $p(\vect{x}) \in \setfieldqn$ whose canonical
factorization can be expressed as 
\begin{equation}
  p(\vect{x})=\n{\fieldqn}{\prod_{\vect{a}_i\in\set{D}}r_{i}(\vect{a}_i\vect{x}^{T})}  \textrm{,}
\end{equation} 
where $\set{D}$ is a subset of $\set{H}$ and $r_{i}(\vect{a}_i\vect{x}^{T})$
is the projection of $p(\vect{x})$ onto $\imagesop{\vect{a}_i}$. Let $\set{B}$
be a subset of $\set{D}$ which consists of all of the parity check coefficient
vectors in $\set{D}$ of \emph{weight two or more}. Moreover, let $\vect{B}$ be a $\setsize{B}\times N$ matrix
whose rows are the vectors in $\set{B}$.  Then we may define the substitute pmf $t_{p}(\vect{x}_F)$ which can be 
used to perform inference on $p(\vect{x})$ as
\begin{equation}
    t_p(\vect{x}_F) \triangleq \left\{ \begin{array}{ll}p(\vect{x}), & \textrm{if }{\vect{x}}_B^{T}=\vect{B}\vect{x}^{T}\\
    0, &\textrm{otherwise}\end{array} \right.   \textrm{.} 
\end{equation}
where $\vect{x}_{B}$ and $\vect{x}_F$ are given by
\begin{eqnarray}
\vect{x}_{B}&\triangleq& [x_{N+1} \ x_{N+2}\ \ldots \ x_{\setsize{B}+N} ] \textrm{,}\\
\vect{x}_{F}&\triangleq& [  \vect{x} \quad  \vect{x}_{A} ] \textrm{.}
\end{eqnarray}
 It can be shown through a similar path to the one in Section \ref{maximizationsection}
and Section \ref{marginalizationsection} that we can maximize or marginalize
$t_p(\vect{x}_F)$  if we wish to determine the configuration maximizing 
$p(\vect{x})$ or marginalize $p(\vect{x})$. Moreover, we can use the 
ML codeword and symbolwise decoders of the linear code with 
parity check matrix 
\begin{equation}
  \vect{H}_S\triangleq [\vect{B} \quad -\vect{I}]
\end{equation}  
to maximize or marginalize $t_p(\vect{x}_F)$. Hence, these decoders
can be used to maximize or marginalize $p(\vect{x})$. 

As in Section \ref{maximizationsection} and Section \ref{marginalizationsection},
the ML codeword and symbolwise decoders of the linear code
described by parity check matrix $\vect{H}_S$  should be configured
to perform inference on $p(\vect{x})$ by applying a certain set of
inputs. The $i^{th}$ of these inputs is $\pseudoinv{L}{v_{i}(x)}$
where $v_i(\vect{b}_{i}\vect{x}^{T})$ is the projection 
of $p(\vect{x})$ onto $\imagesop{\vect{b}_i}$, and $\vect{b}_i$
is the $i^{th}$ canonical basis vector of $\fieldqn$ if $i$ is less than 
or equal to $N$ and $(i-N)^{th}$ row of $\vect{B}$ otherwise.

Since $\set{D}$ is  a subset of $\set{H}$, $\vect{B}$ is a sub-matrix 
of $\vect{A}$ defined in (\ref{adef}). Consequently, $\vect{H}_S$ is
a sub-matrix of $\vect{H}$ defined in (\ref{hdef}). Therefore, 
implementing  the decoder associated with $\vect{H}_S$
is easier than implementing the decoder associated with $\vect{H}$. 

Actually, there are many linear codes whose decoders can be employed
to perform inference on $t_p(\vect{x}_F)$ and $p(\vect{x})$ at the same time.
For instance the decoders of the linear codes with parity check 
matrices in the form given below can be used in performing
inference on $p(\vect{x})$,
\begin{equation}
  \vect{H}_{SE} \triangleq [\vect{C}\  -\vect{I}]\textrm{,}
\end{equation}
where $\vect{C}$ is a sub-matrix of $\vect{A}$ such that it contains all 
rows of $\vect{B}$ and some more. The linear code with parity check matrix
$\vect{H}_S$ is the one with the shortest length among these codes. At a first glance 
 preferring the decoder of a longer code to a shorter one might seem 
useless while solving the same inference problem. However, in the next chapter we are going to provide some  
examples in which choosing the decoder of the longer code might be 
advantageous. 

If a linear code has a parity check matrix which can be obtained
by permuting the columns of $\vect{H}_{SE}$ then the 
decoders of this code can also be employed in maximizing or marginalizing
$p(\vect{x})$. However, in order to obtain the desired 
result we need to apply permuted inputs. 

\subsection{Performing inference by decoders designed for simpler channels\label{simplerchannel}}

Let $\set{Y}$ be the output alphabet of a communication channel which might
be a finite set, real field, complex field, or a vector space. If there exists
a sequence of channel outputs $y_1$, $y_2$, $\ldots$, $y_{\setsize{H}}$ in $\set{Y}$ such that a multivariate
pmf $p(\vect{x}) \in \setfieldqn$ can be expressed as 
\begin{equation}
  p(\vect{x})=\n{\fieldqn}{\prod_{\vect{a}_i\in \set{H}}\Pr\{Y=y_i|X=\vect{a}_i\vect{x}^{T}\}} \textrm{,} 
\end{equation}
where $\Pr\{Y=y|X=x\}$ denotes the likelihood function of the channel, then the decoder
of the dual Hamming code designed for this channel can be employed to perform inference
on $p(\vect{x})$. The inputs that should be applied to this decoder to perform
inference on $p(\vect{x})$ are obviously $y_1$, $y_2$, $\ldots$, $y_{\setsize{H}}$. 

In some problems preferring other channel models to the one described 
in (\ref{channelmodel}) might be simpler. We selected the channel model
therein since for each $r(x)\in \fieldq$ there exists a $\vect{y}\in\real^{q}$ such that 
$r(x)=\n{\fieldq}{\Pr\{\vect{Y}=\vect{y}|X=x\}}$. Consequently, the decoders 
of the dual Hamming code designed for this channel can be employed to 
perform inference on any pmf $p(\vect{x})\in \fieldqn$.  

\section{The Generic Factor Graph and Equivalent Tanner graph\label{graphicalrepresentation}}

In Chapter \ref{thechapter}, it is shown that the canonical factorization of 
any multivariate pmf $p(\vect{x})$ in $\setfieldqn$ exists which 
is given by
\begin{equation}
p(\vect{x})=\n{\fieldqn}{\prod_{i=1}^{L}r_{i}(\vect{a}_i\vect{x}^{T})} \nokta
\end{equation}
In this chapter, we made some assumptions on $\vect{a}_i$. We assumed without loss of generality that
the first $N$ parity check coefficient vectors, $\vect{a}_1$, $\vect{a}_2$, 
$\ldots$, $\vect{a}_{N}$, are the canonical basis vectors of $\fieldqn$. 
Therefore, the canonical factorization of any $p(\vect{x})$ becomes
\begin{equation}
p(\vect{x})=\n{\fieldqn}{\prod_{i=1}^{N}r_i(x_i)\prod_{i=N+1}^{L}r_{i}(\vect{a}_i\vect{x}^{T})} \nokta
\end{equation}
The factor graph representing this factorization is shown in Figure \ref{genericgraphs}-a. 
This factor graph can represent any pmf in $\setfieldqn$ since all of the pmfs in $\setfieldqn$
has a factorization given above. The only difference between any two factor graphs representing
two different joint pmfs are the factor functions in the factor graph. 

\begin{figure}
  \begin{center}
\begin{tabular}{c}
  \includegraphics[scale=.6]{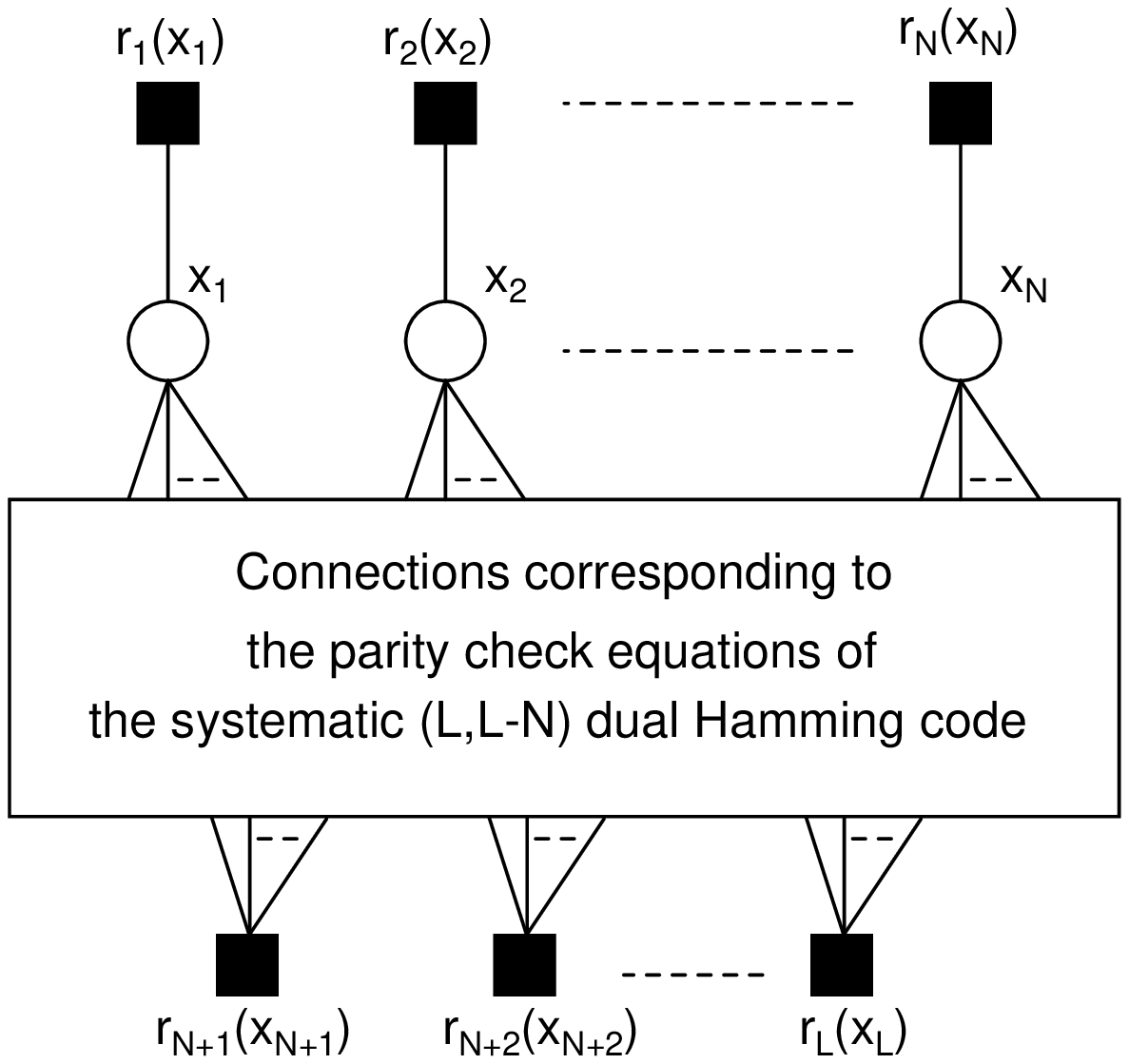}\\
  (a)\\
  \includegraphics[scale=.6]{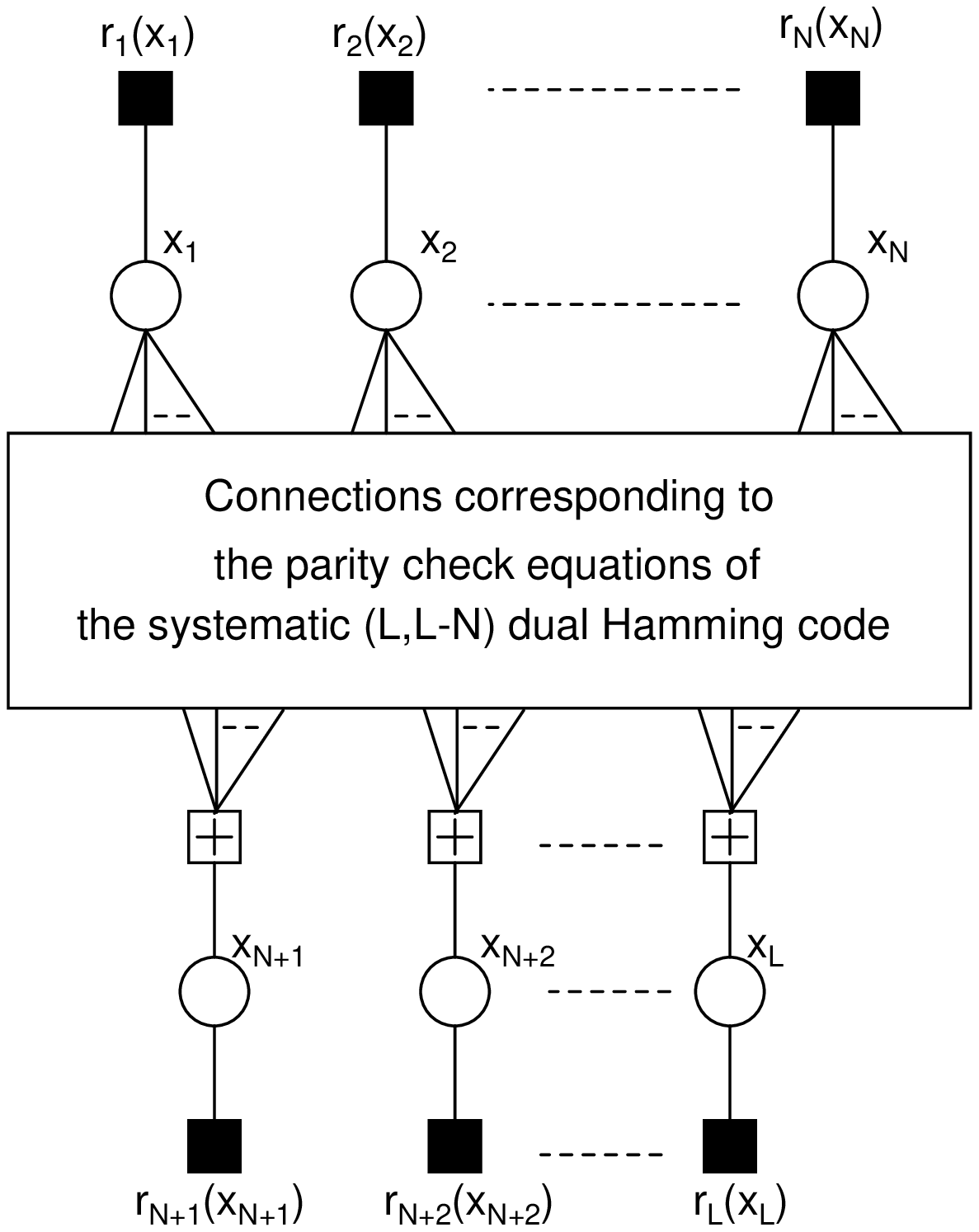}\\
  (b)
\end{tabular}
  \caption{(a) The generic factor graph which can represent any $p(\vect{x})$ in $\setfieldqn$.  
(b) The equivalent Tanner graph of the generic Tanner graph. \label{genericgraphs}}    
  \end{center}
\end{figure}

In this chapter, we showed that  performing inference on $p(\vect{x})$  is equivalent to 
performing inference  on $t_p(\vect{x}_E)$ which is defined in (\ref{tpdef}). 
The factorization of $t_p(\vect{x}_E)$ given in \ref{substitutefactorization}
is represented by the Tanner graph shown in Figure \ref{genericgraphs}-b.in 
Hence, this Tanner graph is the equivalent Tanner graph representing 
the canonical factorization. 
While transforming the factor graph in Figure \ref{genericgraphs}-a to the
Tanner graph in in Figure \ref{genericgraphs}-b,  auxiliary variable
nodes representing the variables $x_{N+1}$, $x_{N+2}$, $\ldots$, $x_{L}$ are 
added. These auxiliary variables are very different from the 
hidden state nodes introduced in the Wiberg style Tanner graphs \cite{Wiberg}.    

\section{Importance \label{importancesection}}

Using channel decoders for inference tasks beyond decoding is important mainly in two 
aspects. Firstly, using a channel decoder for an inference task provides 
new hardware options in  the solution of the inference 
problems. Among these hardware options the analog probability propagation technique
is important in particular \cite{loeligeranalog}. Secondly, 
new approximate algorithms for the solution of the inference
problems can be developed by using the sub-optimal decoders of the codes, which have been
studied for a long time.

\subsection{Performing inference with  probability propagation in analog VLSI}

The semiconductor devices such as transistors and diodes are the most primitive 
building blocks of any electronic device today. By their very nature these
devices are \emph{nonlinear}.  Over the last few decades engineers
developed ways to cope with this nonlinearity. While designing analog
circuitry engineers restrict the operation of the circuit to such 
a region in which these devices behave almost linearly. Another way 
to cope with nonlinearity of these devices is avoiding analog
circuits as much as possible and trying to implement everything in
digital. The signals flowing in a digital circuit are so large that transistors
behave like switches. Hence, digital circuits are robust against
the nonlinearity of the transistors. Digital circuits are also robust 
against other factors such as component mismatch and noise. Due to these
and some other advantages digital circuits are usually preferred to analog
circuits. 

However, Carver Mead, who is one of the 
pioneers of the VLSI revolution, claimed in his book \cite{meadbook} that digital computation is inefficient
and analog computation
is the way to achieve the capacity and the efficiency of the brains of the animals. 
Moreover, he claimed that analog computation can be made as robust as 
 digital computation to the factors such as noise and component mismatch. 
He provided many practical  examples to support his claims in his book.  

A decade after Mead's book, another evidence arise from coding theory to 
support his claims. Just two operations are sufficient to perform 
soft-input soft-output decoding. These operations are addition, which can be
easily implemented with analog circuitry, and the hyperbolic tangent function \cite{fgsp}. 
Since the differential pair exhibits tangent hyperbolic function this second
function can also be implemented with analog circuits. Motivated with this
idea, 
 Loeliger and his group designed and tested analog circuits
to perform decoding of channel codes \cite{loeligerdigeranalog,loeligeranalog}.  
 They report that their analog decoding circuitry consumes two orders of magnitude
less power than their digital counterparts. This efficiency arises from the
fact that their analog decoding circuit does not fight with the nonlinearities 
of the transistors but exploits those nonlinearities \cite{loeligerdigeranalog}. 
They also report that these circuits are robust to component mismatch. 

Loeliger's ``probability propagation in analog VLSI'' has an important limitation. This 
approach can be applied
to  probabilistic inference problems if a condition related
to  the factorization of the multivariate pmf under concern is satisfied. This condition
states that the pmf should be able to be expressed as a product
of zero-one valued functions and functions of degree one \cite{loeligeranalog}.
Although this condition is satisfied in decoding problems, it is not 
satisfied in other  problems arising in communication theory such as channel equalization and 
MIMO detection. Hence, equalizers or MIMO detectors could not be built directly with their brilliant idea
whereas decoders could. A pure decoder implemented with probability
propagation in analog is not very useful without implementing the equalizer or detector 
in analog since the interface required between the decoder and the equalizer (or detector) would cancel
all the efficiency of the analog decoder. 

In this chapter, we showed that inference problems can be solved by using channel decoders. Hence,
it is possible to solve the equalization or MIMO detection problems by decoders. Consequently,
the results presented in this chapter, allows us to implement channel equalizers or MIMO detectors
with the very efficient analog probability propagation approach. 
It is reasonable to expect, based on the experience 
on analog decoding, that such receiver blocks would be two orders of magnitude smaller in size and
consumes two orders of magnitude less power than current receivers. 
Probably this aspect 
will be the most important contribution of this thesis.  


\subsection{New approximate inference algorithms}

The iterative sum-product algorithm running on loopy Tanner graphs is proven to 
be efficient decoding algorithm for various codes.  The sum-product algorithm
is characterized by the Tanner graph representing the code. A code might be represented
with many different parity check matrices. For each parity check matrix, more than 
one Tanner graphs might be obtained representing the code. Hence, for each 
code we have various alternative Tanner graphs to represent the code. Consequently, 
we may have various versions of the sum-product algorithm to decode the same code.    
Each of these alternative versions have different characteristics in terms
of complexity and performance \cite{wiberg}. Therefore, employing a channel
decoder to perform an inference task allows us to choose among different 
sum-product algorithm versions to handle the inference task. Hence, 
 new approximate inference algorithms can be developed in this manner. 
We provide an example on MIMO detection in the next chapter.

\chapter{USING  CHANNEL DECODERS AS DETECTORS\label{applicationchapter}}

\section{Introduction}

This chapter contains examples to the idea presented in Chapter \ref{decodingchapter} by showing
how to employ channel decoders as the detectors of communication receivers. 
One of these examples which 
is MIMO detection by using the decoder of a tail biting convolutional code demonstrates that
new inference algorithms with low complexity can be developed by employing channel 
decoders for other purposes.   

Unfortunately, some of the derivations presented in this chapter might appear quite tedious, Sections \ref{mpamdemodsection}, 
\ref{graydemodsection}, and \ref{mimodemodsection} in particular. Actually, the derivations
in these sections are straightforward applications of the methods presented in the previous chapter. 
Most of these derivations are so straight forward that they can be derived with symbolic programming.
Indeed, we used the  GiNaC symbolic programming library in C++ while deriving some of the cumbersome derivations
presented in this chapter. Hence, reporting and following these derivations is much more difficult than deriving them.  
However, these sections include examples to make the subject more concrete. These examples also demonstrate
how the same decoder can be used for different purposes by changing its  inputs. 
 
This chapter begins with analyzing the multiple-input single-output (MISO) detection. 
Then the results obtained in Section \ref{misoqpsk} are used to derive the channel 
decoder which can be used in the detection of naturally mapped
pulse amplitude modulation (PAM) signals in Section \ref{mpamdemodsection}.
Section \ref{graydemodsection} explains the detection of gray mapped 
PAM signals by using channel decoders. Section \ref{mimodemodsection}
investigates the multiple-input multiple-output (MIMO) detection
of QPSK signal by using decoders. Special attention is paid to 
the MIMO detection of QPSK signals by using the decoders 
of tail biting convolutional codes in Section \ref{tailbitingsection}.
This section also includes some simulation results.
This chapter ends with briefly reporting that the Viterbi and BCJR decoders of 
the convolutional codes can be used channel equalizers.

\section{MISO detection  of $q$-ary PSK signaling with prime $q$ \label{misoqpsk} by using a channel decoder}

The MISO detection of $q$-ary PSK signaling under additive  Gaussian noise  is the simplest task 
(in terms of derivation) to be handled by a decoder. Moreover, analyzing this case 
first helps to transform other detection problems to decoding problems. ML MISO
detection task is finding the most likely input sequence given the 
received symbol. This task can be handled by 
ML codeword decoders. Soft output MISO detection is the computation 
of marginal a posteriori probabilities. This task can be handled by 
symbolwise decoders.

\subsection{Signal Model}
Let $\psk{x}$ be a function from $\fieldq$ to $\complex$ representing 
the $q$-ary PSK \footnote{$q$-ary PSK is not the same as QPSK. } mapping, i.e.  
\begin{equation}
\psk{x} \triangleq \exp\left(j \frac{2\pi}{q}\integer{x} \right) \textrm{,}
\end{equation}
where $\integer{.}$ denotes the usual mapping from $\fieldq$ to $\mathbb{N}$. 
Let a complex-valued random variable $Y$ be related to an $\fieldq$-valued random vector $\vect{X}=[X_1,X_2,\ldots,X_N]$
as 
\begin{equation}
  Y \triangleq \sum_{i=1}^{N}h_{i}\psk{X_i} +Z \label{misomodel}
\end{equation} 
where $h_i$ is a complex constant and $Z$ is a zero mean circularly symmetric complex Gaussian noise
with $\expectation{ZZ^{*}}=2\sigma^{2}$. Clearly, $Y$ models the received symbol after the symbols $X_{1}$, $X_{2}$, $\ldots$, $X_{N}$
are modulated with $q$-ary PSK and passed  through a $1 \times N$ multi-input single output (MISO) channel with channel coefficients $h_{i}$.
With these assumptions the a posteriori pmf $\vect{X}$ is 
\begin{equation}
  \Pr\{\vect{X}=\vect{x}|Y=y\} = \n{\fieldqn}{ \exp\left(-\frac{\abs{y-\sum_{i=1}^{N}h_{i}\psk{x_i}}^{2}}{2\sigma^{2}}  \right)} \textrm{,}
\label{appmisoqpsk}
\end{equation}
where $\vect{x}=[x_1,x_2,\ldots,x_N]$. We assume perfect channel information is known at the receiver side.

\subsection{The canonical factorization of the joint a posteriori pmf} 
The first step in determining 
the linear code whose ML codeword (symbolwise)  decoder can be used to maximize (marginalize) the a posteriori  
pmf $\Pr\{\vect{X}=\vect{x}|Y=y\}$ is obtaining the canonical factorization of the a posteriori pmf of $\vect{X}$.
The generic procedure of obtaining this canonical factorization is explained in detail 
in Chapter \ref{thechapter}, which could 
have been prohibitively tedious for this problem. Fortunately, the joint a posteriori pmf in this problem,  $\Pr\{\vect{X}=\vect{x}|Y=y\}$,
enjoys many special properties so that deriving its canonical factorization is easier. 

Let $p(\vect{x})$ denote  $\Pr\{\vect{X}=\vect{x}|Y=y\}$. As shown in Appendix \ref{facmisoqpsk},
$p(\vect{x})$ can be factored as in 
\begin{equation}
  p(\vect{x})= \n{\fieldqn}{\prod_{i=1}^{N}\exp\left(\frac{2\Re{yh_i^*\psk{x_i}^*}}{2\sigma^2}\right) 
  \prod_{j=2}^{N}\prod_{i=1}^{j-1}  \exp\left(-\frac{2\Re{h_ih_j^*\psk{x_i}\psk{x_j}^*}}{2\sigma^2}\right) }  \textrm{.}
\end{equation}
Since we have a known factorization for $p(\vect{x})$, we can apply Theorem \ref{localfactorizationtheorem}
to obtain the \emph{canonical factorization} of $p(\vect{x})$ as explained below.

Let two functions  $\gamma(\omega;\rho,\sigma)$ and $\theta(\omega_1,\omega_2;\chi,\sigma)$ be defined as in 
\begin{eqnarray}
  \gamma(\omega;\rho,\sigma)&\triangleq&\exp\left( \frac{2\Re{\rho \psk{\omega}^*}}{2\sigma^2} \right)\label{gammadef}\textrm{,} \\
  \theta(\omega_1,\omega_2;\chi,\sigma)&\triangleq&\exp\left(-\frac{2\Re{\chi\psk{\omega_1}\psk{\omega_2}^*}}{2\sigma^2} \right) \label{thetadef}
 \textrm{,}
\end{eqnarray}
for $\omega$, $\omega_1$, $\omega_2$  in $\fieldq$, $\sigma$ in $\real$, and $\rho$, $\chi$ in $\complex$.
The function $\gamma(\omega;\rho,\sigma)$ is nothing but the likelihood function of $\omega$ when 
it is modulated with $q$-ary PSK, passed through an additive white Gaussian noise (AWGN) channel with power spectral density (PSD)
 $N_{0}/2=\sigma^{2}$, and 
given that the  value  at the output of the matched filter is $\rho$. 
Using these functions the factorization of $p(\vect{x})$ becomes
\begin{equation}
  p(\vect{x})=\n{\fieldqn}{\prod_{i=1}^{N} \gamma(x_i;yh_i^*,\sigma) \prod_{j=2}^{N}\prod_{i=1}^{j-1}\theta(x_i,x_j;h_ih_j^*,\sigma) } 
  \label{appfactorizationgammatheta}\textrm{.}
\end{equation}

Notice that the factorization of $p(\vect{x})$ given above is composed of degree one and degree two factors only \footnote{We 
regard $\rho$, $\chi$ and $\sigma$ as parameters of functions $\gamma(.;.)$ and $\theta(.;.)$, not their arguments.}. 
Therefore, the \emph{canonical factorization} of $p(\vect{x})$ should be composed of SPC factors
of degree one and two due to Theorem \ref{localfactorizationtheorem}. The 
SPC factors of degree one composing $p(\vect{x})$ are simply the normalizations of $\gamma(x_i;yh_i^*,\sigma)$'s.  

The SPC factors of degree two composing $p(\vect{x})$ can be derived by obtaining the 
canonical factorization of $\theta(x_i,x_j;h_ih_j^*,\sigma)$. The straightforward 
way of deriving the canonical factorization of $\theta(x_i,x_j;h_ih_j^*,\sigma)$
might be projecting this function onto the subspaces $\imagesop{(\vect{f}_i+\alpha\vect{f}_j)}$
for all nonzero $\alpha\in \fieldq$, where $\vect{f}_i$ is the $i^{th}$ canonical 
basis vector of $\fieldqn$. However, the required canonical factorization 
can be obtained in a simpler way by exploiting the fact that $q$ is assumed
to be a prime number in this section. Since $q$ is a prime 
number, $\fieldq$ is a prime field. Consequently, the subtraction in $\fieldq$
is the subtraction modulo $q$. Due to this fact,
\begin{equation}
  \psk{\omega_1}\psk{\omega_2}^*=\psk{\omega_1-\omega_2} \textrm{.}
\end{equation} 
Therefore, 
\begin{eqnarray}
\theta(\omega_i,\omega_j;\chi,\sigma)&=&\exp\left(-\frac{2\Re{\chi\psk{\omega_1-\omega_2}}}{2\sigma^2} \right)\\
&=&\gamma(\omega_1-\omega_2;-\chi,\sigma)  \textrm{.}  
\end{eqnarray}
Inserting this result into (\ref{appfactorizationgammatheta}) yields, 
\begin{equation}
  p(\vect{x})=\n{\fieldqn}{\prod_{i=1}^{N} \gamma(x_i;yh_i^*,\sigma) \prod_{j=2}^{N}\prod_{i=1}^{j-1}\gamma(x_i-x_j;-h_ih_j^*,\sigma) } 
  \label{gammafactorization}\textrm{.}
\end{equation} 
We can define  pmfs in $\setfieldq$ by scaling $\gamma(x;yh_i^*,\sigma)$ and $\gamma(x;-h_ih_j^*,\sigma)$
as in 
\begin{eqnarray}
  r_i(x)&\triangleq&\n{\fieldq}{\gamma(x;yh_i^*,\sigma)} \textrm{,} \\
  r_{i,j}(x) &\triangleq&\n{\fieldq}{\gamma(x;-h_ih_j^*,\sigma)}\textrm{.}
\end{eqnarray}
The factorization of $p(\vect{x})$ can be expressed by using these pmfs as 
\begin{eqnarray}
  p(\vect{x})&=&\n{\fieldqn}{\prod_{i=1}^{N} r_i(x_i)\prod_{j=2}^{N}\prod_{i=1}^{j-1}r_{i,j}(x_i-x_j)}\\
  &=&\n{\fieldqn}{\prod_{i=1}^{N} r_i(\vect{f}_i\vect{x}^{T})\prod_{j=2}^{N}\prod_{i=1}^{j-1}r_{i,j}(\vect{a}_{i,j}\vect{x}^{T})}
  \label{canonicalfactofapp} \textrm{,}
\end{eqnarray}
where $\vect{a}_{i,j}$ is  
\begin{equation}
  \vect{a}_{i,j}\triangleq \vect{f}_i-\vect{f}_j \label{aijdef}\textrm{.}
\end{equation}
Notice that all the factor functions in factorization above are SPC factors. 
Moreover, the parity check coefficient vectors of all SPC factors are pairwise 
linearly independent. Hence, due to Definition \ref{canonicdefinition}, the factorization of $p(\vect{x})$ 
given in (\ref{canonicalfactofapp}) is the canonical factorization of $p(\vect{x})$.

\subsection{The  decoders which are able to perform inference on the joint a posteriori pmf\label{misodecodersubsection}}

As explained in Section \ref{universalitydualhamming} the ML codeword and symbolwise decoders of the 
dual Hamming code of length $\frac{q^{N}-1}{q-1}$ can perform inference on $p(\vect{x})$. However, 
since the canonical factorization of $p(\vect{x})$ given in (\ref{canonicalfactofapp}) 
consists of less than $\frac{q^{N}-1}{q-1}$ SPC  factors, the ML codeword or symbolwise decoders
of a shorter code can be employed for maximizing and marginalizing $p(\vect{x})$ as discussed 
in Section \ref{specialdecoders}. Following the discussion in Section \ref{specialdecoders}
the parity check matrix of this code whose decoder can be employed in the demodulation 
of $1\times N$ MISO system is
\begin{equation}
  \vect{H}_{qPSK}(N)\triangleq \left[\begin{array}{c}
      \vect{a}_{1,2} \\
      \vect{a}_{1,3}\\
      \vect{a}_{2,3}\\
      \vdots \\
      \vect{a}_{1,N}\\
      \vect{a}_{2,N}\\
      \vdots\\
      \vect{a}_{N-1,N}
    \end{array}\quad -\vect{I}_{\frac{N(N-1)}{2} \times\frac{N(N-1)}{2}}\right] \textrm{.}
\end{equation}
For a neater  representation of $\vect{H}_{qPSK}(N)$, we define a matrix  parameterized on $i$ and $N$ $\vect{K}(i,N)$ 
as  
\begin{equation}
 \vect{K}(i,N) \triangleq  \left[\vect{I}_{i\times i } \quad -\vect{1}_{i\times 1} \quad \vect{0}_{i\times (N-i-1)} \right] \label{Kdef}
 \textrm{.}
\end{equation}
Then $\vect{H}_{qPSK}(N)$ can be expressed as 
\begin{equation}
  \vect{H}_{qPSK}(N)= \left[\begin{array}{c} \vect{K}(1,N)\\\vect{K}(2,N)\\ \vdots \\\vect{K}(N,N)  \end{array} \quad 
-\vect{I}_{\frac{N(N-1)}{2} \times\frac{N(N-1)}{2}}\right] \textrm{.}
\end{equation}

The complete specification of a decoder of a linear code consists of a parity 
check matrix and a channel model. The parity check matrix of the decoders
which can detect received symbols of $1\times N$ MISO system are explained above. 
As the channel model we can use the one  described in (\ref{channelmodel}). 
However, we can use a more natural channel model in this case as explained
in Section \ref{simplerchannel}. Recall that the factorization 
of $p(\vect{x})$ given in (\ref{gammafactorization}) is composed of likelihood functions of 
the channel which first modulates  an $\fieldq$-valued symbol with $q$-ary PSK
and then passes through an AWGN channel with PSD $N_{0}/2=\sigma^2$.  Therefore,
the received symbols of $1\times N$ MISO system can be detected with the decoders
of the code with parity check matrix $\vect{H}_{qPSK}$ which is designed for $q$-ary PSK
modulation and AWGN channel with variance $\sigma^2$. In order to achieve the desired
detection inputs that should be applied to these decoders are components of 
the vector given below. 
\begin{equation}
[yh_{1}^{*} \ yh_{2}^{*}\ \ldots \ yh_{N}^{*}\ -h_1h_2^*\ -h_1h_3^*\ -h_2h_3^*\ \ldots \ -h_1h_N^*\ -h_2h_N^*\ \ldots -h_{N-1}h_N^*   ]\nonumber
\end{equation}
We can also use a modification of the same decoder which is designed for standard noise with $\sigma^{2}=1$. In this case
all of the inputs given above should be scaled by $\frac{1}{\sigma}$.

\begin{example} This example  demonstrates how can we  employ a symbolwise decoder to compute
the marginal APPs in a $1\times 4$ MISO system.  Let a complex-valued random variable $Y$
be given as
\begin{equation}
  Y=\sum_{i=1}^{4}h_i\psk{X_i}+Z
\end{equation}
where $X_i$ is an $\fieldq$-valued random variable and $Z$ is the circularly symmetric Gaussian noise
with $\expectation{ZZ^*}=2$. Our aim is to compute $\Pr\{X_i=x_i|Y=y\}$ by using 
a symbolwise decoder. As explained above the parity check matrix of this decoder is
\begin{eqnarray}
  \vect{H}_{qPSK}(4)&=&\left[ \begin{array}{c} \vect{K}(1,4)\\\vect{K}(2,4)\\\vect{K}(3,4) \end{array} \quad -\vect{I}_{6\times 6} \right]\\
  &=&
\left[\begin{array}{ccccccccccccccc}  
1& -1 & 0 & 0& -1 & 0 & 0 & 0 & 0 & 0 \\
1& 0 & -1 & 0&  0 & -1 & 0 & 0 & 0 & 0 \\
0& 1 & -1 & 0&  0 & 0 & -1 & 0 & 0 & 0 \\
1& 0 & 0 & -1&  0 & 0 & 0 & -1 & 0 & 0 \\
0& 1 & 0 & -1&  0 & 0 & 0 & 0 & -1 & 0 \\
0& 0 & 1 & -1&  0 & 0 & 0 & 0 & 0 & -1 
\end{array} \right] \textrm{.} \label{hqpsk4}
\end{eqnarray}
The input vector that should be applied to this decoder is 
\begin{equation}
[yh_{1}^{*} \quad yh_{2}^{*}\quad yh_{3}^{*} \quad
 yh_{4}^{*}\quad -h_1h_2^*\quad -h_1h_3^*\quad -h_2h_3^*\quad  -h_1h_4^*\quad -h_2h_4^*\quad  -h_{3}h_4^*    ]\nonumber \textrm{.}
\end{equation}
Notice that configuring the demodulator for a new observation and new set of channel coefficients requires only 
changing the inputs to the decoder. This example is illustrated in Figure \ref{qpskmisofig}. 
\end{example}
\begin{figure}
\begin{center}
\begin{tabular}{c}
  \includegraphics[scale=.6]{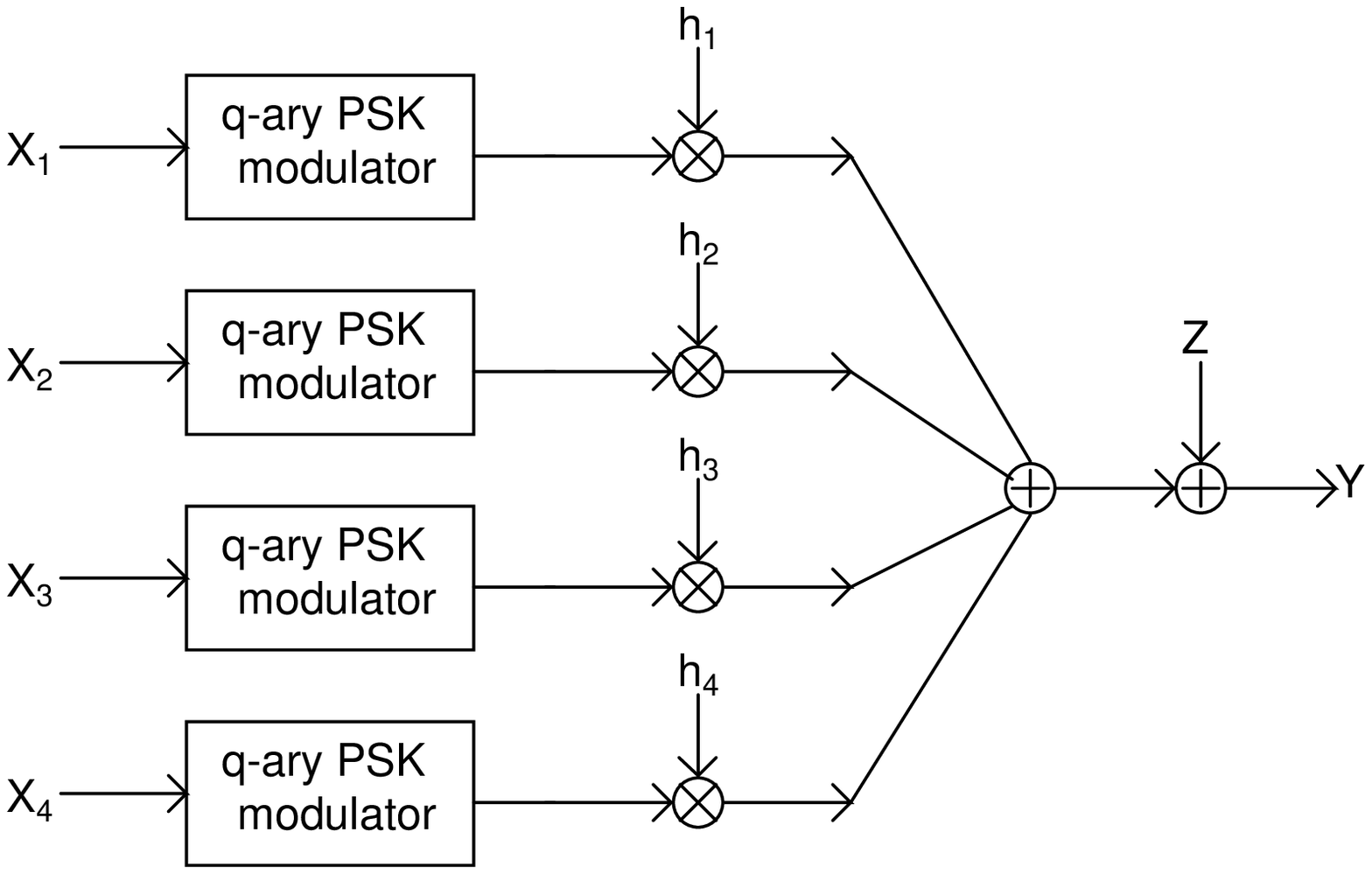}\\
  (a)\\
  \includegraphics[scale=.7]{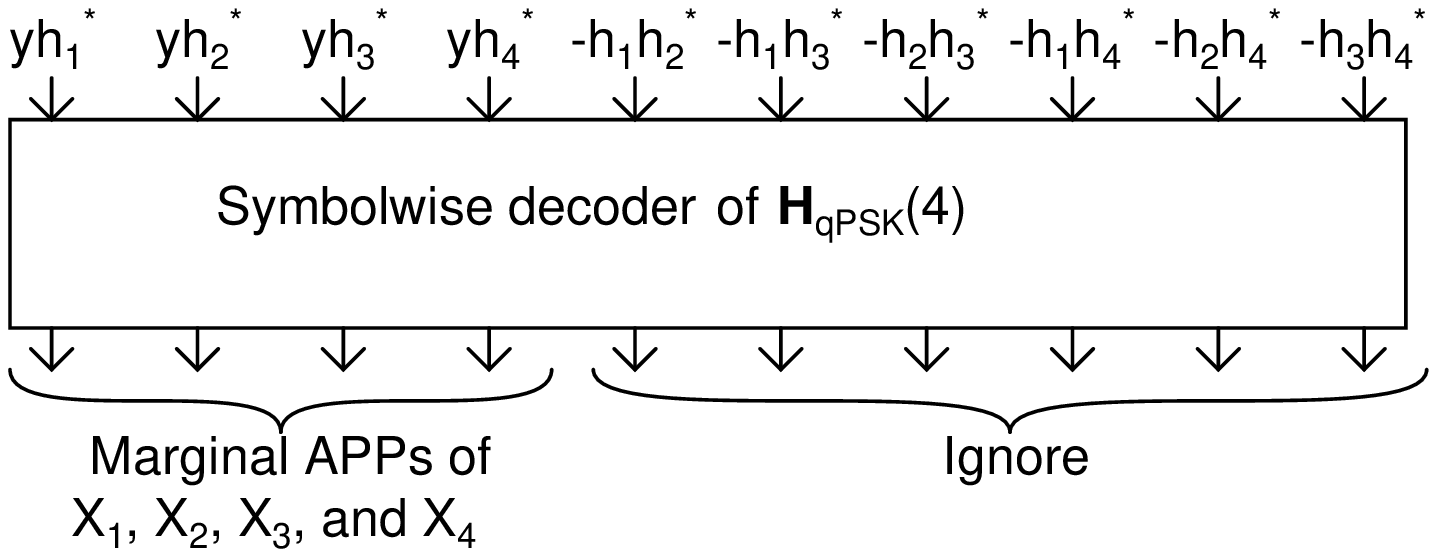}\\
  (b)
\end{tabular}
  \caption{$1\times 4$ MISO system. (a) The system model. (b) Demodulating the received symbol by using a symbolwise decoder. \label{qpskmisofig}}
\end{center}
\end{figure}

\section{Channel decoders as detectors of naturally mapped M-PAM\label{mpamdemodsection}}
In this section we show how to demodulate the naturally mapped M-PAM modulation 
by using a channel decoder. Let $\pam{N}{\vect{x}}$ be a function from $\fieldn{2}$ to $\real$
which maps binary valued vectors of length $N$ to $M\triangleq 2^{N}$ real amplitude 
values as in the naturally mapped PAM modulation, i.e. 
\begin{equation}
  \pam{N}{\vect{x}}\triangleq \sum_{i=1}^{N}2^{i-1}\bpsk{x_i} \virgul \label{pamdef}
\end{equation}
where $\vect{x}=[x_1,x_2,\ldots,x_N]$ and $\bpsk{x}$ denotes binary antipodal mapping given in 
\begin{equation}
  \bpsk{x}\triangleq \left\{\begin{array}{lc}1, & x=0 \\ -1, &x=1 \end{array} \right. \nokta
\end{equation}
Assume that $\pam{N}{\vect{X}}$ is transmitted through a discrete additive Gaussian noise channel and $Y$ is received. In other words,  
\begin{equation}
  Y=\pam{N}{\vect{X}}+Z\virgul \label{pamchannel}
\end{equation}
where $\vect{X}=[X_1,X_2,\ldots,X_N]$ and $Z$ is a real Gaussian random variable with variance $\sigma^{2}$.
Inserting the definition of $\pam{N}{\vect{X}}$ into (\ref{pamchannel}) yields
\begin{equation}
  Y=\sum_{i=1}^{N}2^{i-1}\bpsk{X_i} + Z \nokta \label{pamasmiso}
\end{equation}
Since $\bpsk{X_i}$ is equal to $\psk{X_i}$ for $q=2$ and $2$ is a prime number, 
(\ref{pamasmiso}) is a special case of (\ref{misomodel}). Consequently, 
naturally mapped M-PAM detection is a special case of MISO detection of binary phase shift keying (BPSK)
with channel coefficients $h_i=2^{i-1}$. 
Hence, the parity check matrix of the code whose decoder can demodulate M-PAM is 
$\vect{H}_{2PSK}(\log_{2}M)$. Since $-1$ is equal to $1$ in the binary field,
all of the minus ones in  $\vect{H}_{2PSK}(\log_{2}M)$ can be replaced with ones. 
The input vector that should applied to the decoder in order achieve demodulation of  M-PAM
is
\begin{multline}
\bigg[\frac{y}{\sigma} \quad  \frac{2y}{\sigma}\quad  \ldots \quad  \frac{2^{N-1}y}{\sigma} \quad -\frac{2^{0}2^{1}}{\sigma}\quad 
-\frac{2^{0}2^{2}}{\sigma} \quad -\frac{2^12^2}{\sigma} \quad -\frac{2^{0}2^3}{\sigma} \quad -\frac{2^{1}2^3}{\sigma} \quad
-\frac{2^{2}2^3}{\sigma} \quad  \ldots 
\\\ldots \quad -\frac{2^02^{N-1}}{\sigma}\quad -\frac{2^12^{N-1}}{\sigma}\quad \ldots
-\frac{2^{N-2}2^{N-1}}{\sigma}    \bigg] \nonumber  \virgul
\end{multline}  
where $y$ denotes the received value. 

Implementing an ML M-PAM detector by using the ML codeword decoder of the code with parity check 
matrix $\vect{H}_{2PSK}(\log_{2}M)$ might not be  practical since there are simpler ways 
to implement such a detector. However, implementing a soft output M-PAM detector
by using the symbolwise decoder of the same code might be of practical importance. 

\begin{example}\label{naturalexample}This example shows how to compute marginal APPs of four bits
which are modulated with naturally mapped 16-PAM and passed through
an AWGN channel with PSD $N_0/2=\sigma^2$. Constellation diagram 
of the naturally mapped 16-PAM is shown in Figure \ref{natural16pamfig}-a. 

The parity check matrix of the code whose symbolwise decoder can be used
to compute marginal APPs of the individual bits is 
\begin{equation}
  \vect{H}_{2PSK}(4) = 
\left[\begin{array}{ccccccccccccccc}  
1& 1 & 0 & 0& 1 & 0 & 0 & 0 & 0 & 0 \\
1& 0 & 1 & 0&  0 & 1 & 0 & 0 & 0 & 0 \\
0& 1 & 1 & 0&  0 & 0 & 1 & 0 & 0 & 0 \\
1& 0 & 0 & 1&  0 & 0 & 0 & 1 & 0 & 0 \\
0& 1 & 0 & 1&  0 & 0 & 0 & 0 & 1 & 0 \\
0& 0 & 1 & 1&  0 & 0 & 0 & 0 & 0 & 1 
\end{array} \right] \textrm{.}
\end{equation}
Notice that this parity check matrix is a special case of the $\vect{H}_{qPSK}(4)$ matrix given in the previous 
example for $q=2$. Since $-1$ is equal to $1$ in the binary field, minus ones in that matrix are replaced 
with plus ones. 

If the received value is denoted with $y$ then the input vector that should be applied to this decoder 
is 
\begin{equation}
\left[\frac{y}{\sigma}\quad \frac{2y}{\sigma}\quad \frac{4y}{\sigma}\quad \frac{8y}{\sigma}\quad 
-\frac{2}{\sigma}\quad -\frac{4}{\sigma}\quad -\frac{8}{\sigma}\quad -\frac{8}{\sigma}\quad 
-\frac{16}{\sigma}\quad -\frac{32}{\sigma} \right]\nonumber \textrm{.}
\end{equation}
 This example is illustrated in Figure \ref{natural16pamfig}. 
\end{example}

\begin{figure}
  \begin{center}
\begin{tabular}{c}
  \includegraphics[scale=.5]{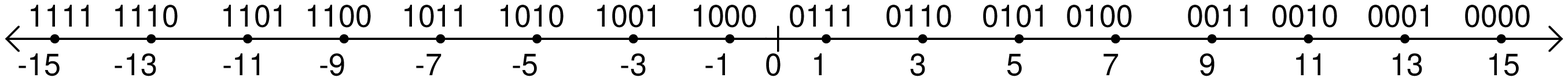}\\
  (a)\\
  \includegraphics[scale=.7]{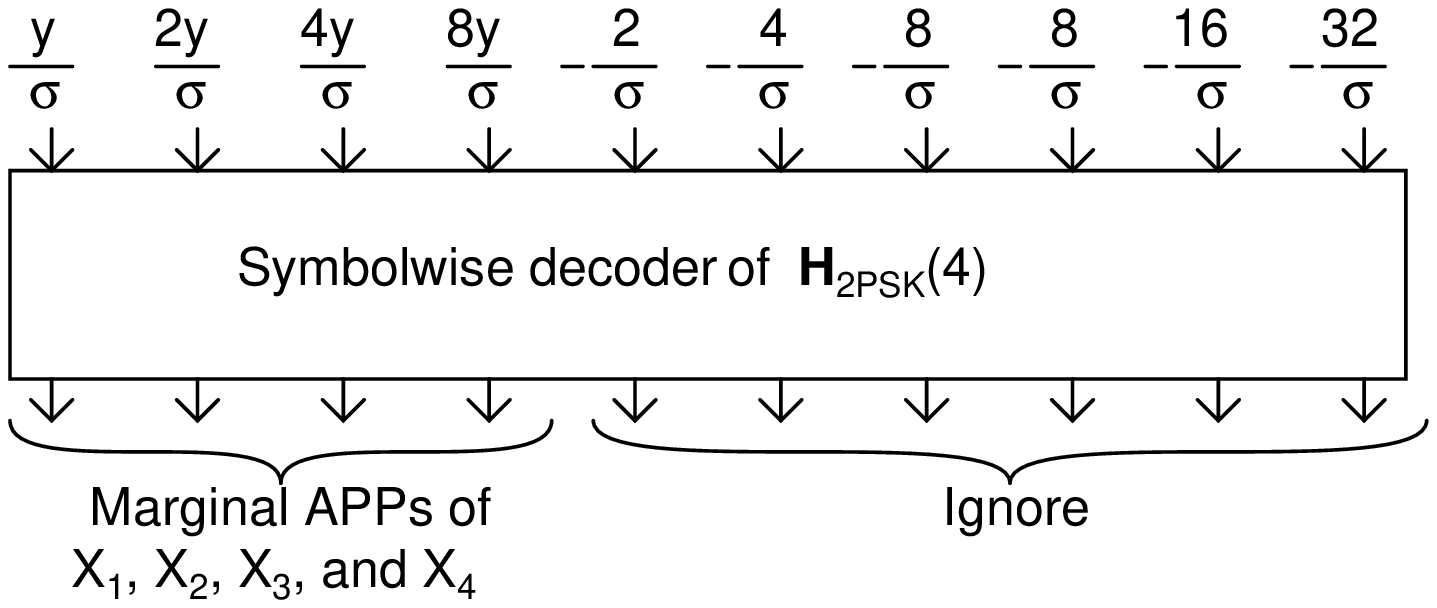}\\
  (b)
\end{tabular}
  \caption{(a) Constellation diagram of naturally mapped 16-PAM modulation. 
(b) Computing marginal APPs from  the received symbol by using the  symbolwise decoder of $\vect{H}_{2PSK}(4)$. \label{natural16pamfig}}    
  \end{center}
\end{figure}

\section{Channel decoders as the detectors of  gray mapped M-PAM\label{graydemodsection}}

Naturally mapped M-PAM , whose detection by using a decoder is investigated in the previous
section, suffers from the fact that more than one bits may differ between two adjacent symbols. This 
problem is overcome with the gray mapping in which a one bit differs between two adjacent 
symbols. In this section detection of gray mapped M-PAM  by using a decoder
is investigated. Let $\gpam{N}{\vect{x}}$ be a function from $\fieldn{2}$ to $\real$
which maps binary valued vectors of length $N$ to $M\triangleq 2^{N}$ real amplitude 
values as in the gray mapped M-PAM, i.e. 
\begin{equation}
  \gpam{N}{\vect{x}}\triangleq \sum_{i=1}^{N}2^{i-1}\bpsk{\sum_{j=i}^{N}x_j} \virgul \label{gpamdef}
\end{equation}
where $\vect{x}=[x_1,x_2,\ldots,x_N]$ and the summation inside the $\beta(.)$
function takes places in $\field{2}$. Unfortunately, due this summation inside
the $\beta(.)$ function, detection of gray mapped M-PAM is not 
a special case MISO detection of BPSK as opposed to the detection of 
naturally mapped M-PAM. Hence, in order to determine the parity 
check matrix and inputs of the decoder to detect the M-PAM 
we need to obtain the canonical factorization of the joint a posteriori pmf. 

Assume that $\gpam{N}{\vect{X}}$ is transmitted through a discrete additive Gaussian noise channel and $Y$ is received. In other words,  
\begin{equation}
  Y=\gpam{N}{\vect{X}}+Z\virgul \label{gpamchannel}
\end{equation}
where $\vect{X}=[X_1,X_2,\ldots,X_N]$ and $Z$ is real Gaussian random variable with variance $\sigma^{2}$.
Let $p(\vect{x})$ denote the joint a posteriori probability $\Pr\{\vect{X}=\vect{x}|Y=y\}$. 
The canonical factorization of $p(\vect{x})$ can be obtained by following the generic procedures
explained in Chapter \ref{thechapter}. However, the canonical factorization of $p(\vect{x})$ 
can be obtained more easily by exploiting the relation between $\gpam{N}{\vect{x}}$ and $\pam{N}{\vect{x}}$.

The relation between $\gpam{N}{\vect{x}}$ and $\pam{N}{\vect{x}}$ can be expressed as in
\begin{equation}
\gpam{N}{\vect{x}} =\pam{N}{\vect{x}\vect{G}(N)} \virgul
\end{equation}
where  $\vect{G}(N)$ is the $N\times N$ matrix defined as 
\begin{equation}
  \vect{G}(N)\triangleq \left[ \begin{array}{cccc} 
      1 & 0 & \ldots & 0 \\
      1 & 1 & \ldots & 0 \\
    \vdots & \vdots & \ddots & \vdots \\
  1 & 1 & \ldots &1 \end{array}  \right] \textrm{.}
\end{equation}
Since $\vect{G}(N)$ is a reversible matrix, the canonical factorization of $p(\vect{x})$
can be derived from the canonical factorization of the APP $\Pr\{\vect{W}=\vect{w}|\pam{N}{\vect{W}}+Z=y\}$
by following the discussion in Section \ref{reversiblesection}.
 
Let $t(\vect{w})$ be the shorthand notation for the APP  $\Pr\{\vect{W}=\vect{w}|\pam{N}{\vect{W}}+Z=y\}$.
Since $t(\vect{w})$ represents the APP in the naturally mapped M-PAM case, its canonical
factorization is a special case of the canonical factorization given in (\ref{canonicalfactofapp})
with channel coefficients $h_i=2^{i-1}$ and $\psk{w}=\bpsk{w}$.
The $\gamma(w;\rho,\sigma)$ function for the BPSK modulation  is 
\begin{equation}
  \gamma(w;\rho,\sigma)=\exp\left(\frac{2\rho\bpsk{w}}{2\sigma^{2}}\right) \nokta
\end{equation}
Consequently, $r_i(w)$ and $r_{i,j}(w)$ in this specific case of (\ref{canonicalfactofapp}) are
\begin{eqnarray}
    r_i(w)&=&\n{\fieldq}{\gamma(w;2^{i-1}y,\sigma)}=\n{\fieldq}{\exp\left(\frac{2^{i-1}y\bpsk{w}}{2\sigma^2}\right)} \textrm{,} \\
  r_{i,j}(w) &=&\n{\fieldq}{\gamma(w;-2^{i-1}2^{j-1},\sigma)}=\n{\fieldq}{\exp\left(-\frac{2^{i+j-2}\bpsk{w}}{2\sigma^2}\right)}\textrm{.}
\end{eqnarray}
Finally, the canonical factorization of $t(\vect{w})$ is
\begin{equation}
t(\vect{w})=\n{\fieldqn}{\prod_{i=1}^{N} r_i(\vect{f}_i\vect{w}^{T})\prod_{j=2}^{N}\prod_{i=1}^{j-1}r_{i,j}(\vect{a}_{i,j}\vect{w}^{T})} \nokta
\end{equation}

Consequently, due to the discussion in Section \ref{reversiblesection} the canonical factorization of $p(\vect{x})$ is
\begin{equation}
  p(\vect{x})=\n{\fieldqn}{\prod_{i=1}^{N} r_i(\vect{f}_i\vect{G}(N)^{T}\vect{x}^{T})\prod_{j=2}^{N}
    \prod_{i=1}^{j-1}r_{i,j}(\vect{a}_{i,j}\vect{G}(N)^{T}\vect{x}^{T})} \nokta
\end{equation}
Let $\vect{b}_{i,j}$ defined as 
\begin{equation}
  \vect{b}_{i,j}\triangleq \sum_{k=i}^{j} \vect{f}_k \nokta 
\end{equation}
 $\vect{f}_i\vect{G}(N)^{T}$ and $\vect{a}_{i,j}\vect{G}(N)^{T}$ can be expressed by using 
$\vect{b}_{i,j}$ as 
\begin{eqnarray}
  \vect{f}_i\vect{G}(N)^{T}&=& \vect{b}_{i,N} \virgul \\
  \vect{a}_{i,j}\vect{G}(N)^{T}&=&\vect{b}_{i,j-1}\nokta
\end{eqnarray}
Consequently,
\begin{eqnarray}
  p(\vect{x})&=&\n{\fieldqn}{\prod_{i=1}^{N} r_i(\vect{b}_{i,N}\vect{x}^{T})\prod_{j=2}^{N}
    \prod_{i=1}^{j-1}r_{i,j}(\vect{b}_{i,j-1}\vect{x}^{T})} \\
  &=&\n{\fieldqn}{r_{N}(\vect{f}_{N}\vect{x}^{T})\prod_{i=1}^{N-1}r_{i,i+1}(\vect{f}_{i}\vect{x}^{T})
\prod_{i=1}^{N-1}r_i(\vect{b}_{i,N}\vect{x}^{T}) \prod_{j=2}^{N}
    \prod_{i=1}^{j-2}r_{i,j}(\vect{b}_{i,j-1}\vect{x}^{T})}            \nokta
\end{eqnarray}
This last form of the factorization clearly shows which  parity check coefficient vectors 
are of weight two or more. Then, following the discussion in Section \ref{specialdecoders}
the parity check matrix of this code whose decoder can be employed in the detection 
of gray mapped M-PAM is
\begin{equation}
  \vect{H}_{GRAY}(N)\triangleq \left[\begin{array}{c}
      \vect{b}_{1,2} \\
      \vect{b}_{1,3}\\
      \vect{b}_{2,3}\\
      \vdots \\
      \vect{b}_{1,N}\\
      \vect{b}_{2,N}\\
      \vdots\\
      \vect{b}_{N-1,N}
    \end{array}\quad \vect{I}_{\frac{N(N-1)}{2} \times\frac{N(N-1)}{2}}\right] \textrm{.}
\end{equation}
Notice that  the sizes of  $\vect{H}_{GRAY}(N)$ and $\vect{H}_{qPSK}(N)$ are same. 

The symbolwise and ML codeword decoders of the code  with parity check matrix $\vect{H}_{GRAY}(N)$
can be designed for BPSK modulation and AWGN channel. In order to achieve the desired detection
the inputs applied to this decoder should be a permuted version of the inputs
applied for the naturally mapped detection since the canonical factorization 
of $p(\vect{x})$ is derived from the canonical factorization of $t(\vect{w})$. 
The first $N$ of these inputs are
 \begin{equation}
\left[-\frac{2^02^1}{\sigma}\quad -\frac{2^12^2}{\sigma}\quad \ldots \quad -\frac{2^{N-2}2^{N-1}}{\sigma} \quad \frac{2^{N-1}y}{\sigma}  \right] 
\nonumber \nokta   
\end{equation}
The last $N-1$ of these inputs are
\begin{equation}
 \left[\frac{y}{\sigma}\quad \frac{2y}{\sigma} \quad\ldots \quad  \frac{2^{N-2}y}{\sigma}\right] \nonumber \nokta 
\end{equation}
The remaining $\frac{(N-2)(N-1)}{2})$ inputs in between are
\begin{equation}
\left[-\frac{2^02^2}{\sigma} \quad -\frac{2^02^3}{\sigma} \quad -\frac{2^12^3}{\sigma} \quad \ldots \quad 
-\frac{2^{0}2^{N-1}}{\sigma} \quad -\frac{2^{1}2^{N-1}}{\sigma} \quad \ldots \quad -\frac{2^{N-3}2^{N-1}}{\sigma} \right] \nonumber \nokta
\end{equation}

\begin{example}\label{grayexample}
This example shows how to compute marginal APPs of four bits
which are modulated with gray mapped 16-PAM and passed through
an AWGN channel with PSD $N_{0}/2=\sigma^2$. Constellation diagram 
of the  16-PAM modulation with gray mapping is shown in Figure \ref{gray16pamfig}-a. 
This example demonstrates an interesting property of demodulating gray mapped
M-PAM modulation with decoders. 

The parity check matrix of the code whose symbolwise decoder can be used
to compute marginal APPs of the individual bits is 
\begin{equation}
  \vect{H}_{GRAY}(4) = 
\left[\begin{array}{ccccccccccccccc}  
1& 1 & 0 & 0& 1 & 0 & 0 & 0 & 0 & 0 \\
1& 1 & 1 & 0&  0 & 1 & 0 & 0 & 0 & 0 \\
0& 1 & 1 & 0&  0 & 0 & 1 & 0 & 0 & 0 \\
1& 1 & 1 & 1&  0 & 0 & 0 & 1 & 0 & 0 \\
0& 1 & 1 & 1&  0 & 0 & 0 & 0 & 1 & 0 \\
0& 0 & 1 & 1&  0 & 0 & 0 & 0 & 0 & 1 
\end{array} \right] \textrm{.}
\end{equation}

If the received value is denoted with $y$ then the input vector that should be applied to this decoder is 
\begin{equation}
\left[-\frac{2}{\sigma}\quad -\frac{8}{\sigma}\quad -\frac{32}{\sigma}\quad \frac{8y}{\sigma}\quad 
-\frac{4}{\sigma}\quad -\frac{8}{\sigma}\quad -\frac{16}{\sigma}\quad \frac{y}{\sigma}\quad 
\frac{2y}{\sigma}\quad \frac{4y}{\sigma} \right]\nonumber \textrm{.}
\end{equation}

It is well known that carrying out row operations on the parity check matrix of a code
does not alter the code. Hence, we can carry out row operations on $\vect{H}_{GRAY}(4)$
and obtain an alternative parity check matrix for the code. Let $\vect{H}'$ be the
parity check matrix derived from $\vect{H}_{GRAY}(4)$ by  adding the first row onto 
second and fourth rows and then adding sixth row onto fourth and fifth rows, i.e.
\begin{equation}
  \vect{H}' = 
\left[\begin{array}{ccccccccccccccc}  
1& 1 & 0 & 0&  1 & 0 & 0 & 0 & 0 & 0 \\
0& 0 & 1 & 0&  1 & 1 & 0 & 0 & 0 & 0 \\
0& 1 & 1 & 0&  0 & 0 & 1 & 0 & 0 & 0 \\
0& 0 & 0 & 0&  1 & 0 & 0 & 1 & 0 & 1 \\
0& 1 & 0 & 0&  0 & 0 & 0 & 0 & 1 & 1 \\
0& 0 & 1 & 1&  0 & 0 & 0 & 0 & 0 & 1 
\end{array} \right] \nokta
\end{equation}
Since $\vect{H}'$ and $\vect{H}_{GRAY}(4)$ are the parity check matrices of 
the same code, we can use the decoder designed for either $\vect{H}'$ or
$\vect{H}_{GRAY}(4)$ to compute soft outputs in gray mapped 16-PAM modulation. 

 Notice that all rows $\vect{H}'$ are of weight $3$. Moreover, four columns of $\vect{H}'$
are of weight $3$ and the remaining six columns are of weight one. $\vect{H}'$ shares
these properties with $\vect{H}_{2PSK}(4)$. Furthermore, let $\vect{H}''$ be the 
parity check matrix derived from $\vect{H}'$ by replacing the first column with fifth
and fourth column with tenth, i.e. 
\begin{equation}
  \vect{H}'' = 
\left[\begin{array}{ccccccccccccccc}  
1& 1 & 0 & 0&  1 & 0 & 0 & 0 & 0 & 0 \\
1& 0 & 1 & 0&  0 & 1 & 0 & 0 & 0 & 0 \\
0& 1 & 1 & 0&  0 & 0 & 1 & 0 & 0 & 0 \\
1& 0 & 0 & 1&  0 & 0 & 0 & 1 & 0 & 0 \\
0& 1 & 0 & 1&  0 & 0 & 0 & 0 & 1 & 0 \\
0& 0 & 1 & 1&  0 & 0 & 0 & 0 & 0 & 1 
\end{array} \right] \nokta
\end{equation}
$\vect{H}''$ describes a code whose codewords are
permuted form of the codewords of the code described 
by $\vect{H}'$. Hence, we can also use the decoder
designed for $\vect{H}''$ to compute soft outputs in gray mapped
16-PAM modulation. In order to achieve this demodulation
it is necessary to permute the inputs applied
to the decoder designed for $\vect{H}_{GRAY}(4)$ before applying
to the   decoder designed for $\vect{H}''$ in the same order
as the column permutations applied while passing from $\vect{H}'$
to $\vect{H}''$. Hence, the inputs that should be applied
to this decoder are
\begin{equation}
\left[-\frac{4}{\sigma}\quad -\frac{8}{\sigma}\quad -\frac{32}{\sigma}\quad \frac{4y}{\sigma}\quad 
-\frac{2}{\sigma}\quad -\frac{8}{\sigma}\quad -\frac{16}{\sigma}\quad \frac{y}{\sigma}\quad 
\frac{2y}{\sigma}\quad \frac{8y}{\sigma} \right]\nonumber \textrm{.}
\end{equation}

The \textbf{\emph{interesting point}} in here is that $\vect{H}''$ \textbf{\emph{is equal to}} $\vect{H}_{2PSK}(4)$.
Therefore, the symbolwise decoder of the parity check matrix $\vect{H}_{2PSK}(4)$
can be used to compute marginal APPs for both naturally mapped and gray mapped
16-PAM modulation. The decoder can be configured to natural mapping or
gray mapping by permuting the inputs. Computing the soft outputs in of gray mapped 16-PAM modulation depicted in Figure \ref{gray16pamfig}-c.

\end{example}

\begin{figure}
  \begin{center}
\begin{tabular}{c}
  \includegraphics[scale=.5]{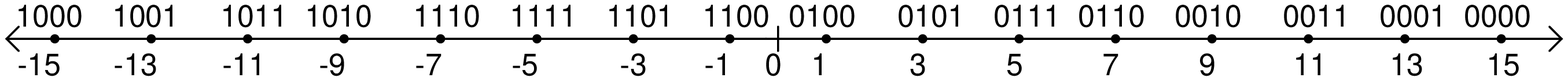}\\
  (a)\\
  \includegraphics[scale=.7]{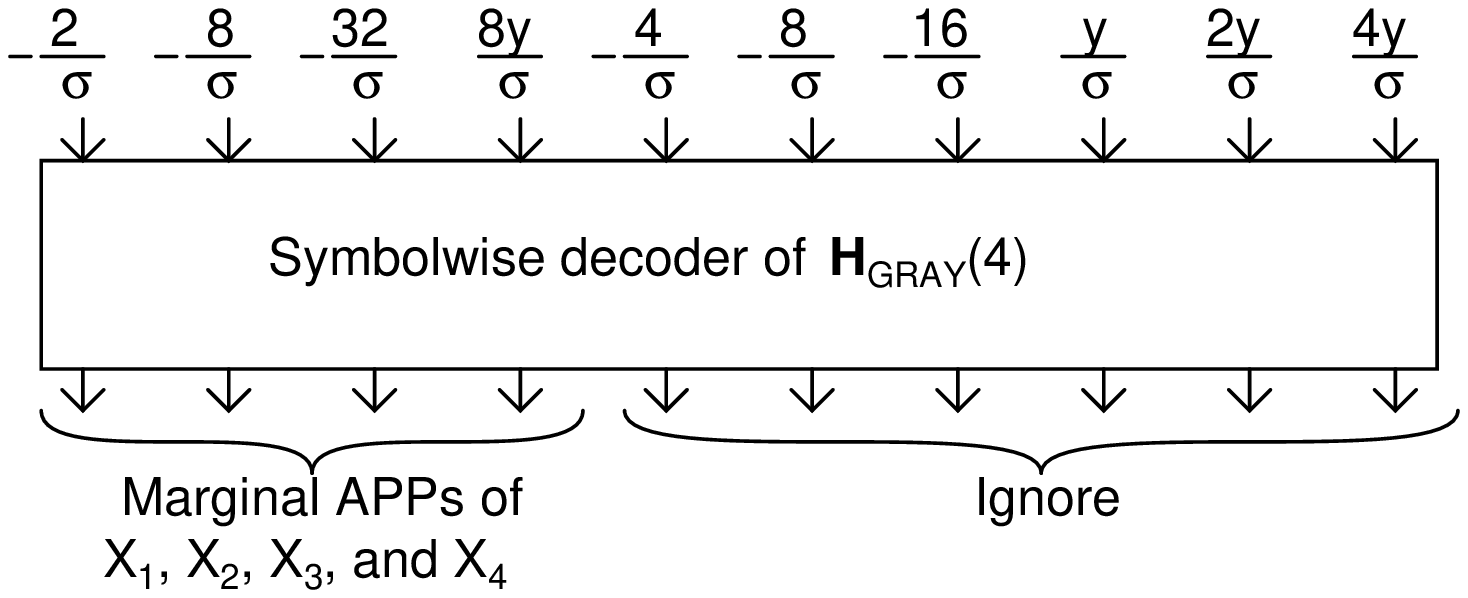}\\
  (b)\\
  \includegraphics[scale=.7]{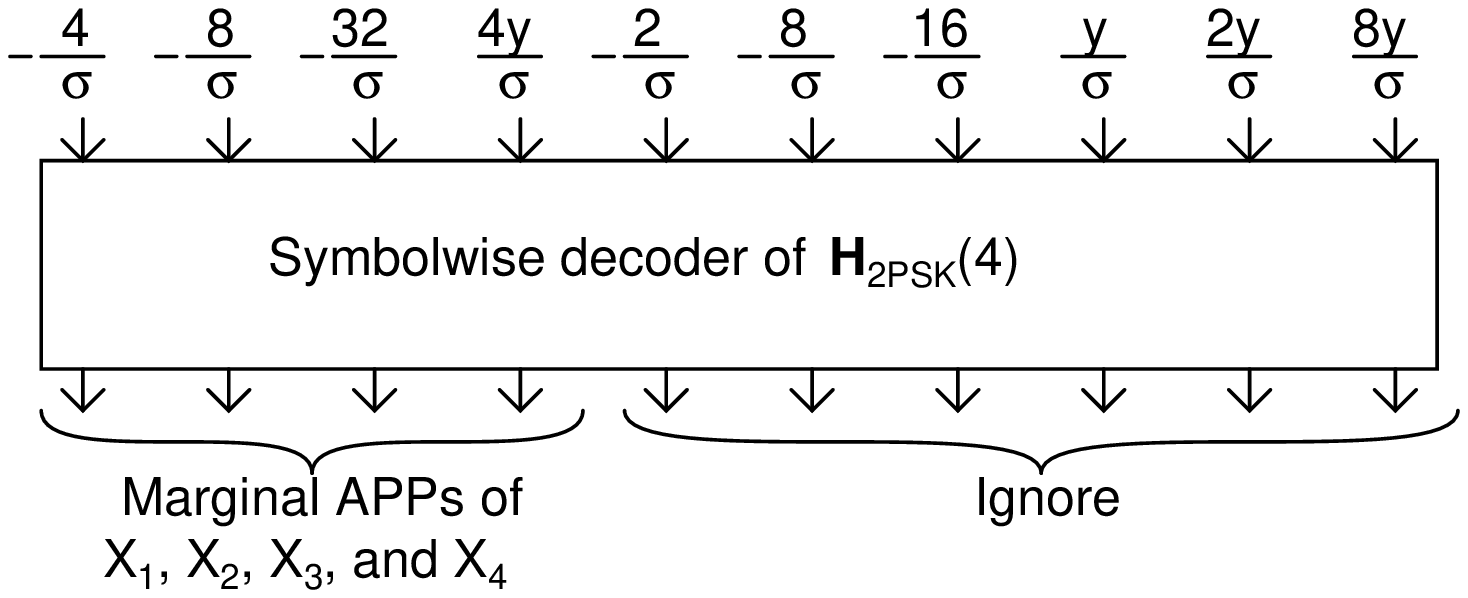}\\
  (c)\\
  
\end{tabular}
  \caption{(a) Constellation diagram of gray mapped 16-PAM modulation. 
(b) Computing marginal APPs from  the received symbol by using the  symbolwise decoder of $\vect{H}_{GRAY}(4)$.
(c) Computing marginal APPs by using the  symbolwise decoder of $\vect{H}_{2PSK}(4)$.
 \label{gray16pamfig}}    
  \end{center}
\end{figure}

The example above shows that the decoder of $\vect{H}_{2PSK}(4)$ can be used to demodulate 
both naturally mapped and gray mapped 16-PAM modulation. The following theorem
states that this is true not only for 16-PAM but for any M-PAM modulation. 

\begin{theorem} \label{naturalgraythm}
  There exist a sequence of row operations such that  performing these row operations
on $\vect{H}_{GRAY}(N)$ leads to $\vect{H}_{2PSK}(N)$ with some columns permuted. 
\end{theorem}

A constructive proof is given in Appendix \ref{naturalgraythmproof}.

\section{MIMO detection by using channel decoders\label{mimodemodsection}}

In this section we show how to employ channel decoders for multiple-input multiple-output (MIMO)
detection. The analysis is presented for QPSK modulation but is straightforward
to extend method to any other PAM or QAM modulation. 

\subsection{System Model}

Let a random vector $\vect{X}_k=[X_{2k-1} , X_{2k}]$  is mapped to a complex symbol
$W_k$ via the function $\qpsk{.}$ as in 
\begin{equation}
  W_k\triangleq \qpsk{\vect{X}_k} \virgul
\end{equation}  
where $\qpsk{.}$ represents the gray mapped QPSK modulation and defined  as
\begin{equation}
  \qpsk{\vect{x}}\triangleq \left\{\begin{array}{ll} 
        1, &\vect{x}=[0 \quad 0 ]\\
        j, &\vect{x}=[0 \quad 1]\\
        -1, &\vect{x}=[1 \quad 1]\\
        -j, &\vect{x}=[1 \quad 0]
\end{array} \right.  \virgul
\end{equation}
and $j$ is the square root of $-1$.  
The constellation diagram of gray mapped QPSK modulation is shown in Figure \ref{Qpskgray}. Furthermore, let a random vector 
$\vect{W}=[W_1,W_2,\ldots,W_{N_t}]^T$ is passed through an $N_r\times N_t$ MIMO channel with independent 
circularly symmetric Gaussian noise and the received vector is $\vect{Y}$.
In other words, 
\begin{equation}
  \vect{Y}= \vect{H}_c\vect{W}+\vect{Z} \virgul
\end{equation}
where $\vect{H}_c$ is the $N_r \times N_t$ channel coefficient matrix, $\vect{Z}$ is the $N_r\times 1$ noise vector 
consisting of independent, zero mean, circularly symmetric normal  distributed random variables of variance $2\sigma^{2}$. 

\begin{figure}
\begin{center}
  \includegraphics[scale=.6]{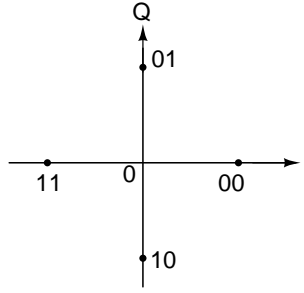}
  \caption{The QPSK constellation with gray mapping\label{Qpskgray}}
\end{center}
\end{figure}

ML MIMO detection is the task of  determining the configuration $\vect{x}$ maximizes the likelihood function
 $\Pr\{\vect{Y}=\vect{y}|\vect{X}=\vect{x}\}$ where $\vect{X}$ is 
\begin{equation}
  \vect{X}\triangleq[\vect{X}_1 \quad \vect{X}_2 \quad \ldots \quad \vect{X}_{N_t}] \nokta 
\end{equation}
We assume that all  $\vect{X}$ is uniformly distributed. Hence, ML MIMO detection is equivalent to 
finding the configuration maximizing the APP $\Pr\{\vect{X}=\vect{x}|\vect{Y}=\vect{y}\}$. Soft output
MIMO detection is the task of computing the marginal APPs $\Pr\{X_{k}=x|\vect{Y}=\vect{y}\}$.

\subsection{The decoders which can be used in MIMO detection with QPSK signaling}

The first step in determining the parity check matrix of the decoders which can be employed
as MIMO demodulators is determining the canonical factorization of the APP.  The APP 
$\Pr\{\vect{X}=\vect{x}|\vect{Y}=\vect{y}\}$ is
\begin{equation}
  \Pr\{\vect{X}=\vect{x}|\vect{Y}=\vect{y}\}= \n{\fieldnn{q}{2N_t}}{\exp\left(- \frac{\norm{\vect{y}-\vect{H}_c\vect{w}}^2}{2\sigma^{2}}\right)}
\end{equation}
where $\vect{x}$ is $[\vect{x}_1, \vect{x}_2,\ldots, \vect{x}_{N_t}]$, $\vect{x}_{k}$ is $[x_{2k-1}, x_{2k}]$, and
$\vect{w}$ is $[\qpsk{\vect{x}_{1}},\qpsk{\vect{x}_2},\ldots, \qpsk{\vect{x}_{N_t}} ]^{T}$. 
As shown in Appendix \ref{mimoappfact},  this APP can factored as
\begin{equation}
\begin{split}
\Pr\{\vect{X}=\vect{x}|\vect{Y}=\vect{y}\} \propto & \prod_{k=1}^{N_t}\gamma\left(x_{2k-1};\frac{\Re{u_{k}}+\Im{u_k}}{2},\sigma\right) 
\gamma\left(x_{2k};\frac{\Re{u_{k}}-\Im{u_k}}{2},\sigma\right) \\
&\cdot \prod_{k=2}^{N_t}\prod_{l=1}^{k-1}\gamma\left(x_{2k-1}+x_{2l-1};-\frac{\Re{(\vect{R})_{k,l}}}{2},\sigma\right)
\gamma\left(x_{2k-1}+x_{2l};-\frac{\Im{(\vect{R})_{k,l}}}{2},\sigma\right)\\
&\cdot \prod_{k=2}^{N_t}\prod_{l=1}^{k-1}\gamma\left(x_{2k}+x_{2l-1};\frac{\Im{(\vect{R})_{k,l}}}{2},\sigma\right)
\gamma\left(x_{2k}+x_{2l};-\frac{\Re{(\vect{R})_{k,l}}}{2},\sigma\right)
\end{split}
\label{mimoappfact1} \virgul
\end{equation}  
where $\vect{R}$ and $\vect{u}$ are 
\begin{eqnarray}
  \vect{R}&\triangleq&\vect{H}_{c}^{H}\vect{H}_c \virgul \label{rinputdef}\\  
  \vect{u}&\triangleq &\vect{H}_c^H\vect{y} \virgul \label{uinputdef}
\end{eqnarray}
and $(\vect{R})_{k,l}$ denotes $k$ by $l^{th}$ entry of $\vect{R}$ and $u_k$ is the 
$k^{th}$ component of $\vect{u}$.

The factorization above can be expressed by using the $\vect{a}_{k,l}$ vectors defined in (\ref{aijdef})
as
\begin{equation}
\begin{split}
\Pr\{\vect{X}=\vect{x}|\vect{Y}=\vect{y}\} \propto & \prod_{k=1}^{N_t}\gamma\left(\vect{f}_{2k-1}\vect{x}^{T};\frac{\Re{u_{k}}+\Im{u_k}}{2},\sigma\right) 
\gamma\left(\vect{f}_{2k}\vect{x}^{T};\frac{\Re{u_{k}}-\Im{u_k}}{2},\sigma\right) \\
&\cdot \prod_{k=2}^{N_t}\prod_{l=1}^{k-1}\gamma\left(\vect{a}_{2k-1,2l-1}\vect{x}^{T};-\frac{\Re{(\vect{R})_{k,l}}}{2},\sigma\right)
\gamma\left(\vect{a}_{2k-1,2l}\vect{x}^{T};-\frac{\Im{(\vect{R})_{k,l}}}{2},\sigma\right)\\
&\cdot \prod_{k=2}^{N_t}\prod_{l=1}^{k-1}\gamma\left(\vect{a}_{2k,2l-1}\vect{x}^{T};\frac{\Im{(\vect{R})_{k,l}}}{2},\sigma\right)
\gamma\left(\vect{a}_{2k,2l}\vect{x}^{T};-\frac{\Re{(\vect{R})_{k,l}}}{2},\sigma\right)
\end{split} \nokta \label{mimoappfact3}
\end{equation}  

The only remaining step in the derivation of canonical factorization is to normalize all of the factor functions 
existing above. We omit this obvious step for the sake of neatness. 
This factorization leads to  the  parity check matrix of the decoder which can be used in MIMO detection given in 
\begin{equation}
\vect{H}_{MIMO,QPSK}(N_t)\triangleq \left[  \begin{array}{c} 
    \vect{L}(1,N_t) \\
    \vect{L}(2,N_t) \\
    \ldots \\   
    \vect{L}(N_t-1,N_t) \\    
  \end{array} \vect{I}_{2N_t(N_t-1)\times 2N_t(N_t-1)}\right] \virgul \label{hmimodefinition}
\end{equation} 
where $\vect{L}(k,N_t)$ is 
\begin{equation}
  \vect{L}(k,N_t)\triangleq \left[  \begin{array}{c} 
\vect{a}_{1,2k+1}\\
\vect{a}_{2,2k+1} \\
\ldots \\
\vect{a}_{2k,2k+1}\\
\vect{a}_{1,2k+2}\\
\vect{a}_{2,2k+2} \\
\ldots \\
\vect{a}_{2k,2k+2}
\end{array}\right] \textrm{.} \label{ldef}
\end{equation}
As in the previous sections we can use a decoder designed for BPSK modulation and AWGN channel for MIMO detection.
The inputs that should be applied to this decoder to achieve MIMO detection are the first parameters  after the semicolon
divided by the second parameters of the $\gamma(.;.,.)$
functions in the factorization given in (\ref{mimoappfact3}).

\begin{figure}
  \begin{center}
\begin{tabular}{c}
  \includegraphics[scale=.7]{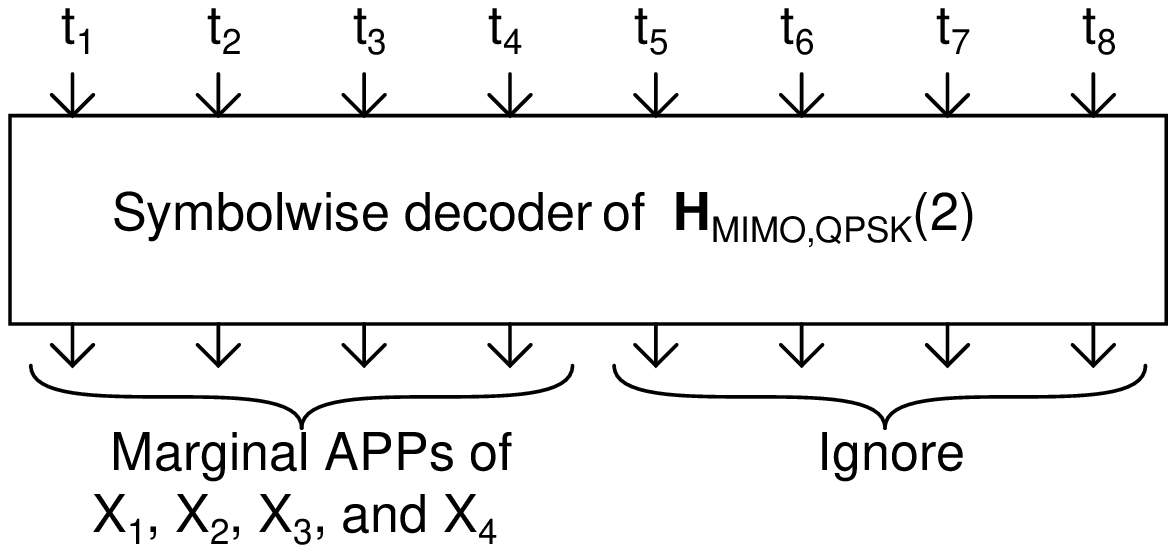}\\
  (a)\\
  \includegraphics[scale=.7]{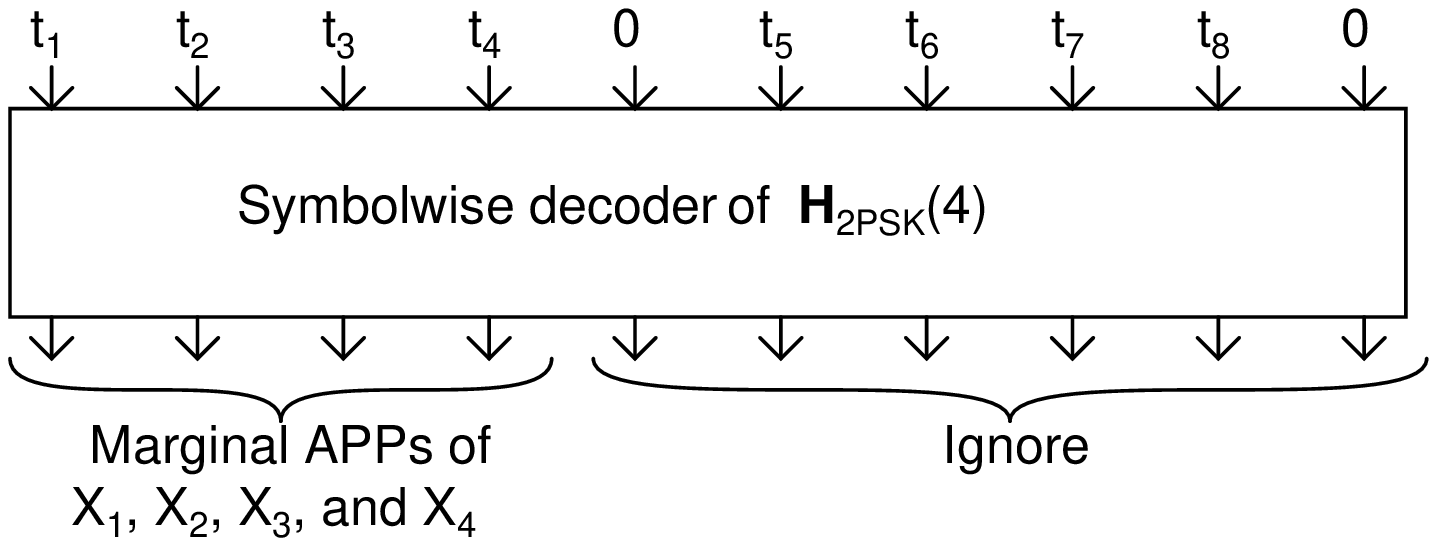}\\
  (b)\\
\end{tabular}
  \caption{Computing the marginal APPs in a MIMO system by using two different decoders. (a) By using the 
decoder of  $\vect{H}_{MIMO,QPSK}(2)$. (b) By using the decoder of  $\vect{H}_{2PSK}(4)$.\label{mimofig}}
\end{center}
\end{figure}

\begin{example}\label{mimoexample} This example shows how to compute marginal APPs of four bits
which are first  modulated with QPSK modulation and passed through a $2\times 2$ MIMO channel 
with channel coefficient matrix $\vect{H}_c$ and noise variance $2\sigma^{2}$. 
The parity check matrix of the symbolwise decoder which can be used for this purpose
is 
\begin{equation}
  \vect{H}_{MIMO,QPSK}(2)=
\left[\begin{array}{ccccccccccccc}  
1& 0 & 1 & 0&  1 & 0 & 0 & 0 \\
0& 1 & 1 & 0&  0 & 1 & 0 & 0  \\
1& 0 & 0 & 1&  0 & 0 & 1 & 0  \\
0& 1 & 0 & 1&  0 & 0 & 0 & 1  \\
\end{array} \right] \nokta
\end{equation}
Let the vector $\vect{t}=[t_1,t_2,\ldots,t_8]$ be
\begin{equation}
\begin{split}
  \vect{t}= \bigg[& \frac{\Re{u_{1}}+\Im{u_1}}{2\sigma}\quad   \frac{\Re{u_{1}}-\Im{u_1}}{2\sigma} \quad
\frac{\Re{u_{2}}+\Im{u_2}}{2\sigma}\quad   \frac{\Re{u_{2}}-\Im{u_2}}{2\sigma}\\
 & -\frac{\Re{(\vect{R})_{1,2}}}{2\sigma}\quad -\frac{\Im{(\vect{R})_{1,2}}}{2\sigma}\quad \frac{\Im{(\vect{R})_{1,2}}}{2\sigma}
\quad  -\frac{\Re{(\vect{R})_{1,2}}}{2\sigma} \bigg] \virgul
\end{split}
\end{equation}
where $\vect{R}=\vect{H}_{c}^{H}\vect{H}_{c}$ and $\vect{u}=\vect{H}_c\vect{y}$. This $\vect{t}$ vector is the 
vector that must be applied to the decoder. 

Notice that  $\vect{H}_{MIMO,QPSK}(2)$ is a sub-matrix of $\vect{H}_{2PSK}(4)$. Therefore, the symbolwise
decoder of $\vect{H}_{2PSK}(4)$ can also be used to compute marginal APP probabilities in MIMO detection. 
The inputs that must be applied in this case are given below. 
\begin{equation}
\begin{split}
  \bigg[& \frac{\Re{u_{1}}+\Im{u_1}}{2\sigma}\quad   \frac{\Re{u_{1}}-\Im{u_1}}{2\sigma} \quad
\frac{\Re{u_{2}}+\Im{u_2}}{2\sigma}\quad   \frac{\Re{u_{2}}-\Im{u_2}}{2\sigma}\\
 & 0 \quad -\frac{\Re{(\vect{R})_{1,2}}}{2\sigma}\quad -\frac{\Im{(\vect{R})_{1,2}}}{2\sigma}\quad \frac{\Im{(\vect{R})_{1,2}}}{2\sigma}
\quad  -\frac{\Re{(\vect{R})_{1,2}}}{2\sigma}\quad 0 \bigg] \virgul
\end{split}
\end{equation} 
 Notice that we added two zeros to the input vector when compared to the vector $\vect{t}$. These 
zeros correspond the missing columns in $\vect{H}_{MIMO,QPSK}(2)$ when compared to  $\vect{H}_{2PSK}(4)$. 
Computing the marginal APPs with these two decoders is depicted in Figure \ref{mimofig}. 
\end{example}

It is worth emphasizing that in Examples \ref{naturalexample}, \ref{grayexample}, and \ref{mimoexample}
the decoder of $\vect{H}_{2PSK}(4)$ is used for \emph{three different purposes}.   
\section{Usage of  decoders of tail biting convolutional codes as approximate  MIMO detectors\label{tailbitingsection}}

Trellis representation is mainly used for representing convolutional codes.
However, it is also possible to represent block codes with trellises \cite{mackaybook}. 
Block codes can also be represented with a special type of 
trellis which is the tail biting trellis. 
 Maximum trellis width in a tail biting trellis might  be as low as the square root
of the maximum width of the ordinary trellis representing the same code \cite{blahut,Wiberg}.

If a block code has a parity check matrix as in the form given below 
\begin{equation}
 \vect{H}= \left[\begin{array}{c} 
      ((\vect{L}_{r\times c}))_{0}\\
      ((\vect{L}_{r\times c}))_{1} \\ 
      \ldots \\   
      ((\vect{L}_{r\times c}))_{c-1} 
    \end{array} \vect{I}_{rc\times rc}\right] \virgul
\end{equation}
where $\vect{L}_{r\times c}$ is any $r\times c$ matrix and $((\vect{L}))_i$ denotes
cyclically shifting the columns of $\vect{L}$ towards right $i$ times, 
then it is called a tail biting convolutional code of rate $1/(r+1)$.  For instance,
the Golay code is of this type \cite{blahut}. Tail biting convolutional 
codes can be encoded by the encoders of the convolutional codes
by applying the data bits cyclically. 

The tail biting convolutional codes
have simple approximate decoders enjoying low complexity \cite{blahut,Wiberg}.  
 Hence, there are many studies and standards,
such as LTE, exploiting this reduction in complexity and simplicity of the tail biting trellises. Even an analog
implementation of such a decoder is proposed in \cite{loeligeranalog}.

In this section we show that the MIMO detection problem can be handled by the decoder of a tail biting
convolutional code. The characteristics of this code depend on the number of  transmitting antennae 
and modulation used. We are going to analyze the MIMO detectors for QPSK modulation  as we did in the previous section, 
although it is possible to generalize the technique to other QAM and PAM modulations as well.

We are going to use the same channel model and notation  as in the previous section. That model 
lead us the parity check matrix  $\vect{H}_{MIMO,QPSK}(N)$ given in (\ref{hmimodefinition}). 
This parity check matrix hardly looks like the parity check matrix of a tail biting convolutional code. 

Let a permutation matrix $\vect{P}$ is defined as in 
\begin{equation}
  \vect{P}\triangleq \left[\vect{f}_1^{T} \quad \vect{f}_3^{T} \quad \ldots \quad \vect{f}_{2N_t-1}^{T} 
    \quad \vect{f}_2^{T} \quad \vect{f}_4^{T} \quad \ldots \quad \vect{f}_{2N}^{T}  \right] \nokta
\end{equation}
Furthermore, let $\vect{V}$ be obtained by permuting $\vect{X}$ as in
\begin{equation}
  \vect{V}\triangleq \vect{X}\vect{P} \nokta \label{vxrelation}
\end{equation}

Since $\vect{V}$ is a permutation of $\vect{X}$, maximizing (marginalizing) the APP $\Pr\{\vect{X}=\vect{x}| \vect{Y}=\vect{y} \}$
  is equivalent to maximizing (marginalizing) the APP $\Pr\{\vect{V}=\vect{v}| \vect{Y}=\vect{y} \}$. 
Let $t(\vect{v})$ be a shorthand notation for $\Pr\{\vect{V}=\vect{v}| \vect{Y}=\vect{y} \}$. Then the factorization 
of $t(\vect{v})$ can be derived from the factorization (\ref{mimoappfact3}) as
\begin{equation}
\begin{split}
t(\vect{v}) \propto & \prod_{k=1}^{N_t}\gamma\left(\vect{f}_{2k-1}\vect{P}\vect{v}^{T};\frac{\Re{u_{k}}+\Im{u_k}}{2},\sigma\right) 
\gamma\left(\vect{f}_{2k}\vect{P}\vect{v}^{T};\frac{\Re{u_{k}}-\Im{u_k}}{2},\sigma\right) \\
&\cdot \prod_{k=2}^{N_t}\prod_{l=1}^{k-1}\gamma\left(\vect{a}_{2k-1,2l-1}\vect{P}\vect{v}^{T};-\frac{\Re{(\vect{R})_{k,l}}}{2},\sigma\right)
\gamma\left(\vect{a}_{2k-1,2l}\vect{P}\vect{v}^{T};-\frac{\Im{(\vect{R})_{k,l}}}{2},\sigma\right)\\
&\cdot \prod_{k=2}^{N_t}\prod_{l=1}^{k-1}\gamma\left(\vect{a}_{2k,2l-1}\vect{P}\vect{v}^{T};\frac{\Im{(\vect{R})_{k,l}}}{2},\sigma\right)
\gamma\left(\vect{a}_{2k,2l}\vect{P}\vect{v}^{T};-\frac{\Re{(\vect{R})_{k,l}}}{2},\sigma\right)
\end{split} \virgul 
\end{equation}  
since $(\vect{P}^{-1})^{T}=\vect{P}$. 
Consequently, a parity check matrix
whose ML codeword (symbolwise) decoder can be employed in maximization (marginalization) of  $\Pr\{\vect{V}=\vect{v}| \vect{Y}=\vect{y} \}$
is 
\begin{equation}
  \vect{H}_{\vect{V}}(N_t) \triangleq  \left[\vect{B}(N_t) \quad \vect{I}_{2N_t(N_t-1)\times 2N_t(N_t-1)}\right] 
\end{equation}
where $\vect{B}(N)$  is 
\begin{equation}
  \vect{B}(N) \triangleq  \left[
 \begin{array}{c} 
    \vect{L}(1,N)\vect{P} \\
    \vect{L}(2,N)\vect{P} \\
    \ldots \\   
    \vect{L}(N-1,N)\vect{P} \\    
  \end{array}
  \right] \nokta
\end{equation}
Other alternative parity check matrices whose decoder can be employed in performing inference on 
$\Pr\{\vect{V}=\vect{v}| \vect{Y}=\vect{y} \}$ are in the form of 
\begin{equation}
\left[\vect{B}'(N_t) \quad \vect{I}_{2N_t(N_t-1)\times 2N_t(N_t-1)}\right] \nonumber \virgul
\end{equation}
where  $\vect{B}'(N_t)$ is derived from $\vect{B}(N_t)$ by permuting rows (not columns this time). 
Fortunately, there exists a special row permutation which forms $\vect{B}(N_t)$ into the form given in 
\begin{equation}
 \vect{B}_{TB}(N_t) \triangleq \left[\begin{array}{c} 
      ((\ \vect{L}_{TB}(N_t)\quad \vect{0}_{(N_t-1)\times N_t  } ))_{0}\\
      ((\ \vect{L}_{TB}(N_t)\quad \vect{0}_{(N_t-1)\times N_t  } ))_{1} \\ 
      \ldots \\   
      ((\ \vect{L}_{TB}(N_t)\quad \vect{0}_{(N_t-1)\times N_t  } ))_{2N_t} 
      \end{array} \right] \virgul
\end{equation}
where $\vect{L}_{TB}(N)$ is 
\begin{equation}
  \vect{L}_{TB}(N)\triangleq \left[\vect{I}_{(N_t-1)\times (N_t-1)} \quad \vect{1}_{(N_t-1)\times  1}  \right] \textrm{.}
\end{equation}
Consequently, the decoders of the  parity check matrix given in 
\begin{equation}
  \vect{H}_{TB,MIMO}(N_t)=\left[\vect{B}_{TB}(N_t) \quad  \vect{I}_{2N_t(N_t-1)\times 2N_t(N_t-1)}\right]
\end{equation}
can be employed in performing inference
on $\Pr\{\vect{V}=\vect{v}| \vect{Y}=\vect{y} \}$. $\vect{H}_{TB,MIMO}(N_t)$ is the parity 
check matrix of the tail biting convolutional code of rate $(1/(N_{t}))$ and constraint length $N_t$, whose encoder is shown in Figure \ref{tbconvenc}.

\begin{figure}
\begin{center}
\includegraphics[scale=.6]{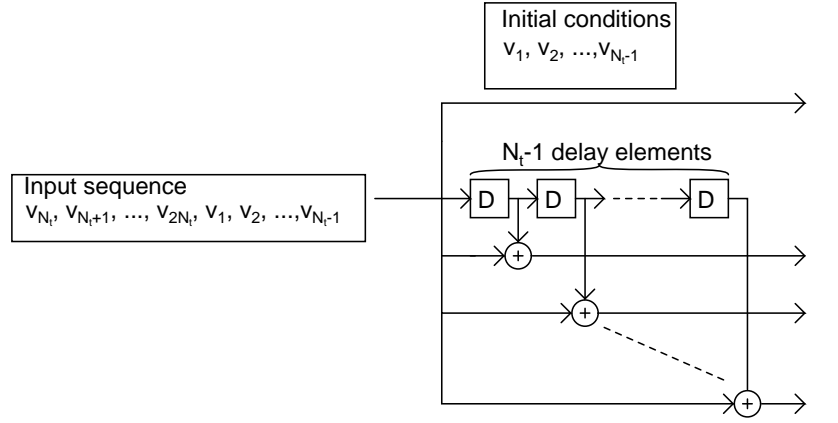}
  \caption[The encoder of the tail biting convolutional code whose decoder can be used as a $N_r\times N_t$ MIMO detector.]
  {The encoder of the tail biting convolutional code whose decoder can be used as the detector MIMO system 
    with $N_t$ transmit and $N_r$ receiving antennae.
    Notice that as opposed to ordinary convolutional encoders 
    the encoder does not initiate from the all zero state. 
    Tail biting nature of the decoder arises from the fact that after all the input sequence is applied
    the decoder returns to its initial condition. 
\label{tbconvenc}}
\end{center}
\end{figure}

\begin{example}
  In this example we are demonstrate that $\vect{B}_{TB}(N_t)$ can be derived from $\vect{B}(N_t)$ by permuting
rows for cases $N_t=2$ and $N_t=3$. 

For $N_t=2$, $\vect{B}(N_t)$  is equal to $\vect{L}(1,2)\vect{P}$. By (\ref{ldef}), $\vect{L}(1,2)$ is
\begin{equation}
  \vect{L}(1,2)=\left[ \begin{array}{cccc} 
      1& 0 & 1 & 0 \\
      0& 1 & 1 & 0 \\ 
      1& 0 & 0 & 1 \\
      0& 1 & 0 & 1 \end{array} \right] \nokta
\end{equation} 
Consequently, $\vect{B}(2)$ is
\begin{equation}
  \vect{B}(2)=\left[ \begin{array}{cccc} 
      1& 1 & 0 & 0 \\
      0& 1 & 1 & 0 \\ 
      1& 0 & 0 & 1 \\
      0& 0 & 1 & 1 \end{array} \right] \nokta
\end{equation}
Changing the places of third and fourth rows gives  $\vect{B}_{TB}(2)$, which is 
\begin{equation}
  \vect{B}_{TB}(2)=\left[ \begin{array}{cccc} 
      1& 1 & 0 & 0 \\
      0& 1 & 1 & 0 \\ 
      0& 0 & 1 & 1 \\
      1& 0 & 0 & 1 
\end{array} \right] \nokta
\end{equation} 
The Tanner graph of the resulting $\vect{H}_{TB,MIMO}(2)= [\vect{B}_{TB}(2) \quad \vect{I}_{4\times 4}]$ is shown in 
Figure \ref{exampletannergraphs}-a.  

For $N_t=2$, $\vect{B}(N_t)$  is equal to $\left[\begin{array}{c} \vect{L}(1,3) \\\vect{L}(2,3) \end{array}\right]\vect{P}$
where $\left[\begin{array}{c} \vect{L}(1,3) \\\vect{L}(2,3) \end{array}\right]$ is 
\begin{equation}
\left[\begin{array}{c} \vect{L}(1,3) \\\vect{L}(2,3) \end{array}\right]=
\left[\begin{array}{cccccc} 
    1 & 0 & 1 & 0 & 0 & 0 \\    
    0 & 1 & 1 & 0 & 0 & 0 \\
    1 & 0 & 0 & 1 & 0 & 0 \\    
    0 & 1 & 0 & 1 & 0 & 0 \\
    1 & 0 & 0 & 0 & 1 & 0 \\    
    0 & 1 & 0 & 0 & 1 & 0 \\
    0 & 0 & 1 & 0 & 1 & 0 \\    
    0 & 0 & 0 & 1 & 1 & 0 \\
    1 & 0 & 0 & 0 & 0 & 1 \\    
    0 & 1 & 0 & 0 & 0 & 1 \\
    0 & 0 & 1 & 0 & 0 & 1 \\    
    0 & 0 & 0 & 1 & 0 & 1 \\
 \end{array}\right] \nokta
\end{equation}
Then $\vect{B}(3)$ is 
\begin{equation}
\vect{B}(3)=
\left[\begin{array}{cccccc} 
    1 & 1 & 0&0  & 0  & 0 \\    
    0 & 1 & 0&1  & 0  & 0 \\
    1 & 0 & 0&0  & 1  & 0 \\    
    0 & 0 & 0 &1  & 1  & 0 \\
    1 & 0 & 1 &0  & 0  & 0 \\    
    0 & 0 & 1 &1  & 0  & 0 \\
    0 & 1 & 1&0  & 0  & 0 \\    
    0 & 0 & 1&0  & 1  & 0 \\
    1 & 0 & 0 &0  & 0  & 1 \\    
    0 & 0 & 0 &1  & 0  & 1 \\
    0 & 1 & 0&0  & 0  & 1 \\    
    0 & 0 & 0 &0  & 1  & 1 \\
 \end{array}\right] \nokta
\end{equation}
Finally  carrying the  $5^{th}$, $7^{th}$, $2^{nd}$, $6^{th}$, $8^{th}$, $4^{th}$, $10^{th}$, $12^{th}$, $3^{rd}$, 
$9^{th}$, $11^{th}$, and $1^{st}$ rows  to $1^{st}$, $2^{nd}$, $\ldots$, $12^{th}$ rows gives  
  $\vect{B}_{TB}(3)$ as in 
\begin{equation}
\vect{B}_{TB}(3)=
\left[\begin{array}{cccccc} 
    1 & 0 & 1& 0  & 0  & 0 \\    
    0 & 1 & 1& 0  & 0  & 0 \\
    0 & 1 & 0& 1  & 0  & 0 \\    
    0 & 0 & 1 &1  & 0  & 0 \\
    0 & 0 & 1 &0  & 1  & 0 \\    
    0 & 0 & 0 &1  & 1  & 0 \\
    0 & 0 & 0& 1  & 0  & 1 \\    
    0 & 0 & 0& 0  & 1  & 1 \\
    1 & 0 & 0 &0  & 1  & 0 \\    
    1 & 0 & 0 &0  & 0  & 1 \\
    0 & 1 & 0& 0  & 0  & 1 \\    
    1 & 1 & 0 &0  & 0  & 0 \\
 \end{array}\right] \nokta
\end{equation}
The Wiberg style Tanner graph of the resulting $\vect{H}_{TB,MIMO}(3)= [\vect{B}_{TB}(3) \quad \vect{I}_{12\times 12}]$ is shown in 
Figure \ref{exampletannergraphs}-b.  
\end{example}

\begin{figure}
  \begin{center}
\begin{tabular}{c}
  \includegraphics[scale=.7]{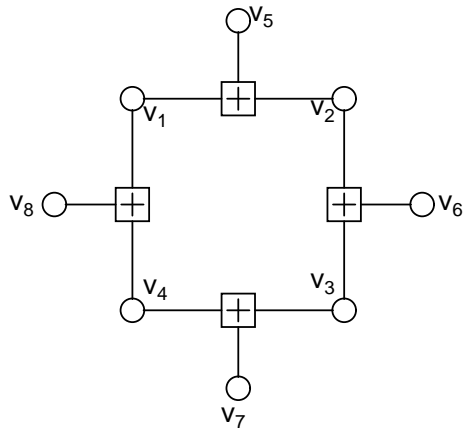}\\
  (a)\\
  \includegraphics[scale=.7]{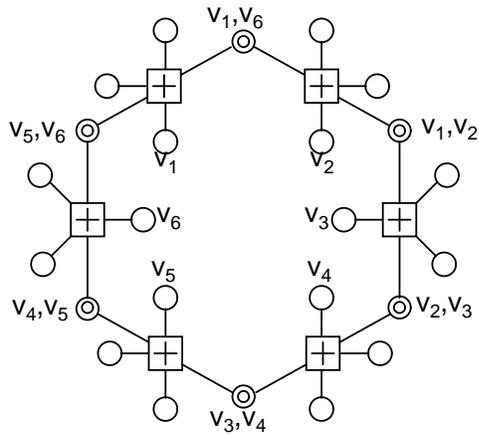}\\
  (b)\\
\end{tabular}

    \caption[The Tanner graphs of $\vect{H}_{TB,MIMO}(N_t)$ for $N_t=2$ and $N_t=3$.]
    {The Tanner graphs of $\vect{H}_{TB,MIMO}(N_t)$ for $N_t=2$ and $N_t=3$. (a) Tanner graph of  
$\vect{H}_{TB,MIMO}(2)$. (b) Wiberg style Tanner graph of $\vect{H}_{TB,MIMO}(3)$ \label{exampletannergraphs}}
  \end{center}
\end{figure}

\subsection{Using the decoding algorithms of tail biting convolutional codes for MIMO detection\label{approxmimodet}}

Since a tail biting trellis does not have a starting or ending state, Viterbi and
BCJR algorithms cannot be run on such trellises directly. To process a tail biting
trellis with Viterbi algorithm we need to run the Viterbi algorithm 
$\nu$ times on the trellis where $\nu$  denotes the trellis width. In each run, the Viterbi
algorithm  determines a candidate path which is the most probable path among the paths 
starting and ending on a certain state on the trellis. Then the most probable path can 
be chosen among the  $\nu$ candidate paths. Since the complexity of 
each running of the  Viterbi algorithm is $O(L\nu)$, where $L$ denotes the length of the trellis,
 the complexity of determining the most possible path with Viterbi algorithm is $O(L\nu^{2})$.
Recall that the complexity would be $O(L\nu)$ if the trellis were an ordinary trellis. 
Similar arguments are true for the BCJR algorithm as well. 

The complexity of ML codeword and exact symbolwise decoders of $\vect{H}_{TB,MIMO}(N_t)$
is $O(N_t2^{2N_t})$ as explained in the previous paragraph. The complexity 
of the trivial MIMO detection algorithm is $O(2^{2N_t})$. Hence, using the exact decoders
of $\vect{H}_{TB,MIMO}(N_t)$ for MIMO detection does not make  sense.   

Fortunately, tail biting convolutional codes have an \emph{approximate} symbolwise decoder.
This decoder operates by running BCJR algorithm on the tail biting trellis \emph{iteratively}.  
Equivalently, this decoder can be viewed as the iterative sum-product algorithm running 
on the Wiberg style Tanner graph an example of which is shown in Figure \ref{exampletannergraphs}-b.
Usually, a few iterations are sufficient to converge \cite{Wiberg}.
We propose implementing an approximate soft output MIMO detector by using this approximate 
symbolwise as the decoder $\vect{H}_{TB,MIMO}(N_t)$.
 Such a MIMO detector is also capable of using any a priori information available since it uses the BCJR algorithm. 
 The block diagram of this \emph{approximate} soft output MIMO detector is shown in Figure \ref{approxmimodetector}. 

\begin{figure}
  \begin{center}
   \includegraphics[scale=.42]{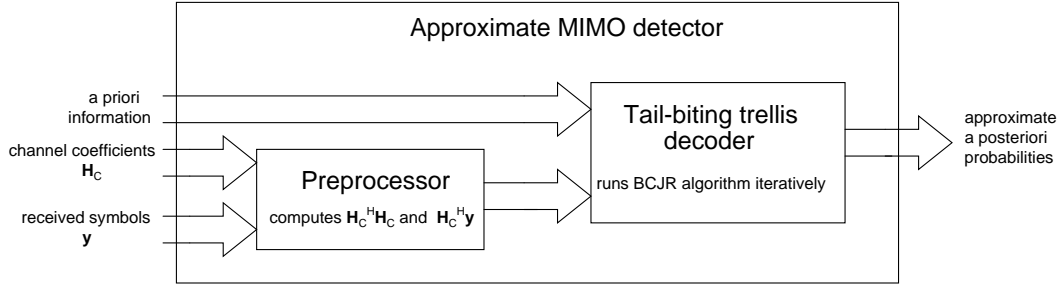}
    \caption{Block diagram of the proposed approximate soft output MIMO detector which uses the 
    approximate decoder of a tail biting convolutional code.\label{approxmimodetector}}
  \end{center}
\end{figure}
\subsection{Complexity issues}

There are two subtasks when the decoder mentioned above is employed as an approximate soft ouput MIMO detector. 
These tasks are the computation of the inputs applied to the decoder and processing
the decoder trellis.

As explained in Section \ref{mimodemodsection}, the inputs that must be applied
to the decoder of $\vect{H}_{TB,MIMO}$ are the components of $\vect{u}$
and the entries of $\vect{R}$ defined in (\ref{rinputdef}) and (\ref{uinputdef}) respectively.  
The computation of $\vect{u}$ has a complexity $O(N_{r}N_{t})$ whereas 
the computation of $\vect{R}$ has a complexity  $O(N_{t}^2N_{r})$. 

Processing the decoding trellis with the BCJR algorithm 
has a complexity $O(N_t2^{N_t})$.  This complexity 
is almost the square root of the complexity of the trivial ML and soft output MIMO  detectors which
is $O(2^{2N_t})$. From  a computer scientific point of view, this last 
component of the complexity might be dominant to the complexity of the computation of $\vect{R}$. 
However, in an engineering point of view computing
$\vect{R}$ is a more computationally demanding task than processing 
the decoding trellis for two reasons. First, in a practical scenario $N_{r}$ and
$N_t$ is eight at most. Hence, $N_t2^{N_t}$ and $N_t^{2}N_r$  are comparable 
in practical scenarios. Second,
the decoding trellis processing involves only additions and 
maximizations\footnote{We assume Max-Log-MAP approximation is used for the BCJR algorithm
running on the trellis.}
whereas computing $\vect{R}$ involves  complex multiplications
which require much more complex hardware than addition. Therefore, computing $\vect{R}$
is the most computationally demanding subtask of the proposed method. 
However, it should be noted that $\vect{R}$ is computed only once
for a constant $\vect{H}_c$. 

 The proposed technique, which employs a tail biting decoder as the MIMO detector, is comparable to other 
 sub optimal methods such as minimum mean square error (MMSE) or zero forcing (ZF) detectors in terms
of hardware complexity which both have a complexity $O(N^3)$ if $N_t=N_r=N$. Furthermore,
other sub optimal methods require matrix inversion. Although,
matrix inversion have complexity $O(N^3)$, it requires complex number divisions which 
require even more complex hardware than multiplication. Hence, the proposed technique 
still has an advantage in terms of hardware complexity over MMSE and ZF detectors.  

\subsection{Simulation Results \label{simulationresults1}}

We simulated the proposed approximate soft output MIMO detector for the $8\times 8$ Rayleigh fading MIMO channel. 
In this channel entries of $\vect{H}_c$ are  independent, zero-mean, circularly symmetric Gaussian random variables
where the variances of the real and imaginary parts are $1/2$. We assumed that 
$\vect{H}_c$ changes for every transmitted MIMO symbol and perfectly known at the
receiver side.  The noise vector added at the receiver
also consists of independent, zero-mean, circularly symmetric Gaussian random variables where the variances
of the  real and imaginary parts are $N_0/2$.  The signal to noise ratio (SNR) per receiving 
antenna is $E_{b}/N_{0}$. Since there are $N_{r}$ receiving antennae in a MIMO system the 
convention is to use $N_rE_b/N_0$  as SNR \cite{tenBrinkMIMOLDPC}.

\begin{figure}
  \begin{center}
    \includegraphics{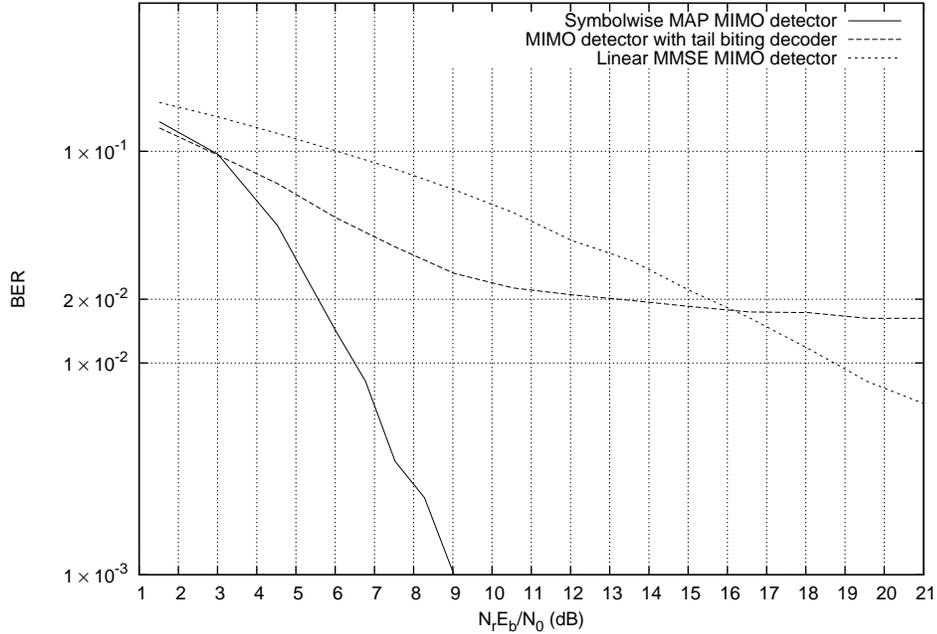}
  \caption{BER performances of the MIMO detector using the decoder of a  
    tail biting convolutional code, the symbolwise MAP MIMO detector, and the linear MMSE MIMO detector in a 
    Rayleigh fading $8\times 8$ MIMO channel. \label{expresult1}}
  \end{center}
\end{figure}

The bit error rate (BER) performance of the  proposed algorithm is shown in Figure \ref{expresult1}.
These results show that the proposed method has an 
unexpected poor performance when compared to the symbolwise MAP MIMO 
detector. Moreover, the proposed method exhibits an error floor as early as $2\times 10^{-2}$ level. 
The performance of the proposed algorithm is better than the linear minimum mean square 
error (MMSE) MIMO detector \cite{mimobook} until $16\ \textrm{dB}$. After $16\ \textrm{dB}$
the performance of the linear MMSE becomes better due to the early error floor of the proposed MIMO detector.     
We provide some comments on this unexpected performance in the next section and propose an
improvement in Section \ref{improvementsubsection}.

\subsection{Comments on the convergence of the sum-product algorithm on factor graphs with a single cycle}

The Wiberg style Tanner graph that represents the tail biting trellis contains only a single loop,
as in Figure \ref{exampletannergraphs}-a.
There are many studies in the sum-product algorithm literature which claim that the sum-product 
algorithm running on Tanner graph with  a single cycle always converges such as \cite{sp_on_simple_graphs,aji_single_cycle,weiss_empirical}. 
These studies also claim that the approximate marginals computed by  the sum-product algorithm 
is close to the exact marginals when the sum-product runs on these graphs. 
According to these studies, our proposed MIMO detector 
was supposed to converge at all times and it was expected to yield good results. 
However, our empirical results shown in Figure \ref{expresult1} do not agree with these 
expectations. 

Our experimental results verify that the sum-product algorithm running on a Tanner factor
graph with a single cycle always converges. However, in some cases this convergence
require as few as two or three iterations to converge whereas in some other rare cases
 it might require thousands of iterations.
A detailed analysis of the experimental results shows that the relatively high error floor in Figure \ref{expresult1} 
is caused by the cases in which the sum-product algorithm requires thousands of iterations to converge. 
Therefore, the sum-product algorithm 
produces good approximations of the exact marginals only if it converges in a few iterations. 
Otherwise,  the results generated by the sum-product algorithm is not a good approximation. 
We provide a numerical example in which sum-product algorithm requires thousands of iterations 
to converge below. 

\begin{example}
We provide the example for the Tanner graph shown in Figure \ref{exampletannergraphs}-a which 
is a factor graph with just a single cycle and contains only binary variable nodes. 
Let the inputs applied to the decoder represented by the Tanner graph shown in Figure \ref{exampletannergraphs}-a
designed for BPSK modulation and AWGN channel be 
\begin{equation}
\left[-55 \quad 60\quad -25\quad -20\quad 40\quad 55 \quad 40 \quad -55 \right]
\end{equation}
 If one runs the sum product algorithm on the Tanner graph shown in Figure \ref{exampletannergraphs}-a
with these inputs, it can be observed that the sum-product algorithm achieves a reasonable convergence
 at least after $3000$ iterations.  
Such an input settings can be observed in a scenario in which that decoder
is employed as a MIMO detector for a $2\times 2$ channel with coefficients 
\begin{equation}
  \vect{H}_c=\begin{bmatrix} 1.5j & 1-0.5j \\1+0.5j&-0.5-1.5j  \end{bmatrix} \nonumber
\end{equation}
and a sequence $[-j,1]$ is transmitted when noise has a variance $\sigma^{2}=0.01$.  

We would like to note that the likelihoods given above are very unlikely to be observed
in a real channel decoding problem. Therefore, such likelihoods is probably never
observed in  \cite{sp_on_simple_graphs,aji_single_cycle,weiss_empirical}. Hence, they claimed
that the sum-product algorithm produces good approximations for exact marginals if the sum-product
algorithm converges. Unfortunately, this claim is not quite true as this counter example shows.
\end{example}

Even if the sum-product algorithm produced good approximations in cases requiring thousands 
of iterations to converge, a practical MIMO detection algorithm cannot wait that 
much to complete the demodulation of a single MIMO symbol. Therefore,  
this late convergence problem requires a solution to develop a practical MIMO  detection 
algorithm with tail  biting decoders which we provide in the next section.

\subsection{Performance Improvements by using tail biting convolutional codes of 
  longer constraint length \label{improvementsubsection}}

Recall that the tail biting decoder of $\vect{H}_{TB,MIMO}(N_t)$ 
is used for performing inference 
on $\Pr\{\vect{V}=\vect{v}|\vect{Y}=\vect{y}\}$ where $\vect{V}$ was 
a permutation of $\vect{X}$ given by (\ref{vxrelation}).  This decoder
can only perform inference for this specific permutation of $\vect{X}$.

We define   an \emph{extended} version of parity check matrix $\vect{H}_{TB,MIMO}(N_t)$ 
as in
\begin{equation}
\vect{H}_{ETB,MIMO}(N_t)\triangleq \left[\vect{C}_{TB}(N_t) \quad  \vect{I}_{2N_t^2 \times 2N_t^2 } \right] \virgul
\end{equation} 
where $\vect{C}_{TB}(N_t)$ is 
\begin{equation}
 \vect{C}_{TB}(N_t) \triangleq \left[\begin{array}{c} 
      ((\ \vect{L}_{TB}(N_{t+1})\quad \vect{0}_{N_t\times (N_t-1)  } ))_{0}\\
      ((\ \vect{L}_{TB}(N_{t+1})\quad \vect{0}_{N_t\times (N_t-1)  } ))_{1} \\ 
      \ldots \\   
      ((\ \vect{L}_{TB}(N_{t+1})\quad \vect{0}_{N_t\times (N_t-1)  } ))_{2N_t} 
      \end{array} \right] \nokta
\end{equation}
 Notice that 
$\vect{H}_{ETB,MIMO}(N_t)$ is the parity check matrix of a tail biting 
convolutional code of rate $1/(N_t+1)$ and of constraint length $N_t+1$  and
can be derived from $\vect{H}_{TB,MIMO}(N_t)$  by adding $2N_t$ more parity 
checks. 

As opposed to the decoder of $\vect{H}_{TB,MIMO}(N_t)$, which can perform inference only on
$\Pr\{\vect{V}=\vect{v}|\vect{Y}=\vect{y}\}$ , the decoder of $\vect{H}_{ETB,MIMO}(N_t)$ can be used to perform 
inference on  $\Pr\{\vect{V}_A=\vect{v}|\vect{Y}=\vect{y}\}$ 
where $\vect{V}_A$ is \emph{any permutation} of $\vect{X}$. 

An improved 
soft output MIMO detector can be implemented by using the approximate symbolwise detector of 
 $\vect{H}_{ETB,MIMO}(N_t)$ instead of $\vect{H}_{TB,MIMO}(N_t)$. The main advantage of the detector
with extended tail biting decoder when compared original tail biting decoder is that it can 
work with any permutation of $\vect{X}$. Moreover, a certain permutation can work 
better for a given $\vect{H}_c$ and noise realization while another permutation 
can work better with another  $\vect{H}_c$ and noise realization. This 
flexibility comes at the cost of increasing trellis processing complexity by two
which is acceptable. 

We propose a soft output MIMO detection algorithm by using the approximate 
symbolwise decoder of the extended tail biting code as follows. 
\begin{enumerate}
  \item Select a permutation $\vect{P}_{A}$ from a set $\set{P}$ of 
    permutations. 
  \item Apply the inputs properly permuted with the permutation $\vect{P}_A$ 
    to the approximate symbolwise  decoder of $\vect{H}_{ETB,MIMO}(N_t)$.
  \item Run the BCJR algorithm iteratively on the tail biting trellis 
    until it converges or a maximum number of 
    iterations reached.
  \item If the iterative BCJR algorithm converges declare 
    its result as the output of the MIMO detector and 
    halt.
  \item If the iterative BCJR algorithm does not converge
    select another permutation $\vect{P}_A$ from $\set{P}$
    and goto Step 2. If there is not any remaining permutation in 
    $\set{P}$ then declare a failure. 
\end{enumerate}

We tested this algorithm on the $8 \times 8$ MIMO channel described in Section \ref{simulationresults1}.
The set $\set{P}$ we used in this simulations consists of $15$ specific
permutations among $16!$ possible permutations. 
These permutations are given Appendix \ref{mimopermutations}. BER performance of this MIMO detector is given 
in Figure \ref{improvedresults}. These results show that the MIMO detector using 
the extended tail biting decoder improves the error floor performance 
by an order of magnitude. Furthermore, the BER performance before reaching 
the error floor is also improved significantly. The improved MIMO detector
is just $2\textrm{dB}$ away from the optimum algorithm when it reaches
the error floor.

\begin{figure}
  \begin{center}
    \includegraphics[scale=1]{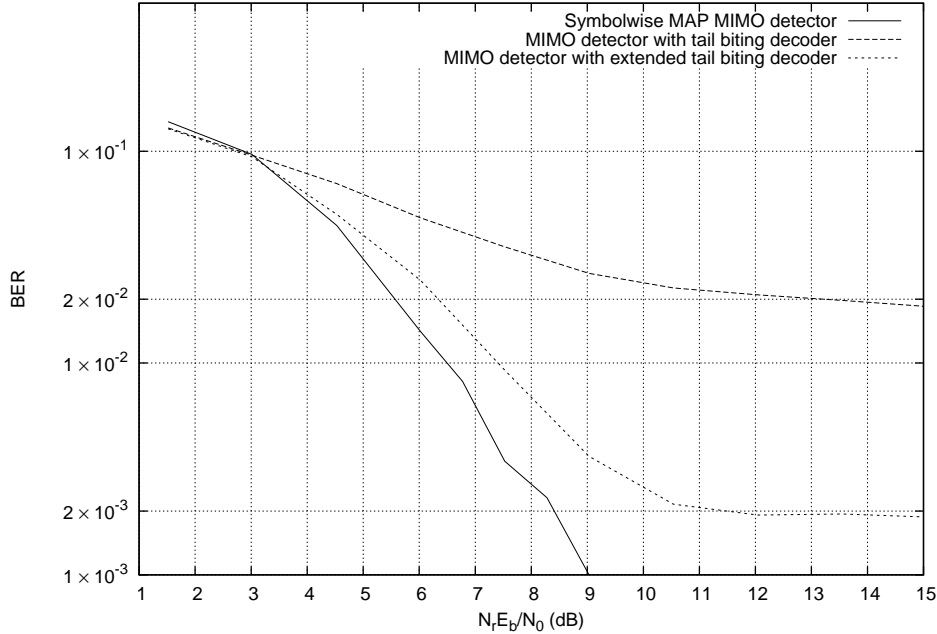}
    \caption{ BER performance of the MIMO detector using the extended tail biting decoder
      with different permutations together with the MIMO detector with the normal tail biting decoder
and symbolwise MAP MIMO detector in Rayleigh fading $8 \times 8$ channel.\label{improvedresults}}.
  \end{center}
\end{figure}

\begin{figure}
  \begin{center}
    \includegraphics[scale=1]{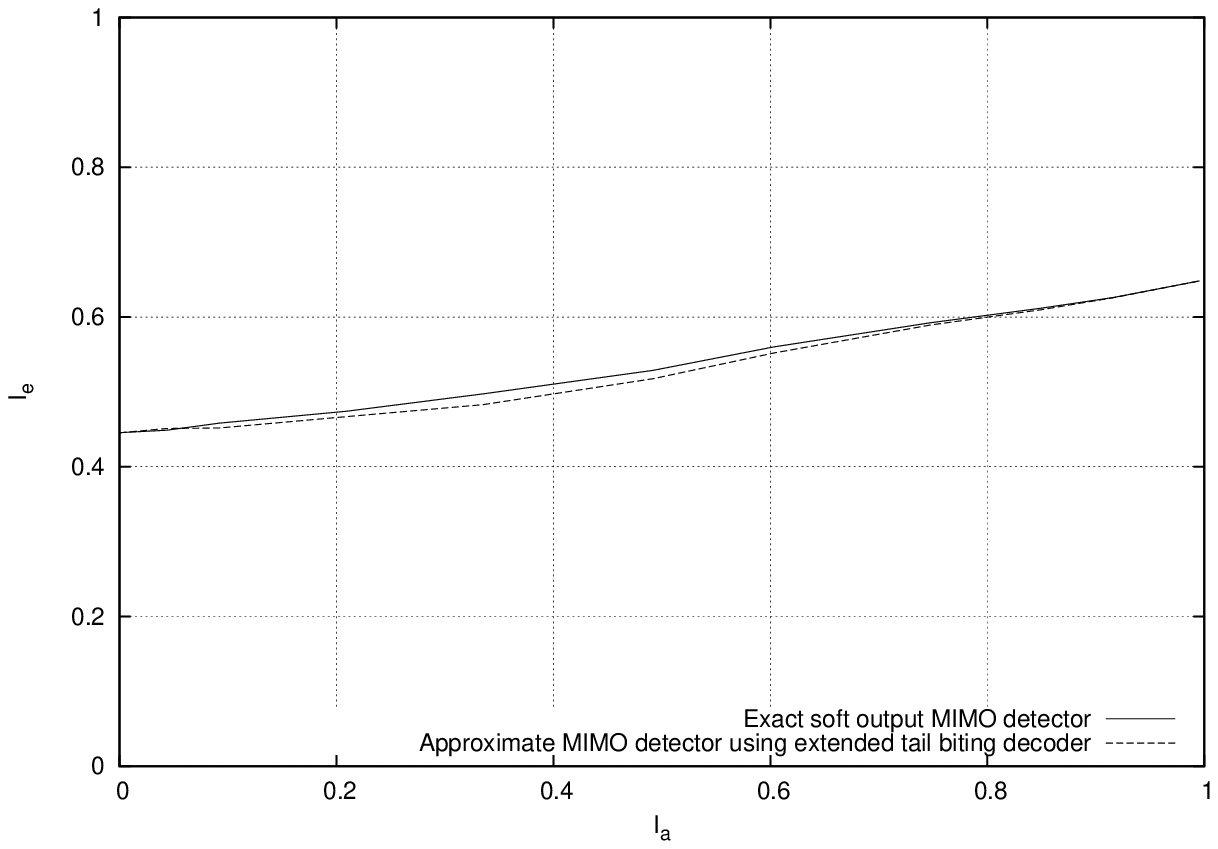}
    \caption{EXIT curves of the approximate MIMO detector using extended tail biting decoder
      and the exact soft output MIMO detector at $N_rE_b/N_{0}=-0.96 \textrm{dB}$\label{exit1}}
  \end{center}
\end{figure}

\begin{figure}
  \begin{center}
    \includegraphics{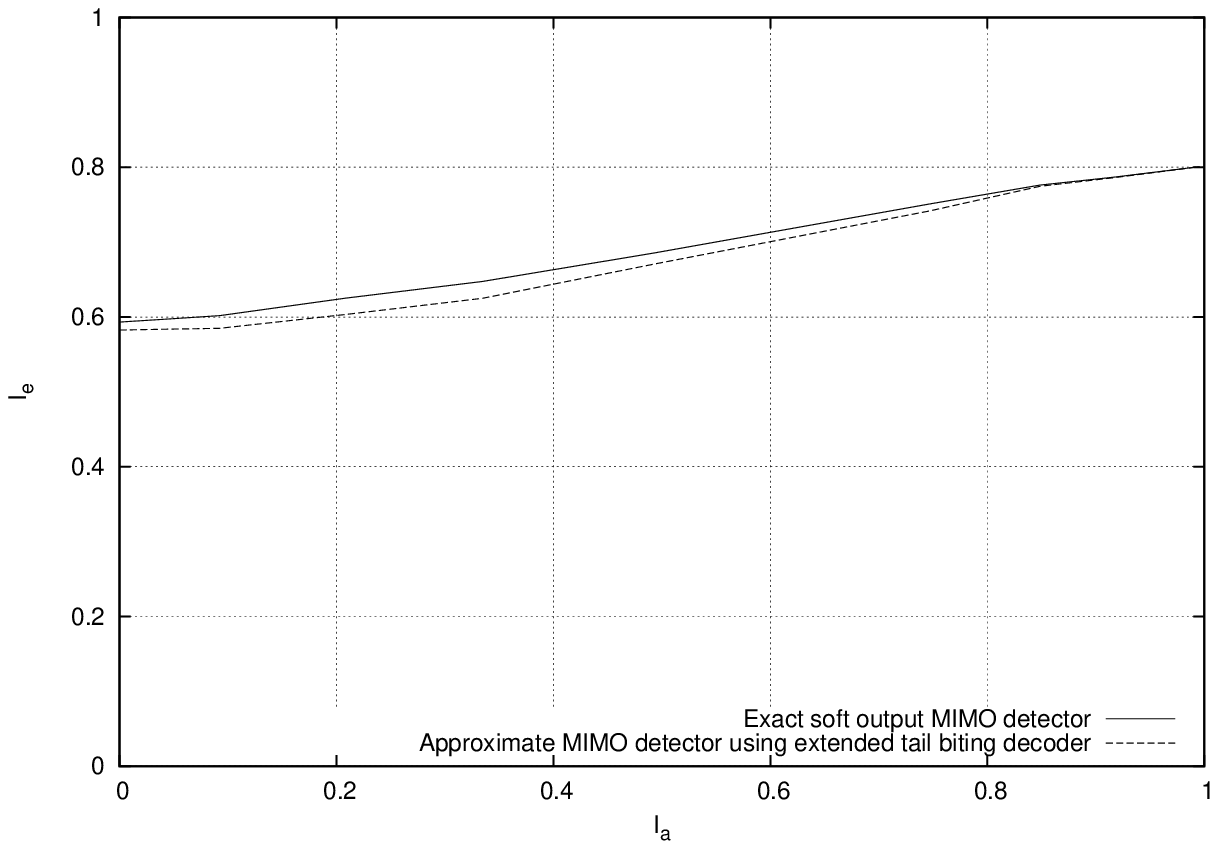}
    \caption{EXIT curves of the approximate MIMO detector using extended tail biting decoder
      and the exact soft output MIMO detector $N_rE_b/N_{0}=1.25\textrm{dB}$\label{exit2}}
  \end{center}
\end{figure}

\begin{figure}
  \begin{center}
    \includegraphics{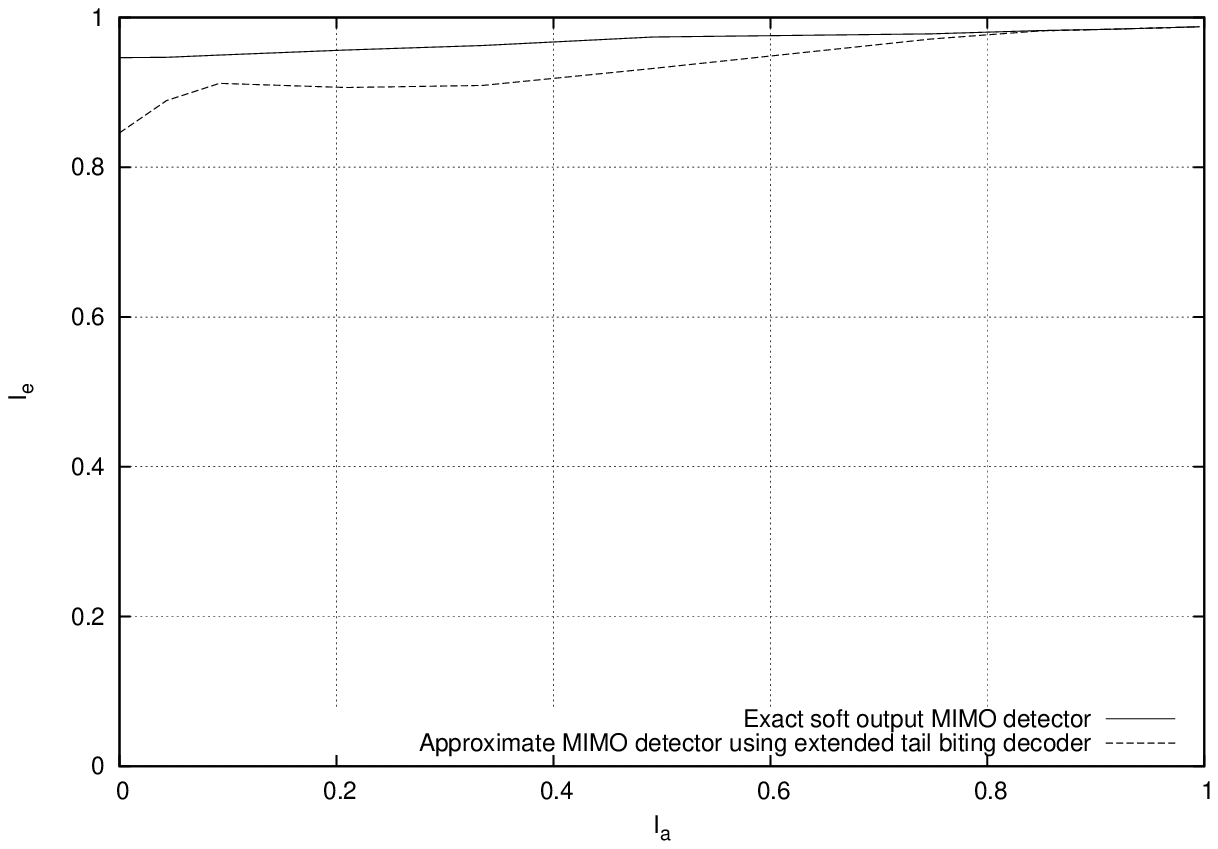}
    \caption{EXIT curves of the approximate MIMO detector using extended tail biting decoder
      and the exact soft output MIMO detector $N_rE_b/N_{0}=6.02\textrm{dB}$\label{exit3}}
  \end{center}
\end{figure}

Recall that this MIMO detector is capable of using a priori information and produces soft output. 
Hence, it can be easily used in a iterative detection-decoding scheme. In order estimate 
the possible performance of the improved MIMO detector in such an iterative scheme, we computed extrinsic information transfer (EXIT) 
curves \cite{tenBrink1,tenBrink2,tenBrinkMIMOLDPC}.
The area under the EXIT curve of a MIMO detector is an approximate estimation of the maximum possible rate of the 
code which can be used in an iterative detection-decoding scheme and can achieve arbitrarily small error rate. 
In this aspect the area under exact soft output MIMO detector is an approximate estimation 
of the MIMO channel capacity \cite{tenBrink2}.  
  
We computed the EXIT curves at three different SNR values.
 These results are shown in Figures \ref{exit1}, \ref{exit2},
and \ref{exit3}.  Since the area between the two EXIT curves in  Figure \ref{exit1} is negligible,
the proposed algorithm can be used in an iterative detection-decoding scheme with the same
code as the optimum algorithm at low SNR or in the power limited region. The EXIT curves
shown in Figure \ref{exit2} lead to similar conclusion. The area between 
the two EXIT curves becomes $0.04$ in Figure \ref{exit3}. This means that the proposed
algorithm can also be used in the bandwidth limited region but at the cost of  a rate loss of $0.04 \textrm{bits}$
which is quite acceptable.

\section{Usage of the decoders of the convolutional codes  as channel equalizers}

Let $X(t)$ be a stochastic process defined as follows. 
\begin{equation}
  X(t)=\sum_{n}\pam{N}{\vect{X}_{n}}f(t-nT) \virgul
\end{equation}
where $f(t)$ is the impulse response of a pulse shaping filter and $\vect{X}_{n}$ is 
a random vector consisting of $N$ bits. Furthermore, let $Y(t)$ be
\begin{eqnarray}
  Y(t)&=&X(t)*g(t)+Z(t)\\
  &=&\sum_{n}\pam{N}{\vect{X}_{n}}*h(t) +Z(t) \virgul
\end{eqnarray} 
where $Z(t)$ is a zero mean white Gaussian noise process with power spectral density $\frac{N_0}{2}$, $*$ denotes 
convolution, $g(t)$ is the impulse response of a causal channel, 
and $h(t)$ is the convolution of the $g(t)$ and $f(t)$. Then it can be shown by following 
similar procedures applied in the previous sections that the Viterbi and BCJR decoders of a certain 
convolutional code $C$ can be used as ML sequence estimator and marginal APP receiver for this inter-symbol
interference system respectively. This code $C$ is the non-recursive systematic convolutional code of
rate $1/NL$ and of constraint length $NL$ where $L$ is the smallest integer such that $h(t)=0$ for $t>LT$. 
The generator polynomials of this code are $1$, $1+x$, $1+x^2$, $\ldots$, $1+x^{NL-1}$.

The inputs that must be applied to these decoders to achieve the desired results consists of samples
taken from the output of the matched filter i.e. $y(t)*h(-t)$ with sampling period $T$, where $y(t)$ is the received signal, 
 samples taken from the time autocorrelation function 
$h(t)*h(-t)$ again with sampling period $T$, and scaling of these samples with $2$'s powers
\footnote{We dropped conjugations since $\pam{N}{\vect{X}_{n}}$ is real}. 

The Viterbi decoder of the mentioned code above actually 
works as an alternative device to compute the Ungerboeck's metric \cite{ungerboeck}.
Therefore, this result would be much more interesting if we achieved it before Ungerboeck.     
However, using a Viterbi decoder  as an alternative device to compute  Ungerboeck's
might still be of practical importance since this approach takes all of the multiplications 
outside of the Viterbi data path. 

We have also empirically verified that the BCJR decoder of the convolutional code
mentioned above with the mentioned inputs returns the exact marginal APPs of the transmitted bits.

\newpage
\chapter[DETERMINING CONDITIONAL  INDEPENDENCE RELATIONS FROM THE CANONICAL FACTORIZATION]
{DETERMINING CONDITIONAL INDEPENDENCE RELATIONS  FROM  THE 
  CANONICAL FACTORIZATION
\label{markovchapter}}
\section{Introduction}
Investigating the conditional independence relations
of random variables is important in many different disciplines \cite{bishop,amsmrfbook}.
These conditional independence relationships are well represented by 
a graphical model called Markov random field (MRF) or undirected
graphical model. In this section we  show that the MRF representing a joint
pmf can be determined from the projections of the joint PMF onto the 
subspaces described in Chapter \ref{thechapter}.

This chapter begins with introducing the relation between conditional independence 
of two random variables and the canonical factorization. Then we explain
how to determine Markov blankets from the canonical factorization. 
This chapter ends with comparing the canonical factorization with 
the Hammersley-Clifford Theorem.  

\section{Conditional Independence of Two Random Variables}

Suppose that it is desired to determine the
conditional independence relations between the components
of the random vector $\vect{X}=[X_1,X_2,\ldots,X_{N}]$ which is
distributed with a $p(\vect{x}) \in \setfieldqn$. 
Then a random variable $X_i$ is  said to be conditionally 
independent of $X_j$ given all the other components of $\vect{X}$
if and only if the following relation
is satisfied:  
\begin{equation}
  \Pr\left\{X_i=x_i| \vectex{X}{i}=\vectex{x}{i}\right\}
  = \Pr\left\{X_i=x_i|\vectex{X}{i,j}=\vectex{x}{i,j}\right\}
\label{conditionalindependence1}
\end{equation}
where $\vect{X}_{\setminus\set{I}}$ ($\vect{x}_{\setminus \set{I}}$) denotes the  vector obtained by removing 
the  components having indices in $\set{I}$ from  $\vect{X}$ ($\vect{x}$). The following theorem 
states the necessary and sufficient conditions for the conditional independence of two 
random variables in terms of the canonical factorization. 
\begin{theorem}\label{conditionalindependencetheorem}
 Let $\vect{X}$ be a random vector distributed with  $p(\vect{x})$ in $\setfieldq$.
$X_k$ and $X_l$ are  conditionally independent given $\vectex{X}{k,l}$ if and only if
$p(\vect{x})$ can be factored as  
\begin{equation}
  p(\vect{x}) = \n{\fieldqn}{\prod_{\vect{a}_i \in \set{K}_{k} \cup \set{K}_{l} }
  r_{i}(\vect{a}_i\vect{x}^{T})} \label{yoruldum_artik}
\end{equation} 
where $\set{K}_{k}$ and $\set{K}_{l}$  are defined as
 \begin{eqnarray}
  \set{K}_{k} &\triangleq& \left\{\vect{a}_i \in \set{H}: \vect{f}_k\vect{a}_i^T=0 \right\} \nonumber\\ 
   \set{K}_{l} &\triangleq& \left\{\vect{a}_i \in \set{H}: \vect{f}_l\vect{a}_i^T=0 \right\} \nonumber  \textrm{.}
\end{eqnarray}
\end{theorem}

The proof is given Appendix \ref{conditionalindependencethmproof}.

The forward statement of this theorem asserts that if none of the SPC factors composing the canonical 
factorization of $p(\vect{x})$ depend 
on both $x_i$ and $x_j$ simultaneously then $X_i$ and $X_j$ are conditionally independent given 
$\vectex{X}{i,j}$. Actually, this result is true not only for the canonical factorization 
but also for any factorization.

The backward statement of Theorem \ref{conditionalindependencetheorem} states that if an SPC factor
of $p(\vect{x})$
\emph{with nonzero norm} 
depends on  $x_i$ and $x_j$ simultaneously then $X_i$ and $X_j$ are definitely conditionally
dependent given $\vectex{X}{i,j}$. On the other hand, in an ordinary factorization  
a  factor function  may depend on $x_i$ and $x_j$ together 
but $X_i$ and $X_j$ can still be conditionally independent given $\vectex{X}{i,j}$.  Therefore, the backward statement of 
Theorem \ref{conditionalindependencetheorem} is specific to the canonical factorization 
and does not hold for all factorizations in general.  
This fact is another reason why we call the proposed factorization the canonical factorization.

\section{Determining Markov Blankets and the Markov Random Field}
 
The Markov blanket of a random variable $X_i$, which is denoted with $\partial X_i$,
is the \emph{smallest} possible set containing the components of $\vectex{X}{i}$
which satisfies
\begin{equation}
  \Pr\{X_i|\vectex{X}{i}\}= \Pr\{X_i|\partial X_i\} \textrm{.}
\end{equation}
Clearly, $\partial X_i$ consists of variables
$X_j$ which are not conditionally independent of $X_i$ given $\vectex{X}{i,j}$.
Based on Theorem \ref{conditionalindependencetheorem}, $\partial X_i$ can be 
obtained in terms of projections onto the SPC constraints as follows.

\begin{corollary}\label{connectioncorollary}
Let the canonical factorization of $p(\vect{x})$ be given by
\begin{equation}
  p(\vect{x})=\n{\fieldqn}{\prod_{\vect{a}_i\in \set{H}} r_{i}(\vect{a}_i\vect{x}^{T})}\nokta 
\end{equation}
$X_k$ is in $\partial X_l$ if and only if  
there exist a parity check coefficient vector $\vect{a}_i \in \set{H}$
such that 
\begin{eqnarray}
  \vect{f}_k\vect{a}_i^{T}&\neq& 0 \nonumber \\
  \vect{f}_l\vect{a}_i^{T}&\neq& 0 \nonumber 
\end{eqnarray}
and
\begin{equation}
   \norm{r_{i}(x)}>0  \nonumber \textrm{.}
\end{equation}
\end{corollary}
\begin{proof}
If such a vector $\vect{a}$ exist then $p(\vect{x})$ cannot be factored 
as in (\ref{yoruldum_artik}) and hence, $X_i$ and $X_j$
are conditionally dependent given $\vectex{X}{i,j}$ due Theorem \ref{conditionalindependencetheorem}.

If there is no such $\vect{a}$ then    $p(\vect{x})$ can be factored as
in (\ref{yoruldum_artik}), which means that
$X_i$ and $X_j$ are conditionally independent given $\vectex{X}{i,j}$.  
\end{proof}

The MRF is an undirected graphical model representing a probability 
distribution where each variable is represented with a node. The node 
representing $X_i$ is connected to the node representing $X_j$ in the MRF
if $X_j$ is in $\partial X_i$. Since the Markov blankets of every 
variable can be determined from the canonical factorization by Corollary
\ref{connectioncorollary}, the MRF can also be determined from the 
canonical factorization. 

Notice that every argument (arguments associated with nonzero parity check coefficient)  
of a non-constant SPC factor are in the Markov blankets of the other arguments of the 
SPC factor. Therefore, the nodes representing these variables in the MRF are all pairwise 
connected. In graph theoretic terminology, these nodes form a clique in the MRF. 
Hence, SPC factors are functions of the cliques (not necessarily maximal) of the MRF.

\section{Comparison to the Hammersley-Clifford Theorem}

The relation between the factorization of 
a multivariate PMF and Markov properties is 
first established by Hammersley and Clifford
in \cite{hct1,hct2}. In this work they show 
that any \emph{strictly positive} multivariate PMF can be expressed as
\begin{equation}
  p(\vect{x}) = \frac{1}{C}\prod_{\vect{D}\in \set{D}_{C}} \phi_{\vect{D}}(\vect{D}\vect{x})\textrm{,}
\label{hcteq1}
\end{equation}
where each element of $\set{D}_C$ is associated with a clique in the MRF. 
Moreover, their proof is constructive. The factor functions
are given as  
\begin{equation}
  \phi_{\vect{D}}(\vect{D}\vect{x}) \triangleq \prod_{\vect{D}': \vect{D}' \vect{D}=\vect{D}'} 
  p(\vect{D}'\vect{x} +(\vect{I}-\vect{D}')\vect{x}_B)^{\left((-1)^{\abs{\vect{D}-\vect{D}'}}\right)}
\label{hcteq2}
\end{equation}
where $\vect{x}_{B}$ is  a fixed configuration \footnote{This configuration corresponds
to the all-black coloring in \cite{hct1,hct2}.}. 
Although both in our and their approaches the 
factor functions appear to be the functions of the 
cliques of the MRF, our approach differs significantly 
from theirs in many aspects.


First of all, the dependencies between the random variables imposed by factor
functions in (\ref{hcteq2}) are rather arbitrary. SPC factors, on the other hand,   
impose an algebraic form of dependency. In other words, SPC factors explain
how a random variable is related to a linear combination of other variables. 
This property is quite important and allows us to express
an inference problem as a decoding problem. 

Second, the factor functions defined in  (\ref{hcteq2}) depend on  a  certain 
fixed configuration $\vect{x}_B$. A different factorization
is obtained for each different $\vect{x}_B$.  Therefore, the factorization 
proposed by Hammersley and Clifford  is not unique. On the other hand, 
the canonical factorization is unique as explained in Section \ref{uniquenesssection}.

In addition, there is at most one factor function per clique in the factorization 
given in (\ref{hcteq1}) whereas there may be  more than one SPC factors
depending on the same set of variables in non-binary fields. 

Finally, the applicability of our approach is more restricted than 
that of the Hammersley and Clifford's. Our method is applicable only if 
the event space of the combined experiment  can be mapped to $\fieldq^{N}$
 whereas the Hammersley-Clifford theorem 
is applicable to any strictly positive pmf. 
Moreover, it should be emphasized that both approaches
are applicable to strictly positive pmfs only.

\chapter{Conclusions and Future Directions \label{conclusionchapter}}

\section{Summary}

In this thesis the Hilbert space of pmfs is introduced. Then the tools provided
by this Hilbert space, is utilized to develop an analysis method for multivariate
pmfs. The aim of this analysis method is to obtain a factorization of the multivariate pmf. 
The resulting factorization from this analysis method possess some important properties. 
First of all it is the ultimate factorization possible. Secondly, it is unique. Thirdly,
the conditional independence relations can be determined completely from this factorization. 
Probably the most important property of the resulting factorization is the fact that 
it reveals the algebraic dependencies between the involved random variables. Thanks to this 
fact  probabilistic inference problems can be transformed into channel decoding problems
and channel decoders can be used for other tasks beyond decoding. 
Many examples are provided in thesis on how channel decoders can be used as detectors of 
communication receivers. It is also shown that the decoders of tail biting 
convolutional codes can be used as a MIMO detector. This approach  results
in a significant reduction in complexity while maintaining good performance.

\section{Future directions}

The application of the Hilbert space of pmfs is presented in this thesis is the canonical factorization. 
We believe that the Hilbert space of pmfs might lead to further applications in  communication
theory, information theory, and probabilistic inference.

The most important consequence of the canonical factorization is that it shows
how to employ channel decoders for other purposes. The MIMO detector 
which uses the decoder of a tail biting convolutional code demonstrates
that new detection and probabilistic inference algorithms can be developed
by using channel decoders for tasks beyond decoding. 

Employing channel decoders for other tasks also allows to apply the analog 
probability propagation method proposed in \cite{loeligeranalog,loeligerdigeranalog}
for other probabilistic inference problems. In particular, by implementing channel equalizers and 
MIMO detectors with analog probability propagation much more power efficient communication 
receivers can be implemented.  We anticipate that this direction will
be the most important application area of this thesis. 

Some other possible future directions are summarized below. 

\subsection{Applications on machine learning}

Estimating the factorization of a joint pmf from samples generated from the pmf  is an important problem
in machine learning, e.g. \cite{pabbeel}. A straightforward approach after this thesis could 
be estimating the joint pmf first and obtain the canonical factorization by applying
the procedure explained in Chapter \ref{thechapter}. However, such an approach 
both require too many samples to estimate the joint pmf accurately and extensive computational 
resources to obtain the canonical factorization. 
A more interesting solution to this problem might be proposed by combining the
results obtained in this thesis and the results presented in \cite{massey}.
By combining these results it can be concluded that the necessary algorithm 
for estimating the factorization of a joint pmf from samples
is exactly the \emph{inverse of the sum-product algorithm}. 

As it is explained in Section \ref{ultimatenesssection} the ultimate  factorization of
a pmf is the canonical factorization. The equivalent Tanner graph representing
the canonical factorization is shown in Figure \ref{genericgraphs}-b. 
Hence, estimating the canonical factorization is equivalent to estimating 
all of the local evidences in this Tanner graph. 

Let $\vect{X}=[X_1,X_2,\ldots,X_N]$ be distributed with a $p(\vect{x})$ in $\setfieldqn$.
Estimating all the marginals $\Pr\{X_i=x_i\}$ from experimental data is much easier
than estimating the joint distribution $p(\vect{x})$ from data. Let
\begin{equation}
  X_i \triangleq \vect{a}_i \vect{X}^{T}, \quad \textrm{for} i=N+1,N+2,\ldots, \setsize{H} \virgul \nonumber
\end{equation} 
where $\vect{a}_{N+1}$, $\vect{a}_{N+2}$, $\ldots$, $\vect{a}_{\setsize{H}}$ are the 
elements of $\set{H}$ of weight two or more as we assumed in Chapter \ref{decodingchapter}.
Since $X_i$ for $i>N$ is completely determined by $\vect{X}$, the marginal distributions
of $X_i$ for $i>N$ can also be estimated from the data.
Consequently, the marginal distributions of $X_1$, $X_2$, $\ldots$, $X_{\setsize{H}}$ can be easily 
estimated from the experimental data.

However, what we need to estimate  the canonical factorization are not the marginal distributions of 
the random variables  $X_1$, $X_2$, $\ldots$, $X_{\setsize{H}}$ but the local evidences 
in Figure \ref{genericgraphs}-b. Therefore, we need an algorithm which computes 
the local evidences from the marginals. Notice that, this task is exactly \emph{the inverse of 
the sum-product algorithm} as the sum-product algorithm computes the marginals from local evidences.

A question might arise on the existence and uniqueness of the set of the local evidences
corresponding to a  set of marginals. Indeed, if the Tanner graph  in Figure \ref{genericgraphs}-b 
represented an arbitrary code then we might not find a set of local evidences resulting in 
a given set of marginal distributions at all or might find more than one 
set of local evidences resulting in the same set of marginal distributions. 
Any linear combination of the vector $\vect{X}$ is equal to 
$\alpha X_i$ for an $\alpha \in \fieldq$ and $1\leq i \leq \setsize{H}$. 
Massey showed in \cite{massey} that the marginal distributions 
of the linear combinations of a sequence of random variables 
is enough to specify their joint distribution. Hence, the marginal distributions 
of  $X_1$, $X_2$, $\ldots$, $X_{\setsize{H}}$  uniquely specifies $p(\vect{x})$
and consequently its canonical factorization.

To the best of our knowledge, neither exact nor approximate versions of the 
inverse of the sum-product algorithm is known. As explained above, developing 
the inverse of the sum-product algorithm solves an important problem 
in machine learning. 

\subsection{Using channel decoders for channel estimation}

In the examples presented in Chapter \ref{applicationchapter}, we assumed
that  the channel coefficients are completely known at the receiver. In a
practical communication receiver, the channel coefficients must be estimated.
Employing channel decoders for channel estimation would be very interesting.

Actually, the channel estimation problem does not perfectly fit into the framework 
presented in this thesis since the channel coefficients take samples from a 
continuous alphabet rather than a finite alphabet. The apparent solution 
to this problem might be quantizing the channel coefficients. However, 
such an approach would lead to a factor graph topologically equivalent to the one
 in \cite{worthen} which contains too many short cycles. Hence, such 
an approach probably will not be useful.

While employing decoders for detection, we observed that  the channel coefficients and
the channel outputs appeared as the parameters of the canonical factorization 
of the transmitted bits. Therefore, a more interesting approach 
might be bypassing the channel estimation step and estimating the canonical factorization 
of the joint pmf of the transmitted bits and the quantized channel outputs directly from 
a pilot sequence. This approach transforms the channel estimation 
problem into a machine learning problem a solution to which is conjectured 
in the previous section. 





\appendix

\chapter{PROOFS AND DERIVATIONS}
\section{Proofs and derivations in Chapter 2}
\subsection{Proof of Lemma \ref{mappinglemma}\label{proofmappinglemma}}
The function $\sigma(p(x),r(x))$ defined in (\ref{innerproductdef}) is
an  inner product on $\setfieldq$ if it satisfies three inner product axioms stated below. 
\begin{itemize}
  \item \emph{Symmetry:} This property of $\sigma(.,.)$ is directly inherited from the inner product on $\mathbb{R}^{q}$.
  \item \emph{Linearity w.r.t. first argument:} If $\operator{M}{.}$ is linear this property is also inherited from 
    the inner product on $\mathbb{R}^{q}$.
  \item \emph{Positive definiteness:} For any $p(x) \in \setfieldq$
    \begin{eqnarray}
      \sigma(p(x),p(x)) &=& \innerproduct{\operator{M}{p(x)}}{\operator{M}{p(x)}} \nonumber \\
      &\geq& 0 \nonumber
    \end{eqnarray}
    due to the non-negativity of the inner product on $\mathbb{R}^{q}$. The equality is 
    satisfied only if $\operator{M}{p(x)}$ equals to $\vect{0}$. Since $\operator{M}{.}$ is linear
    and an injection  $\operator{M}{p(x)}$ is equal to $\vect{0}$ if and only if $p(x)=\theta(x)$. 
\end{itemize}

\subsection{Rationale behind the proposal for $\operator{L}{.}$\label{anglerationale}}
The trivial way of mapping a pmf $p(x)\in \setfieldq$ by a vector $\vect{p}\in \mathbb{R}^{q}$ 
is making the $i^{th}$ \footnote{We enumerate the components of the vector with the elements of $\fieldq$ instead
of positive integers.} component of $\vect{p}$ equal to $p(i)$. Let this trivial mapping be denoted by $\operator{T}{.}$, i.e.,
\begin{equation}
  \operator{T}{p(x)} \triangleq \sum_{i\in\fieldq}p(i) \vect{e}_i \textrm{.} \nonumber
\end{equation} 
Although this mapping is injective, it is obviously nonlinear. Therefore, $\operator{T}{.}$ does not
satisfy one of the two requirements imposed by Lemma \ref{mappinglemma} and consequently it cannot be employed
as a tool for borrowing the inner product on $\mathbb{R}^{q}$. However, we can define a notion 
of angle between pmfs using $\operator{T}{.}$ and then reach a proposal for a mapping which satisfies the requirements
of Lemma \ref{mappinglemma}.      

Whatever the definition of the angle between two pmfs is, the sine of the angle should be kept constant if two
pmfs are scaled by some nonzero scalars. In other words, for any $p(x),r(x)\in \setfieldq$ and $\alpha,\beta \in \mathbb{R}\setminus\{0\}$
\begin{equation}
  \sin \angle (p(x),r(x))= \sin \angle (\alpha\boxtimes p(x),\beta\boxtimes r(x)) \nonumber \textrm{,} 
\end{equation}
where $\angle (p(x),r(x))$ denotes the angle between $p(x)$ and $r(x)$. This property of
angle imposes that the angle between two pmfs should be a function of the two 
 parametric curves on $\mathbb{R}^q$ based on $p(x)$ and $r(x)$ as follows.
\begin{eqnarray}
\vect{c}_{p}(t)&\triangleq& \operator{T}{t\boxtimes p(x) }\textrm{,} \nonumber \\
\vect{c}_{r}(t)&\triangleq& \operator{T}{t\boxtimes r(x) }\textrm{.} \nonumber 
\end{eqnarray} 
For $t=0$ both of these curves pass through $\frac{1}{q}\vect{1}$. An example consisting of a pair of such curves for $\setfield{3}$
is depicted in Figure \ref{angleplot}. Then we can reasonably define the angle between $p(x)$ and $r(x)$ as the angle 
between  $\vect{c}_{p}(t)$ and $\vect{c}_{r}(t)$ at their intersection point. 
\begin{figure}
\begin{center}
\includegraphics[scale=.9]{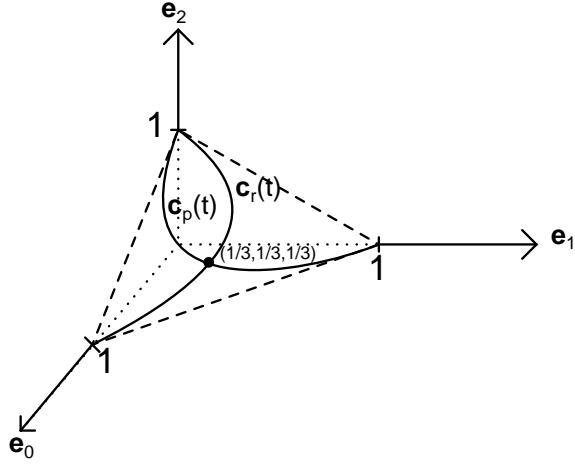}
  \caption{A pair of parametric curves obtained by scaling two pmfs in $\setfield{3}$ and then mapping them to 
$\mathbb{R}^{3}$ via the trivial mapping. \label{angleplot}}
\end{center} 
\end{figure}

In order to derive the angle between $\vect{c}_{p}(t)$ and $\vect{c}_{r}(t)$,  we need to derive vectors tangent to 
these curves at $t=0$. The expression defining   $\vect{c}_{p}(t)$ can be simplified as
\begin{eqnarray}
  \vect{c}_{p}(t)&=&\sum_{i\in\fieldq}\frac{(p(i))^t}{\sum_{j\in\fieldq} (p(j))^t} \vect{e}_i \nonumber \\
  &=&\sum_{i\in\fieldq} \left(\sum_{j\in\fieldq}\exp\left(t \Big( \log p(j)-\log p(i)\Big) \right)\right)^{-1}\vect{e}_i \nonumber \textrm{.}
\end{eqnarray}
Let $\vect{t}_{p}$ denote the vector which is tangent to $\vect{c}_{p}(t)$ at $t=0$. Then $\vect{t}_{p}$  can be
derived using derivation as
\begin{equation}
  \vect{t}_{p}=\sum_{i\in\fieldq }\left(q\log p(i) -\sum_{j\in \fieldq} \log p(j)\right)\vect{e}_i \nonumber \textrm{.}
\end{equation}
Having inspired from this equation, We proposed the mapping $\operator{L}{.}$   as
\begin{eqnarray}
  \operator{L}{.}&=&\frac{1}{q}\vect{t}_p \nonumber \\
  &=&\sum_{i\in\fieldq }\left(\log p(i) -\frac{1}{q}\sum_{j\in \fieldq} \log p(j)\right)\vect{e}_i \nonumber \textrm{.}
\end{eqnarray}  
Since $\operator{L}{.}$ is defined as above, the angle between the two curves $\vect{c}_{p}(t)$ and $\vect{c}_{q}(t)$, which
is proposed to be the of the angle between $p(x)$ and $q(x)$, is equal to the angle between $p(x)$ and $q(x)$ on $\setfieldq$
defined on (\ref{angledefinition}).

\subsection{Proof of Lemma \ref{operatorlemma}\label{operatorlemmaproof}}
First we are going to prove that $\operator{L}{.}$ is linear and then it is an injection. 
For any $p(x),r(x) \in \setfieldq$,
\begin{eqnarray}
  \operator{L}{p(x) \boxplus r(x)} &=&
  \sum_{i\in \fieldq}\left(\log \n{\fieldq}{p(x)r(x)}\evalat{x=i}-
    \frac{1}{q}\sum_{j\in \fieldq}\log \n{\fieldq}{p(x)r(x)}\evalat{x=j} \right) \vect{e}_{i} \nonumber \\
  &=&   \sum_{i\in \fieldq}\left(\log \frac{1}{\gamma}p(i)r(i)-
    \frac{1}{q}\sum_{j\in \fieldq}\log\frac{1}{\gamma} p(j)r(j)\right) \vect{e}_{i} \nonumber \\
&=& \sum_{i\in \fieldq}\left(\log p(i)-
    \frac{1}{q}\sum_{j\in \fieldq}\log p(j)\right) \vect{e}_{i}+\sum_{i\in \fieldq}\left(\log p(i)-
    \frac{1}{q}\sum_{j\in \fieldq}\log p(j)\right) \vect{e}_{i} \nonumber \\
&=&\operator{L}{p(x)}+\operator{L}{r(x)} \textrm{,}\nonumber 
\end{eqnarray}
where $\gamma$ in the second line above is $\sum_{i\in \fieldq}p(i)r(i)$. Hence, $\operator{L}{.}$ is additive. 
For any $p(x) \in \setfieldq$ and $\alpha \in \mathbb{R}$,
\begin{eqnarray}
  \operator{L}{\alpha \boxtimes p(x)} &=& \sum_{i \in \fieldq} \left(\log \n{\fieldq}{(p(x))^{\alpha}}\evalat{x=i}-
  \frac{1}{q}\sum_{j \in \fieldq}\log \n{\fieldq}{(p(x)^{\alpha})}\evalat{x=j} \right)\vect{e}_{i} \nonumber\\
&=&\sum_{i \in \fieldq} \alpha\left(\log p(i)-
  \frac{1}{q}\sum_{j \in \fieldq}\log p(j) \right)\vect{e}_{i} \nonumber\\
&=&\alpha \operator{L}{p(x)} \nonumber\textrm{.} 
\end{eqnarray}
Hence, $\operator{L}{.}$ is homogeneous and consequently a linear mapping. 

A linear mapping is injective if its kernel  (null space) is composed of only 
the additive identity. 
If  $\operator{L}{p(x)}=\vect{0}$ for a $p(x) \in \fieldq$ then
\begin{eqnarray}
  \log p(i) - \frac{1}{q}\sum_{j \in \fieldq} \log p(j)= 0 &\quad &\forall i \in \fieldq \nonumber \\ 
  p(i) = \exp\left(\frac{1}{q} \sum_{j \in \fieldq} \log p(j)\right) &\quad &\forall i \in \fieldq \nonumber \textrm{,}
\end{eqnarray}
which is possible only if $p(x)=\frac{1}{q}$ or equivalently $p(x)=\theta(x)$. Since the kernel of $\operator{L}{.}$
consists of only $\theta(x)$, which is the additive identity in $\setfieldq$, the mapping $\operator{L}{.}$ is injective.  
\subsection{Expressing the inner product on $\setfieldq$ as a covariance\label{covariancederivation}}
Let $X$ be a $\fieldq$-valued random variable. Then $\log p(X)$ and $\log r(X)$
are two real-valued functions of an $\fieldq$-valued random variable. 
Their expectations and covariance are well-defined. Clearly, the inner product of $p(x)$ and $r(x)$ can be 
expressed as 
\begin{eqnarray}
  \innerproduct{p(x)}{r(x)} &=&q\expectation{(\log p(X)-\expectation{\log p(X)})(\log r(X)-\expectation{\log r(X)})}\nonumber \\
  &=&q\left(\expectation{\log p(X) \log r(X)}-\expectation{\log p(X)}\expectation{ \log r(X)} \right) \nonumber \textrm{,}
\end{eqnarray}
where $\expectation{.}$ denotes expectation and $X$ is a uniformly distributed random variable in $\fieldq$, i.e.
\begin{equation}
  \Pr\{X=x\}=\theta(x) \nonumber \textrm{.}
\end{equation}

\subsection{Proof of Lemma \ref{operatorortholemma}\label{operatorortholemmaproof}}
For any $p(x)$ in $\setfieldq$
\begin{eqnarray}
  \innerproduct{\operator{L}{p(x)}}{\vect{1}}_{\mathbb{R}^{q}}&=&\innerproduct{\sum_{i \in \fieldq}\left(\log p(i)-\frac{1}{q}
      \sum_{j\in\fieldq}\log p(j)\right)\vect{e}_i}{\vect{1}}_{\mathbb{R}^{q}} \nonumber \\
  &=&\sum_{i\in\fieldq}\left(log p(i)-\frac{1}{q}\sum_{j\in \fieldq} \log p(j)\right)\nonumber \\
    &=&\sum_{i\in \fieldq} \log p(i)-\sum_{j\in\fieldq} \log p(j)\nonumber \\
    &=&0 \textrm{,}\nonumber
\end{eqnarray}
which completes the proof. 
\subsection{Proof of Lemma \ref{pseudoinvlemma}\label{pseudoinvlemmaproof}}
First we are going to simplify the expression defining $\pseudoinv{L}{\vect{p}}$.
\begin{eqnarray}
  \pseudoinv{L}{\vect{p}}&=& \n{\fieldq}{\exp\left(-\frac{1}{2}\norm{\vect{p}-\vect{s}(x)}^2\right)}\nonumber\\
  &=&\n{\fieldq}{\exp\left(-\frac{\norm{\vect{p}}^2-2\innerproduct{\vect{p}}{\vect{s}(x)}_{\mathbb{R}^{q}}+\norm{\vect{s}(x)}^2}{2}\right)}\nonumber\\
  &=&\n{\fieldq}{\exp\left(-\frac{\norm{\vect{p}}^2}{2}\right) 
    \exp\left(-\frac{\norm{\vect{s}(x)}^2}{2}  \right)  \exp\left( \innerproduct{\vect{p}}{\vect{s}(x)}_{\mathbb{R}^{q}}  \right)}   \nonumber
\end{eqnarray}
Since $\norm{\vect{p}}$ and $\norm{\vect{s}(x)}$ is constant for all $x$,
the product $\exp\left(-\frac{\norm{\vect{p}}^2}{2}\right) 
    \exp\left(-\frac{\norm{\vect{s}(x)}^2}{2}  \right)$ has no effect due to the normalization operator. Therefore, 
\begin{equation}
  \pseudoinv{L}{\vect{p}}=\n{\fieldq}{\exp\left( \innerproduct{\vect{p}}{\vect{s}(x)}_{\mathbb{R}^{q}}  \right)} \textrm{.}\label{result1} 
\end{equation}
If $\vect{p}$ is equal to $\operator{L}{p(x)}$ for a $p(x)$ in $\setfieldq$ then the inner product above becomes
\begin{eqnarray}
  \innerproduct{\vect{p}}{\vect{s}(x)}_{\mathbb{R}^{q}}&=& \innerproduct{\operator{L}{p(x)}}{\vect{e}_{x}-\frac{1}{q}\vect{1}} \nonumber \\
  &=&\innerproduct{\operator{L}{p(x)}}{\vect{e}_x}-\frac{1}{q}\innerproduct{\operator{L}{p(x)}}{\vect{1}} \nonumber \textrm{.}
\end{eqnarray}
Due to Lemma \ref{operatorortholemma} the second inner product above is zero.  Inserting this result into (\ref{result1})
yields
\begin{eqnarray}
\pseudoinv{L}{\operator{L}{p(x)}} &=& \n{\fieldq}{\exp\left(\log p(x)-\frac{1}{q}\sum_{j \in \fieldq} \log p(j) \right)}\nonumber \\
&=&\n{\fieldq}{\exp\left(\log p(x)\right)} \nonumber \\
&=&p(x)\nonumber \textrm{,}
\end{eqnarray}  
where the summation in the first line above is cancelled by the normalization operator 
since it is a constant. This completes the proof of the first part of the lemma.

Any $\vect{p}\in \mathbb{R}^{q}$ can be decomposed as
\begin{equation}
  \vect{p}=\vect{p}'+\alpha \vect{1} \nonumber \textrm{,}
\end{equation} 
for an $\alpha$ in $\mathbb{R}$ such that $\vect{p}' \perp \vect{1}$. Inserting this decomposition into (\ref{result1}) yields
\begin{eqnarray}
  \pseudoinv{L}{\vect{p}}&=&\n{\fieldq}{\exp\left( \innerproduct{\vect{p}'+\alpha \vect{1}}{\vect{s}(x)}_{\mathbb{R}^{q}}  \right)} \nonumber \\
  &=&\n{\fieldq}{\exp\left( \innerproduct{\vect{p}'}{\vect{s}(x)}_{\mathbb{R}^{q}}+
\alpha \innerproduct{ \vect{1}}{\vect{s}(x)}_{\mathbb{R}^{q}}  \right)} \nonumber
\end{eqnarray}
$\vect{s}(x)$ is orthogonal to $\vect{1}$ for all $x$. Therefore,
\begin{eqnarray}
\pseudoinv{L}{\vect{p}}&=&\n{\fieldq}{\exp\left( \innerproduct{\vect{p}'}{\vect{s}(x)}_{\mathbb{R}^{q}}  \right)} \nonumber\\
\operator{L}{\pseudoinv{L}{\vect{p}}}&=&\sum_{i \in \fieldq}
\left(\log \frac{1}{\gamma}{\exp\left( \innerproduct{\vect{p}'}{\vect{s}(i)}_{\mathbb{R}^{q}}  \right)}-
\frac{1}{q}\sum_{j\in \fieldq} \log\frac{1}{\gamma}{\exp\left( \innerproduct{\vect{p}'}{\vect{s}(j)}_{\mathbb{R}^{q}}  \right)}\right)\vect{e}_i \nonumber\\
&=&\sum_{i \in \fieldq}\left(\innerproduct{\vect{p}'}{\vect{s}(i)}_{\mathbb{R}^{q}}-\frac{1}{q}\innerproduct{\vect{p}'}{\sum_{j\in \fieldq}\vect{s}{(j)}}_{\mathbb{R}_{q}}\right)\vect{e}_{i} \textrm{,} \nonumber
\end{eqnarray}
where $\gamma=\sum_{i\in \fieldq }\exp\left( \innerproduct{\vect{p}'}{\vect{s}(i)}_{\mathbb{R}^{q}}  \right)$. 
$\sum_{j\in\fieldq} \vect{s}(j)$ is equal to the zero vector. Therefore, 
\begin{eqnarray}
\operator{L}{\pseudoinv{L}{\vect{p}}}&=&\sum_{i \in \fieldq}\left(\innerproduct{\vect{p}'}{\vect{s}(i)}_{\mathbb{R}^{q}}\right)\vect{e}_i \nonumber \\
&=&\sum_{i \in \fieldq}\left(\innerproduct{\vect{p}'}{\vect{e}_i+\frac{1}{q}\vect{1}}_{\mathbb{R}^{q}}\right)\vect{e}_i \nonumber \\
&=&\sum_{i \in \fieldq}\left(\innerproduct{\vect{p}'}{\vect{e}_i}_{\mathbb{R}^{q}}+\frac{1}{q}\innerproduct{\vect{p}'}{\vect{1}}_{\mathbb{R}^q}\right)\vect{e}_i \nonumber \\
&=&\vect{p}' \nonumber 
\end{eqnarray}
If $\vect{p}$ is orthogonal to $\vect{1}$ then $\vect{p}$ becomes equal to $\vect{p}'$ and consequently
\begin{equation}
  \operator{L}{\pseudoinv{L}{\vect{p}}}=\vect{p} \textrm{.} \nonumber
\end{equation}

\section{Proofs and derivations in Chapter 3}

\subsection{Proof of Lemma \ref{spcinnerproduct}\label{spcinnerproductproof}} 
Inserting the expressions for  $p_1(\vect{x})$ and  $p_2(\vect{x})$
into inner product definition  yields 
\begin{eqnarray}
  \innerproduct{p_1(\vect{x})}{p_2(\vect{x})} &=&
\sum_{\vect{i}\in \fieldqn}\log \frac{r_1(\vect{a}\vect{i}^{T})}{q^{N-1}}
\log \frac{r_2(\vect{b}\vect{i}^{T})}{q^{N-1}}
-\frac{1}{q^{N}}\sum_{\vect{i}\in\fieldqn}\log \frac{r_1(\vect{a}\vect{i}^{T})}{q^{N-1}}
\sum_{\vect{i}\in\fieldqn}\log \frac{r_2(\vect{b}\vect{i}^{T})}{q^{N-1}} \nonumber \\
&=&\sum_{\vect{i}\in \fieldqn}\log r_1(\vect{a}\vect{i}^{T})
\log \frac{r_2(\vect{b}\vect{i}^{T})}{q^{N-1}}
-\frac{1}{q^{N}}\sum_{\vect{i}\in\fieldqn}\log r_1(\vect{a}\vect{i}^{T})
\sum_{\vect{i}\in\fieldqn}\log \frac{r_2(\vect{b}\vect{i}^{T})}{q^{N-1}} \nonumber \\
&=&\sum_{\vect{i}\in \fieldqn}\log r_1(\vect{a}\vect{i}^{T})
\log r_2(\vect{b}\vect{i}^{T})
-\frac{1}{q^{N}}\sum_{\vect{i}\in\fieldqn}\log r_1(\vect{a}\vect{i}^{T})
\sum_{\vect{i}\in\fieldqn}\log r_2(\vect{b}\vect{i}^{T}) \label{p1p2def}
 \end{eqnarray} 

First we are going to derive the inner product of the two SPC constraints 
if there exist an $\alpha \in \fieldq$ such that $\vect{b}=\alpha\vect{a}$.
Since $\fieldq$ is a field with $q$ elements
 and $\vect{a}$ is nonzero there are $q^{N-1}$ $\vect{i}$ vectors
satisfying the equation $\vect{a}\vect{i}^{T}=j$ for all $j\in \fieldq$. Hence, 
\begin{eqnarray}
  \innerproduct{p_1(\vect{x})}{p_2(\vect{x})} &=&q^{N-1} \sum_{j \in \fieldq} \log r_1(j) \log r_2(\alpha j)
  -q^{N-2}\left(\sum_{j\in \fieldq} \log r_1(j) \right)\left(\sum_{i\in \fieldq} \log r_2(\alpha j) \right) \nonumber \\
  &=&q^{N-1}\innerproduct{r_1(x)}{r_2(\alpha x)} \nonumber  \textrm{,}
\end{eqnarray}
which completes the proof for the first part. 

In the second part, we derive the inner product of $p_1(\vect{x})$ and $p_2(\vect{x})$ when
 there is not any $\alpha\in \fieldq$ such that $\vect{b}=\alpha \vect{a}$.
In other words, $\vect{a}$ and $\vect{b}$ are linearly independent.
The first summation in (\ref{p1p2def}) can be regrouped for this case as follows.
\begin{eqnarray}
  \sum_{\vect{i}\in \fieldqn} \log r_1(\vect{a}\vect{i}^{T}) \log r_{2}( \vect{b}\vect{i}^{T})&=& 
\sum_{j\in\fieldq}\left(\sum_{\vect{i}: \vect{a}\vect{i}^{T}=j}\log r_1(j) \log r_{2}(\vect{b}\vect{i}^{T})\right) \nonumber\\
&=&\sum_{j\in\fieldq} \log r_1(j) \sum_{\vect{i}: \vect{a}\vect{i}^{T}=j} \log r_{2}(\vect{b}\vect{i}^{T}) \nonumber \\
&=&\sum_{j\in\fieldq} \log r_1(j) \sum_{j\in \fieldq} \left(\sum_{\vect{i}: (\vect{a}\vect{i}^{T}=j \land \vect{b}\vect{i}^{T}=k)}\log r_2(k) \right) \nonumber\\
&=&\sum_{j\in\fieldq} \log r_1(j) \sum_{j\in \fieldq} \log r_2(k) \left(\sum_{\vect{i}: (\vect{a}\vect{i}^{T}=j \land \vect{b}\vect{i}^{T}=k)}1 \right) 
\nonumber \textrm{.} 
\end{eqnarray}
Since $\fieldq$ is a field with $q$ elements and $\vect{a}$, $\vect{b}$ are linearly independent the innermost 
summation above runs $q^{N-2}$ times for all $j$ and $k$. Therefore, 
\begin{eqnarray}
  \sum_{\vect{i}\in \fieldqn} \log r_1(\vect{a}\vect{i}^{T}) \log r_{2}( \vect{b}\vect{i}^{T})&=& 
  q^{N-2}\sum_{j\in\fieldq} \log r_1(j) \sum_{k\in \fieldq} \log r_2(k) \nonumber \\
  &=&q^{N-2}\left(\sum_{j\in\fieldq} \log r_1(j)\right) \left(\sum_{j\in \fieldq}\log r_2(j)\right)\nonumber \textrm{.}
\end{eqnarray}
Inserting this result into (\ref{p1p2def}) 
yields  
\begin{eqnarray}
\innerproduct{p_1(\vect{x})}{p_2(\vect{x})}&=&q^{N-2}\sum_{j\in\fieldq} \log r_1(j)\sum_{j\in \fieldq}\log r_2(j)
-\frac{1}{q^{N}} \sum_{\vect{i}\in \fieldqn} \log r_{1}(\vect{a}\vect{i}^T)
\sum_{\vect{i}\in \fieldqn} \log r_{2}(\vect{b}\vect{i}^T) \nonumber\\ 
&=&q^{N-2}\sum_{j\in\fieldq} \log r_1(j) \sum_{j\in \fieldq}\log r_2(j)-\nonumber \\ &&
\frac{1}{q^{N}}\left(\sum_{j\in\fieldq}\log r_1(j) \sum_{\vect{i}:\vect{a}\vect{i}^{T}=j}1\right) 
\left(\sum_{j\in \fieldq}\log r_2(j)\sum_{\vect{i}:\vect{b}\vect{i}^{T}=j}1\right) \nonumber \\
&=&q^{N-2}\sum_{j\in\fieldq} \log r_1(j) \sum_{j\in \fieldq}\log r_2(j)-
q^{N-2}\sum_{j\in\fieldq} \log r_1(j) \sum_{j\in \fieldq}\log r_2(j) \nonumber \\
&=&0 \nonumber \textrm{,}
\end{eqnarray}
which completes the proof of the second part. 
\subsection{Proof of Lemma \ref{imagesubspacelemma}\label{imagesubspacelemmaproof}}
First we are going to prove that $\imagesop{a}$ is a subspace of $\setfieldqn$ by showing that 
$\sopsub{a}{.}$ is a linear mapping from $\setfieldq$ to $\setfieldqn$. For any $p(x),r(x)\in\setfieldq$
and $\alpha,\beta\in\rfield$
\begin{eqnarray}
  \sopsub{\vect{a}}{\alpha\boxtimes p(x) \boxplus \beta \boxtimes r(x)} &=&
  \n{\fieldqn}{ \n{\fieldq}{(p(x))^{\alpha}(r(x))^{\beta}}\evalat{x=\vect{a}\vect{x}^{T}}} \nonumber \textrm{.}
\end{eqnarray}
The inner normalization operator above can be cancelled since there is another normalization outside. 
\begin{eqnarray}
\sopsub{\vect{a}}{\alpha\boxtimes p(x) \boxplus \beta \boxtimes r(x)} &=&
\n{\fieldqn}{ \left({(p(\vect{a}\vect{x}^{T}))^{\alpha}(r(\vect{a}\vect{x}^{T}))^{\beta}}\right)} \nonumber \textrm{.}
\end{eqnarray}
Using the definition of addition and scalar multiplication on $\setfieldqn$ 
we obtain
\begin{eqnarray}
 \sopsub{\vect{a}}{\alpha\boxtimes p(x) \boxplus \beta \boxtimes r(x)}
 &=&\alpha \boxtimes \n{\fieldqn}{p(\vect{a}\vect{x}^{T})} \boxplus \beta \boxtimes \n{\fieldqn}{r(\vect{a}\vect{x}^{T})} \nonumber \\
  &=&\alpha \boxtimes \sopsub{\vect{a}}{p(x)}\boxplus \beta \boxtimes \sopsub{\vect{a}}{r(x)}   \nonumber  \textrm{,}
\end{eqnarray}
which proves that $\sopsub{\vect{a}}{.}$ is a linear mapping. Since the image of any linear mapping is a subspace 
of the co-domain, $\imagesop{\vect{a}}$ is a subspace of $\setfieldqn$. 

Obviously, $\sopsub{\vect{a}}{.}$ is an injective mapping for nonzero $\vect{a}$. Therefore,
\begin{eqnarray}
  \dim \imagesop{\vect{a}} &=& \dim \setfieldq \nonumber \\
  &=&q-1 \textrm{,} \nonumber 
\end{eqnarray}  
which completes the proof. 
\subsection{Proof of Lemma \ref{subspacerelationlemma} \label{subspacerelationlemmaproof}}
If there exist an $\alpha \in \fieldq$ such that 
$\vect{b}=\alpha \vect{a}$ then for any $p(x) \in \setfieldq$
\begin{eqnarray}
  \sopsub{\vect{b}}{p(x)} &=&\n{\fieldqn}{p(\vect{b}\vect{x}^{T})} \nonumber \\
  &=&\n{\fieldqn}{p(\alpha\vect{a}\vect{x}^{T})} \nonumber \\
  &=&\sopsub{\vect{a}}{p(\alpha x)} \nonumber \textrm{.}
\end{eqnarray}
Since $p(x)$ is in $\setfieldq$, $p(\alpha x)$ is also in $\setfieldq$. Therefore, 
$\imagesop{\vect{b}}\subset \imagesop{\vect{a}}$. Similarly, it can be shown that 
$\imagesop{\vect{a}}\subset \imagesop{\vect{b}}$. Consequently,
\begin{equation}
  \imagesop{\vect{a}}=\imagesop{\vect{b}} \nonumber
\end{equation}
if  $\vect{b}$ is  equal to $\alpha\vect{a}$ for an $\alpha \in \fieldq$.

If there is not any $\alpha$ such that $\vect{b}=\alpha \vect{a}$ then
for any $p_1(\vect{x})\in \imagesop{\vect{a}}$ and $p_2(\vect{x}) \in \imagesop{\vect{b}}$
\begin{equation}
  \innerproduct{p_1(\vect{x})}{p_2(\vect{x})} = 0 \nonumber 
\end{equation} 
due to Lemma \ref{spcinnerproduct}. Hence, 
\begin{equation}
  \imagesop{\vect{a}} \perp \imagesop{\vect{b}} \nonumber  
\end{equation}
if there is not any $\alpha \in \fieldq$ such that $\vect{b}=\alpha \vect{a}$.

\section{Proofs and Derivations in Chapter 4 }
\subsection{Proof of Lemma \ref{propertyorthogonallemma}\label{propertyorthogonallemmaproof}}
We need to show that $p(\vect{x})$ is orthogonal to  $\n{\fieldqn}{r(\vect{a}\vect{x}^T)}$ for any
$r(x) \in \setfieldq$. 
\begin{eqnarray}
  \innerproduct{p(\vect{x})}{\n{\fieldqn}{r(\vect{a}\vect{x}^T)}} &=& 
  \sum_{\vect{i}\in \fieldq}\log p(\vect{i}\vect{D}) \log \frac{r(\vect{a}\vect{i}^{T})}{q^{N-1}}-
  \frac{1}{q^{N}}\sum_{\vect{i}\in \fieldq}\log p(\vect{i}\vect{D}) \sum_{\vect{i}\in \fieldq}\log  \frac{r(\vect{a}\vect{i}^{T})}{q^{N-1}} \nonumber\\
&=& 
  \sum_{\vect{i}\in \fieldq}\log p(\vect{i}\vect{D}) \log r(\vect{a}\vect{i}^{T})-
  \frac{1}{q^{N}}\sum_{\vect{i}\in \fieldq}\log p(\vect{i}\vect{D}) \sum_{\vect{i}\in \fieldq}\log r(\vect{a}\vect{i}^{T}) \label{proport1} \nonumber 
\end{eqnarray}

Let $\vect{j}$ be a vector in $\fieldqn$. For each $\vect{j}$ vector there are
$q^{N-\rank{\vect{D}}}$ $\vect{i}$ vectors in $\fieldqn$  satisfying the relation  
\begin{equation}
\vect{j}\vect{D}=\vect{i}\vect{D} \textrm{,}
\end{equation}
where $\rank{\vect{D}}$ denotes the 
rank of the dependency matrix $\vect{D}$. 
Therefore, the  first summation in (\ref{proport1}) is equal to the following nested summation.  
\begin{eqnarray}
\sum_{\vect{i}\in \fieldq}\log p(\vect{i}\vect{D}) \log r(\vect{a}\vect{i}^{T})&=&
\frac{1}{q^{N-\rank{\vect{D}}}}\sum_{\vect{i}\in \fieldqn}\left(\sum_{\vect{j}\in \fieldqn :\vect{i}\vect{D}=\vect{j}\vect{D}} 
\log p(\vect{j}\vect{D}) \log r(\vect{a}\vect{j}^{T})\right) \nonumber \\
&=&\frac{1}{q^{N-\rank{\vect{D}}}}\sum_{\vect{i}\in \fieldqn}
\log p(\vect{i}\vect{D})\left(\sum_{\vect{j}\in \fieldqn :\vect{i}\vect{D}=\vect{j}\vect{D}} 
 \log r(\vect{a}\vect{j}^{T})\right) \nonumber 
\end{eqnarray}
The inner summation on the right hand side above can be grouped as 
\begin{eqnarray}
\sum_{\vect{i}\in \fieldq}\log p(\vect{i}\vect{D}) \log r(\vect{a}\vect{i}^{T})&=&
\frac{1}{q^{N-\rank{\vect{D}}}}\sum_{\vect{i}\in \fieldqn}
\log p(\vect{i}\vect{D})\left(\sum_{k\in\fieldq}\left(\sum_{\vect{j}\in \fieldqn :\vect{i}\vect{D}=\vect{j}\vect{D} \land \vect{a}\vect{j}^{T}=k} 
 \log r(\vect{a}\vect{j}^{T})\right)\right)
 \nonumber\\
&=&\frac{1}{q^{N-\rank{\vect{D}}}}\sum_{\vect{i}\in \fieldqn}
\log p(\vect{i}\vect{D})\left(\sum_{k\in\fieldq}\log r(k)\left(\sum_{\vect{j}\in \fieldqn :\vect{i}\vect{D}=\vect{j}\vect{D} \land \vect{a}\vect{j}^{T}=k} 
 1\right)\right) \nonumber \textrm{.}
\end{eqnarray}
We have to determine how many times the innermost summation above runs. Let $\vect{d}_1$, $\vect{d}_2$, $\ldots$, $\vect{d}_{\rank{\vect{D}}}$
be the nonzero rows of $\vect{D}$. Then the innermost summation above runs once for all $\vect{j}$ vector satisfying
the system of linear equations below. 
\begin{equation}
 \left[ \begin{array}{c} \vect{d}_1 \\ \vect{d}_2 \\ \vdots \\ \vect{d}_{\rank{\vect{D}}} \\\vect{a}\end{array}\right]\vect{j}^{T} 
=\left[ \begin{array}{c}\vect{d}_1 \vect{i}^{T}\\\vect{d}_2 \vect{i}^{T} \\\vdots \\\vect{d}_{\rank{\vect{D}}} \vect{i}^{T} \\i  \end{array}\right]
\nonumber 
\end{equation} 
Due to the definition of the dependency matrix, all nonzero rows of $\vect{D}$ are linearly independent. Moreover, 
all these nonzero rows of $\vect{D}$ are also linearly independent with $\vect{a}$, since $\vect{a}$ is not
equal to $\vect{a}\vect{D}$. Therefore, the system of linear equations above has $q^{N-\rank{D}-1}$ solutions.     
Hence,
\begin{equation}
\sum_{\vect{i}\in \fieldq}\log p(\vect{i}\vect{D}) \log r(\vect{a}\vect{i}^{T})=\frac{1}{q}\sum_{\vect{i}\in \fieldqn}
\log p(\vect{i}\vect{D})\sum_{k\in\fieldq}\log r(k) \nonumber.
\end{equation}
Inserting this result into (\ref{proport1}) yields
\begin{eqnarray}
  \innerproduct{p(\vect{x})}{\n{\fieldqn}{r(\vect{a}\vect{x}^T)}}&=&
\frac{1}{q}\sum_{\vect{i}\in \fieldqn}
\log p(\vect{i}\vect{D})\sum_{k\in\fieldq}\log r(k)-
\frac{1}{q^{N}}\sum_{\vect{i}\in \fieldq}\log p(\vect{i}\vect{D}) \sum_{\vect{i}\in \fieldq}\log r(\vect{a}\vect{i}^{T})\nonumber\\
&=&\frac{1}{q}\sum_{\vect{i}\in \fieldqn}
\log p(\vect{i}\vect{D})\sum_{k\in\fieldq}\log r(k)-
\frac{1}{q}\sum_{\vect{i}\in \fieldqn}
\log p(\vect{i}\vect{D})\sum_{k\in\fieldq}\log r(k) \nonumber \\
&=&0 \nonumber \textrm{,}
\end{eqnarray}
which completes the proof. 

\section{Proofs and Derivations in Chapter 6}
\subsection{The factorization of a posteriori probability of $\vect{X}$ given in Section \ref{misoqpsk}\label{facmisoqpsk} }
Expanding the absolute value in (\ref{appmisoqpsk}) yields
\begin{eqnarray}
  p(\vect{x})&=&\n{\fieldqn}{ \exp\left(-\frac{1}{2\sigma^2}{\abs{y-\sum_{i=1}^{N}h_{i}\psk{x_i}}^{2}}  \right)} \nonumber \\
  &=&\n{\fieldqn}{ \exp\left(-\frac{1}{2\sigma^2}\left(\abs{y}^{2}-2y\sum_{i=1}^{N}\Re{h_i^*\psk{x_i}^*}+ \abs{\sum_{i=1}^{N}h_i\psk{x_i}}^2 
\right)\right)} \nonumber \textrm{.}
\end{eqnarray}
Since $\abs{y}^2$ does not depend on  $\vect{x}$, it can be cancelled by the normalization operator which gives,
\begin{eqnarray}
p(\vect{x})&=&\n{\fieldqn}{ \exp\left(\frac{1}{2\sigma^2}\left(\sum_{i=1}^{N}2\Re{yh_i^*\psk{x_i}^*}-\left(\sum_{i=1}^{N}h_i\psk{x_i}\right)
    \left(\sum_{i=1}^{N}h_i^*\psk{x_i}^*\right)\right) \right)}\nonumber \\
&=&\n{\fieldqn}{ \exp\left(\sum_{i=1}^{N}\frac{2\Re{yh_i^*\psk{x_i}^*}-\abs{h_i\psk{x_i}}^2}{2\sigma^2}-
    \sum_{j=2}^{N}\sum_{i=1}^{j-1}\frac{2\Re{h_i\psk{x_i}h_j^*\psk{x_j}^*}}{2\sigma^2} \right)} \nonumber 
\end{eqnarray}
Since PSK is a constant amplitude modulation, $\abs{\psk{x_i}}$ is constant for all $x_i$. Consequently, $\abs{h_i\psk{x_i}}^2$ 
does not depend on $\vect{x}$.   Canceling $\abs{h_i\psk{x_i}}^2$  by the normalization operator yields the desired factorization. 
\begin{equation}
p(\vect{x})=\n{\fieldqn}{\prod_{i=1}^{N}\exp\left(\frac{2\Re{yh_i^*\psk{x_i}^*}}{2\sigma^2}\right) 
  \prod_{j=2}^{N}\prod_{i=1}^{j-1}  \exp\left(-\frac{2\Re{h_ih_j^*\psk{x_i}\psk{x_j}^*}}{2\sigma^2}\right) } 
\end{equation}
\subsection{Proof of Theorem \ref{naturalgraythm} \label{naturalgraythmproof}}
The necessary row operations are listed below. 
\begin{enumerate}
\item Add $1^{st}$ row to $(1+\frac{(j-2)(j-1)}{2})^{th}$ row for $j=3$ up to $N$.
\item For $i=4$ up to $N$, add $(\frac{(i-1)(i-2)}{2}+3)^{th}$ row to 
  \begin{enumerate}
    \item $(\frac{(i-1)(i-2)}{2}+1)^{th}$ row,
    \item $(\frac{(i-1)(i-2)}{2}+2)^{th}$ row,
    \item $(\frac{(i-1)(i-2)}{2}+j)^{th}$ row for $j=4$ up to $i-1$,
    \item $(\frac{(i)(i-1)}{2}+3)^{th}$ row,
    \item $(\frac{(j-1)(j-2)}{2}+1+i)^{th}$ row for $j=i+2$ up to $N$.  
  \end{enumerate}

\end{enumerate}

\subsection{Derivation of the factorization in (\ref{mimoappfact1}) \label{mimoappfact}}
\begin{eqnarray}
\Pr\{\vect{X}=\vect{x}|\vect{Y}=\vect{y}\}&=& \n{\fieldnn{q}{2N_t}}{\exp\left(- \frac{\norm{\vect{y}-\vect{H}_c\vect{w}}^2}{2\sigma^{2}}\right)}\nonumber\\
&\propto&\exp\left(-\frac{1}{2\sigma^{2}}\left(\norm{\vect{y}}^{2}-2\Re{\vect{w}^{H}\vect{H}_c^H\vect{y}}+
    \vect{w}^{H}\vect{H}_c^{H}\vect{H}_c\vect{w} \right) \right) \nonumber
\end{eqnarray}
We can cancel $\norm{y}^{2}$ since it is constant for all $\vect{x}$. Let 
$\vect{u}\triangleq \vect{H}_c^H\vect{y}$ and
$\vect{R}\triangleq\vect{H}_{c}^{H}\vect{H}_c$. Then the factorization becomes,
\begin{eqnarray}
  \Pr\{\vect{X}=\vect{x}|\vect{Y}=\vect{y}\}&\propto&\exp\left(-\frac{1}{2\sigma^{2}}\left(-2\Re{\vect{w}^{H}\vect{u}} +
\vect{w}^{H}\vect{R}\vect{w}  \right) \right) \nonumber \\
&\propto&\exp\left(\frac{1}{2\sigma^{2}}\left(2\sum_{k=1}^{N_t}\Re{\qpsk{\vect{x}_k}u_{k}^{*}} - 
\sum_{k=1}^{N_t}\sum_{l=1}^{N_t}\qpsk{\vect{x}_{k}}^{*}(\vect{R})_{k,l}\qpsk{\vect{x}_{l} } \right)\right) \nonumber \virgul
\end{eqnarray}
where $u_k$ is the $k^{th}$ component of $\vect{u}$  and $(\vect{R})_{k,l}$ is the entry in the $k^{th}$ row 
and $j^{th}$ column of the matrix $\vect{R}$. Since $\vect{R}$ is hermitian symmetric,
\begin{eqnarray}
  \Pr\{\vect{X}=\vect{x}|\vect{Y}=\vect{y}\}&\propto& \exp\left(\frac{1}{2\sigma^{2}}\left(
      \sum_{k=1}^{N_t}\left(2\Re{\qpsk{\vect{x}_k}u_{k}^{*}}-\norm{\qpsk{\vect{x}_k}}^{2}(\vect{R})_{k,k} \right)\right)\right) \nonumber \\
      &&\cdot\exp\left(-\frac{1}{2\sigma^{2}}\left( \sum_{k=2}^{N_t}\sum_{l=1}^{k-1}2\Re{\qpsk{\vect{x}_{k}}^{*}
            (\vect{R})_{k,l}\qpsk{\vect{x}_{l} }} \right)\right) \nonumber \nokta
\end{eqnarray}
Since $\norm{\qpsk{\vect{x}_k}}^{2}$ is constant,
\begin{equation}
\Pr\{\vect{X}=\vect{x}|\vect{Y}=\vect{y}\} \propto
{
\exp\left(\frac{1}{2\sigma^{2}}\left(
      \sum_{k=1}^{N_t}2\Re{\qpsk{\vect{x}_k}u_{k}^{*}}-\sum_{k=2}^{N_t}\sum_{l=1}^{k-1}2\Re{\qpsk{\vect{x}_{k}}^{*}
            (\vect{R})_{k,l}\qpsk{\vect{x}_{l} }} \right)\right)
}  \nokta \label{intermediateqpskmimo1} 
\end{equation}

The function $\qpsk{\vect{x}_k}$ can be expressed in terms of $\bpsk{.}$ function as
\begin{equation}
\qpsk{\vect{x}_k}=a\bpsk{x_{2k-1}}+a^*\bpsk{x_{2k}} \nonumber \virgul
\end{equation}
where $a=\frac{1}{2}+j\frac{1}{2}$. Therefore,
\begin{equation}
  \Re{\qpsk{\vect{x}_k}u_{k}^{*}} = \Re{au_{k}^{*}}\bpsk{x_{2k-1}}+\Re{a^*u_{k}^{*}}\bpsk{x_{2k}} \nokta \label{intermediateqpskmimo2}
\end{equation}
Furthermore, 
\begin{eqnarray}
\Re{\qpsk{\vect{x}_{k}}^{*}(\vect{R})_{k,l}\qpsk{\vect{x}_{l} }}&=&
\Re{(a^*\bpsk{x_{2k-1}}+a\bpsk{x_{2k}})(\vect{R})_{k,l}(a\bpsk{x_{2l-1}}+a^*\bpsk{x_{2l}}) }\nonumber \\
&=&\Re{\abs{a}^{2}(\vect{R})_{k,l}\bpsk{x_{2k-1}}\bpsk{x_{2l-1}}}+
\Re{(a^*)^{2}(\vect{R})_{k,l}\bpsk{x_{2k-1}}\bpsk{x_{2l}}}\nonumber\\
&&+\Re{a^2(\vect{R})_{k,l}\bpsk{x_{2k}}\bpsk{x_{2l-1}}} +
\Re{\abs{a}^{2}(\vect{R})_{k,l}\bpsk{x_{2k}}\bpsk{x_{2l}}}\nonumber \\
&=&\frac{1}{2}\big(\bpsk{x_{2k-1}+x_{2l-1}}\Re{(\vect{R})_{k,l}}+\bpsk{x_{2k-1}+x_{2l}}\Im{(\vect{R})_{k,l}} \nonumber\\
&&-\bpsk{x_{2k}+x_{2l-1}}\Im{(\vect{R})_{k,l}} +\bpsk{x_{2k}+x_{2l}}\Re{(\vect{R})_{k,l}} \big) \label{intermediateqpskmimo3} \nokta
\end{eqnarray} 
Inserting  (\ref{intermediateqpskmimo2}) and (\ref{intermediateqpskmimo3}) together with the 
definition of the $\gamma(.;.)$ function into (\ref{intermediateqpskmimo1}) gives the
desired factorization. 
\begin{equation}
\begin{split}
\Pr\{\vect{X}=\vect{x}|\vect{Y}=\vect{y}\} \propto & \prod_{k=1}^{N_t}\gamma\left(x_{2k-1};\frac{\Re{u_{k}}+\Im{u_k}}{2},\sigma\right) 
\gamma\left(x_{2k};\frac{\Re{u_{k}}-\Im{u_k}}{2},\sigma\right) \\
&\cdot \prod_{k=2}^{N_t}\prod_{l=1}^{k-1}\gamma\left(x_{2k-1}+x_{2l-1};-\frac{\Re{(\vect{R})_{k,l}}}{2},\sigma\right)
\gamma\left(x_{2k-1}+x_{2l};-\frac{\Im{(\vect{R})_{k,l}}}{2},\sigma\right)\\
&\cdot \prod_{k=2}^{N_t}\prod_{l=1}^{k-1}\gamma\left(x_{2k}+x_{2l-1};\frac{\Im{(\vect{R})_{k,l}}}{2},\sigma\right)
\gamma\left(x_{2k}+x_{2l};-\frac{\Re{(\vect{R})_{k,l}}}{2},\sigma\right)
\end{split}\nokta
\end{equation}  
\subsection{Permutations used in the simulation in Section \ref{improvementsubsection}\label{mimopermutations}}
The set $\set{P}$ consists of the following permutations. 
\begin{eqnarray}
  \vect{P}_{1}&=& \left[ \vect{f}_{1}^{T}\  \vect{f}_{3}^{T}\  \vect{f}_{5}^{T}\  \vect{f}_{7}^{T}\  
\vect{f}_{9}^{T}\  \vect{f}_{11}^{T}\  \vect{f}_{13}^{T}\  \vect{f}_{15}^{T}\ 
\vect{f}_{2}^{T}\  \vect{f}_{4}^{T}\  \vect{f}_{6}^{T}\  \vect{f}_{8}^{T}\  
\vect{f}_{10}^{T}\  \vect{f}_{12}^{T}\  \vect{f}_{14}^{T}\  \vect{f}_{16}^{T}
\right] \nonumber \\
  \vect{P}_{2}&=& \left[ \vect{f}_{1}^{T}\  \vect{f}_{3}^{T}\  \vect{f}_{5}^{T}\  \vect{f}_{7}^{T}\  
\vect{f}_{9}^{T}\  \vect{f}_{11}^{T}\  \vect{f}_{13}^{T}\  \vect{f}_{15}^{T}\ 
  \vect{f}_{4}^{T}\  \vect{f}_{6}^{T}\  \vect{f}_{8}^{T}\  
\vect{f}_{10}^{T}\  \vect{f}_{12}^{T}\  \vect{f}_{14}^{T}\  \vect{f}_{16}^{T}\ \vect{f}_{2}^{T}
\right] \nonumber \\
  \vect{P}_{3}&=& \left[ \vect{f}_{1}^{T}\  \vect{f}_{3}^{T}\  \vect{f}_{5}^{T}\  \vect{f}_{7}^{T}\  
\vect{f}_{9}^{T}\  \vect{f}_{11}^{T}\  \vect{f}_{13}^{T}\  \vect{f}_{15}^{T}\ 
    \vect{f}_{6}^{T}\  \vect{f}_{8}^{T}\  
\vect{f}_{10}^{T}\  \vect{f}_{12}^{T}\  \vect{f}_{14}^{T}\  \vect{f}_{16}^{T}\ \vect{f}_{2}^{T}\ \vect{f}_{4}^{T}\
\right] \nonumber \\
  \vect{P}_{4}&=& \left[ \vect{f}_{1}^{T}\  \vect{f}_{3}^{T}\  \vect{f}_{5}^{T}\  \vect{f}_{7}^{T}\  
\vect{f}_{9}^{T}\  \vect{f}_{11}^{T}\  \vect{f}_{13}^{T}\  \vect{f}_{15}^{T}\ 
      \vect{f}_{8}^{T}\  
\vect{f}_{10}^{T}\  \vect{f}_{12}^{T}\  \vect{f}_{14}^{T}\  \vect{f}_{16}^{T}\ \vect{f}_{2}^{T}\ \vect{f}_{4}^{T}\ \vect{f}_{6}^{T}\
\right] \nonumber
\end{eqnarray}
\begin{eqnarray} 
  \vect{P}_{5}&=& \left[ \vect{f}_{1}^{T}\  \vect{f}_{3}^{T}\  \vect{f}_{5}^{T}\  \vect{f}_{7}^{T}\  
\vect{f}_{9}^{T}\  \vect{f}_{11}^{T}\  \vect{f}_{13}^{T}\  \vect{f}_{15}^{T}\ 
\vect{f}_{10}^{T}\  \vect{f}_{12}^{T}\  \vect{f}_{14}^{T}\  \vect{f}_{16}^{T}\ \vect{f}_{2}^{T}\ \vect{f}_{4}^{T}\ \vect{f}_{6}^{T}\
\vect{f}_{8}^{T}\right] \nonumber \\  
\vect{P}_{6}&=& \left[ \vect{f}_{1}^{T}\  \vect{f}_{3}^{T}\  \vect{f}_{5}^{T}\  \vect{f}_{7}^{T}\  
\vect{f}_{9}^{T}\  \vect{f}_{11}^{T}\  \vect{f}_{13}^{T}\  \vect{f}_{15}^{T}\ 
  \vect{f}_{12}^{T}\  \vect{f}_{14}^{T}\  \vect{f}_{16}^{T}\ \vect{f}_{2}^{T}\ \vect{f}_{4}^{T}\ \vect{f}_{6}^{T}\
\vect{f}_{8}^{T}\  \vect{f}_{10}^{T} \right] \nonumber \\
\vect{P}_{7}&=& \left[ \vect{f}_{1}^{T}\  \vect{f}_{3}^{T}\  \vect{f}_{5}^{T}\  \vect{f}_{7}^{T}\  
\vect{f}_{9}^{T}\  \vect{f}_{11}^{T}\  \vect{f}_{13}^{T}\  \vect{f}_{15}^{T}\ 
    \vect{f}_{14}^{T}\  \vect{f}_{16}^{T}\ \vect{f}_{2}^{T}\ \vect{f}_{4}^{T}\ \vect{f}_{6}^{T}\
\vect{f}_{8}^{T}\  \vect{f}_{10}^{T}\ \vect{f}_{12}^{T} \right] \nonumber \\
\vect{P}_{8}&=& \left[ \vect{f}_{1}^{T}\  \vect{f}_{3}^{T}\  \vect{f}_{5}^{T}\  \vect{f}_{7}^{T}\  
\vect{f}_{9}^{T}\  \vect{f}_{11}^{T}\  \vect{f}_{13}^{T}\  \vect{f}_{15}^{T}\ 
      \vect{f}_{16}^{T}\ \vect{f}_{2}^{T}\ \vect{f}_{4}^{T}\ \vect{f}_{6}^{T}\
\vect{f}_{8}^{T}\  \vect{f}_{10}^{T}\ \vect{f}_{12}^{T}\ \vect{f}_{14}^{T} \right] \nonumber 
\end{eqnarray}
\begin{eqnarray} 
  \vect{P}_{9}&=& \left[ 
  \vect{f}_{4}^{T}\  \vect{f}_{6}^{T}\  \vect{f}_{8}^{T}\  
\vect{f}_{10}^{T}\  \vect{f}_{12}^{T}\  \vect{f}_{14}^{T}\  \vect{f}_{16}^{T}\ \vect{f}_{2}^{T}\
\vect{f}_{1}^{T}\  \vect{f}_{3}^{T}\  \vect{f}_{5}^{T}\  \vect{f}_{7}^{T}\  
\vect{f}_{9}^{T}\  \vect{f}_{11}^{T}\  \vect{f}_{13}^{T}\  \vect{f}_{15}^{T}\ 
\right] \nonumber \\  
\vect{P}_{10}&=& \left[ 
  \vect{f}_{6}^{T}\  \vect{f}_{8}^{T}\  
\vect{f}_{10}^{T}\  \vect{f}_{12}^{T}\  \vect{f}_{14}^{T}\  \vect{f}_{16}^{T}\ \vect{f}_{2}^{T}\    \vect{f}_{4}^{T}\
\vect{f}_{1}^{T}\  \vect{f}_{3}^{T}\  \vect{f}_{5}^{T}\  \vect{f}_{7}^{T}\  
\vect{f}_{9}^{T}\  \vect{f}_{11}^{T}\  \vect{f}_{13}^{T}\  \vect{f}_{15}^{T}\ \right] \nonumber \\  
\vect{P}_{11}&=& \left[ 
    \vect{f}_{8}^{T}\  
\vect{f}_{10}^{T}\  \vect{f}_{12}^{T}\  \vect{f}_{14}^{T}\  \vect{f}_{16}^{T}\ \vect{f}_{2}^{T}\    \vect{f}_{4}^{T}\ \vect{f}_{6}^{T}\
\vect{f}_{1}^{T}\  \vect{f}_{3}^{T}\  \vect{f}_{5}^{T}\  \vect{f}_{7}^{T}\  
\vect{f}_{9}^{T}\  \vect{f}_{11}^{T}\  \vect{f}_{13}^{T}\  \vect{f}_{15}^{T}\ \right] \nonumber \\  
\vect{P}_{12}&=&   \left[     
\vect{f}_{10}^{T}\  \vect{f}_{12}^{T}\  \vect{f}_{14}^{T}\  \vect{f}_{16}^{T}\ \vect{f}_{2}^{T}\    \vect{f}_{4}^{T}\ \vect{f}_{6}^{T}\ \vect{f}_{8}^{T}\
\vect{f}_{1}^{T}\  \vect{f}_{3}^{T}\  \vect{f}_{5}^{T}\  \vect{f}_{7}^{T}\  
\vect{f}_{9}^{T}\  \vect{f}_{11}^{T}\  \vect{f}_{13}^{T}\  \vect{f}_{15}^{T}\ \right] \nonumber 
\end{eqnarray}
\begin{eqnarray} 
  \vect{P}_{13}&=& \left[ 
 \vect{f}_{12}^{T}\  \vect{f}_{14}^{T}\  \vect{f}_{16}^{T}\ \vect{f}_{2}^{T}\
  \vect{f}_{4}^{T}\  \vect{f}_{6}^{T}\  \vect{f}_{8}^{T}\  \vect{f}_{10}^{T}\ 
\vect{f}_{1}^{T}\  \vect{f}_{3}^{T}\  \vect{f}_{5}^{T}\  \vect{f}_{7}^{T}\  
\vect{f}_{9}^{T}\  \vect{f}_{11}^{T}\  \vect{f}_{13}^{T}\  \vect{f}_{15}^{T}\ 
\right] \nonumber \\  
\vect{P}_{14}&=& \left[ 
  \vect{f}_{14}^{T}\  \vect{f}_{16}^{T}\ \vect{f}_{2}^{T}\    \vect{f}_{4}^{T}\
  \vect{f}_{6}^{T}\  \vect{f}_{8}^{T}\  \vect{f}_{10}^{T}\  \vect{f}_{12}^{T}\
\vect{f}_{1}^{T}\  \vect{f}_{3}^{T}\  \vect{f}_{5}^{T}\  \vect{f}_{7}^{T}\  
\vect{f}_{9}^{T}\  \vect{f}_{11}^{T}\  \vect{f}_{13}^{T}\  \vect{f}_{15}^{T}\ \right] \nonumber \\  
\vect{P}_{15}&=& \left[ 
  \vect{f}_{16}^{T}\ \vect{f}_{2}^{T}\    \vect{f}_{4}^{T}\ \vect{f}_{6}^{T}\
    \vect{f}_{8}^{T}\  \vect{f}_{10}^{T}\  \vect{f}_{12}^{T}\  \vect{f}_{14}^{T}\
\vect{f}_{1}^{T}\  \vect{f}_{3}^{T}\  \vect{f}_{5}^{T}\  \vect{f}_{7}^{T}\  
\vect{f}_{9}^{T}\  \vect{f}_{11}^{T}\  \vect{f}_{13}^{T}\  \vect{f}_{15}^{T}\ \right] \nonumber
\end{eqnarray}


\section{Proofs and Derivations in Chapter 7}
\subsection{Proof of Theorem \ref{conditionalindependencetheorem}\label{conditionalindependencethmproof}}

The proof in the forward direction is actually an implication of the cut-set independence theorem stated in \cite{aloefg}.
An alternative proof is given below. 

Let $t_{k}(\vect{x})$ and $t_{l}(\vect{x})$ be defined as
\begin{eqnarray}
  t_{k}(\vect{x})& \triangleq &\n{\fieldqn}{\prod_{\vect{a}_{i}\in \set{K}_{k}}r_{i}(\vect{a}_{i}\vect{x}^{T})}  \virgul\\
  t_{l}(\vect{x})& \triangleq &\n{\fieldqn}{\prod_{\vect{a}_{i}\in \set{K}_{l}\setminus \set{K}_k}r_{i}(\vect{a}_{i}\vect{x}^{T})} \nokta
\end{eqnarray}
Clearly,
\begin{equation}
  p(\vect{x})=\n{\fieldqn}{t_{k}(\vect{x})t_{l}(\vect{x})} \nokta \label{prrr}
\end{equation}
Due to the definitions of $\set{K}_{k}$ and $\set{K}_{k}$, $t_{k}(\vect{x})$ and $t_{l}(\vect{x})$
satisfies
\begin{eqnarray}
  r_{k}(\vect{x}) &\triangleq& r_k(\vect{x}(\vect{I}-\vect{E}_k)) \virgul \label{rklocality} \\
  r_{l}(\vect{x}) &\triangleq& r_l(\vect{x}(\vect{I}-\vect{E}_l)) \virgul 
\end{eqnarray}  
where $\vect{E}_k$ ($\vect{E}_l$) is the dependency with just a single $1$
on its $k^{th}$ ($l^{th}$) entry on the main diagonal.

An equivalent requirement on conditional independence can be obtained 
by multiplying both sides of (\ref{conditionalindependence1}) with
 $\Pr\{\vectex{X}{k,l}=\vectex{x}{k,l}\} \Pr\{\vectex{X}{k}=\vectex{x}{l}\}$
as follows.
\begin{equation}
  p(\vect{x})\Pr\{\vectex{X}{k,l}=\vectex{x}{k,l}\}=
  \Pr\{\vectex{X}{k}=\vectex{x}{k}\}
  \Pr\{\vectex{X}{l}=\vectex{x}{l}\}
\label{conditionalindependence2} 
\end{equation}

The marginal distribution $\Pr\{\vectex{X}{k}=\vectex{x}{k}\}$ can be derived in terms 
of $t_{k}(\vect{x})$ and $t_{l}(\vect{x})$ as 
\begin{eqnarray}
  \Pr\{\vectex{X}{k}=\vectex{x}{k}\}&=&\sum_{\forall x_{k} \in \fieldq} p(\vect{x}) \nonumber \\
  &=&\n{\fieldq^{n-1}}{\sum_{\forall x_{k} \in \fieldq} r_{k}(\vect{x})r_{l}(\vect{x})} \nonumber \\
   &=&\n{\fieldq^{n-1}}{r_{k}(\vect{x})\sum_{\forall x_{k} \in \fieldq}r_{l}(\vect{x})} \label{mar1} \virgul 
\end{eqnarray}
where the last line follows from (\ref{rklocality}). Other marginal distributions in (\ref{conditionalindependence2}) can similarly be derived 
as
\begin{equation}
\Pr\{\vectex{X}{l}=\vectex{x}{l}\}= \n{\fieldq^{n-1}}{r_{l}(\vect{x})\sum_{\forall x_{l} \in \fieldq}r_{k}(\vect{x})} \label{mar2}
\end{equation}
\begin{equation}
\Pr\{\vectex{X}{k,l}=\vectex{x}{k,l}\}=\n{\fieldq^{n-2}}
{\sum_{\forall x_{k}\in \fieldq}r_{l}(\vect{x})\sum_{\forall x_{l} \in \fieldq}r_{k}(\vect{x})}
\label{mar3}
\end{equation}
Inserting (\ref{prrr}), (\ref{mar1}), (\ref{mar2}), and (\ref{mar3}) into (\ref{conditionalindependence2})
 verifies that the equality in (\ref{conditionalindependence2}) holds and completes the proof in the forward direction. 

The proof in the backward direction starts with multiplying both sides of (\ref{conditionalindependence1}) with
$\Pr\{\vectex{X}{k}=\vectex{x}{k}\}$ which yields
\begin{eqnarray}
p(\vect{x})&=&\Pr\{\vectex{X}{k}=\vectex{x}{k}\}\Pr\{X_{k}=x_{k}|\vectex{X}{k,l}=\vectex{k,l}\} \nonumber \\
&=&\n{\fieldqn}{m_{k}(\vect{x})m_{l}(\vect{x})} \virgul
\end{eqnarray}
where $m_{k}(\vect{x})$ and $m_{l}(\vect{x})$ are $\n{\fieldqn}{\Pr\{\vectex{X}{k}=\vectex{x}{k}\}}$ and   
$\Pr\{X_{k}=x_{k}|\vectex{X}{k,l}=\vectex{k,l}\}$. Clearly, these functions satisfy 
\begin{eqnarray}
m_{k}(\vect{x})&=&m_{k}(\vect{x}(\vect{I}-\vect{E}_{k})) \virgul \\
m_{l}(\vect{x})&=&m_{l}(\vect{x}(\vect{I}-\vect{E}_{l})) \nokta
\end{eqnarray}
Then due to Theorem \ref{localfactorizationtheorem} $p(\vect{x})$ can be factored as in (\ref{yoruldum_artik}).

\chapter*{VITA}
\section*{Personal Information}
\begin{tabular}{rl}
    \textsc{Date and Place of Birth:} & June 12, 1980 | \.{I}stanbul, Turkey    \\
\    \textsc{Nationality:} & Turkish \\
    \textsc{Marital Status:} & Married \\
    \textsc{Address:}   & Elektrik Elektronik M\"{u}h. B\"{o}l\"{u}m\"{u}, ODT\"{U}, 06531, Ankara/Turkey   \\
    \textsc{Phone:}     & +90 505 514 61 14 \\ 
 \textsc{Web Page:}     & {www.eee.metu.edu.tr/\textasciitilde fatih}\\
       \textsc{e-mail:}       & {fatih@eee.metu.edu.tr}\\
    \textsc{}     & {bfatih@gmail.com}
\end{tabular}
\section*{Education}
\begin{tabular}{rp{12cm}}
Sep. 2005 &  \emph{M.Sc. in Electrical and Electronics Engineering} \\
 &Middle East Technical University, Ankara, Turkey \\
&Thesis Title: \emph{``Sub-Graph Approach in Iterative Sum-Product Algorithm''} \\
&Advisor: Prof. Dr. Buyurman Baykal\\
&Co-advisor: Assoc. Prof. Dr. Ali \"Ozg\"{u}r Y\i lmaz \\
& \\
June 2002 & \emph{B.Sc. in Electrical and Electronics Engineering} \\
 &Middle East Technical University, Ankara, Turkey \\
\end{tabular}

\section*{Teaching Experience}
\begin{tabular}{r|p{11cm}}
\textsc{Sep. 2002 -}&\emph{Teaching Assistant} \\ \textsc{Sep. 2009}& {Department of Electrical and Electronics 
Engineering}
\\&{Middle East Technical University}
\\&\footnotesize{Signal processing courses}
\\& \footnotesize{Non-linear electronics for communications course}
\\& \footnotesize{Analog electronics course}
\\& \footnotesize{Electrical circuits laboratory course}
\\\multicolumn{2}{c}{}
 \\ \textsc{Dec. 2007} & \emph{Instructor of a short course on signal processing}
 \\ & Environmental Tectonics Corporation, Turkey Branch
\\\multicolumn{2}{c}{}
\\\textsc{Oct. 2005} & \emph{Establishment of the electrical circuits laboratory}   
\\ & Northern Cyprus Campus of Middle East Technical University 

\end{tabular}

\section*{Engineering Experience}
\begin{tabular}{r|p{11cm}}
\textsc{Oct. 2001-} & \emph{DSP Programmer}\\
\textsc{May 2002} & Implementation of NATO STANAG 4285 (HF modem standard) \\ 
& on TMS320C54 DSP processor\\
\multicolumn{2}{c}{}\\
\textsc{Sep. 2007-}& \emph{Researcher} \\
\textsc{March 2009}& Detecting and classifying low probability of intercept radar project\\
\end{tabular}

\section*{Publications}
\begin{itemize}
\item Bayramo\~glu M. F. and Y\i lmaz A. \"{O}., \emph{``An analysis method of multivariate probability mass functions''}
, to be submitted 
\item Bayramo\~glu M. F. and Y\i lmaz A. \"{O}., \emph{``Klasik tespit kuram\i{}  i\c{c}in \"{O}klid 
geometrisi ile genel bir g\"{o}sterim''}, {IEEE 18. Sinyal \.{I}\c{s}leme ve \.{I}leti\c{s}im Uygulamalar\i{} 
Kurultay\i{}} (S\.IU),  April 2010, Diyarbak\i r, Turkey, An English version is available on arXiv. 
\item Bayramo\~glu M. F. and Y\i lmaz A. \"{O}., \emph{``Factorization of joint probability mass functions 
  into parity check interactions''}, Int. Symposium on Information Theory, June 2009, Seoul, Korea
\item Bayramo\~glu M. F. and  Y\i lmaz A. \"{O}., \emph{``A Hilbert space of probability mass 
functions and applications on the sum-product algorithm''}, Int. Symposium on Turbo Codes,  September
2008, Lausanne, Switzerland 
\item Bayramo\~glu M. F.,  Y\i lmaz A. \"{O}.,  and Baykal B.,  \emph{``Sub-graph approach 
in iterative sum-product algorithm''}, Int. Symposium on Turbo Codes,  April
2006, Munich, Germany 
\end{itemize}
\section*{Honors and Awards}
\begin{tabular}{r|l}
1998 & Ranked  27\textsuperscript{th} in \\
& {Nationwide university admission examination among 1.3 million exam takers}\\
\multicolumn{2}{c}{}\\
1997 & Silver medalist in \\
& National Olympics in Informatics  
\end{tabular}


\section*{Computer Skills}
C++, Matlab, Linux, Agilent VEE,  Texas Instruments TMS320C54 Assembly language, 
  {\LaTeX}

\section*{Interests and Activities}
Road cycling, archery, history
\definecolor{refColor}{rgb}{0.8,0.0,0.3}

\end{document}